\documentclass[11pt,a4paper]{article}
\usepackage[utf8]{inputenc}

\usepackage[hang,flushmargin]{footmisc}
\usepackage[toc,page]{appendix}
\ifx\href\undefined\else\hypersetup{linktocpage=true}\fi

\usepackage[english]{babel}
\usepackage{graphicx}
\usepackage{amsmath}
\usepackage{amssymb}
\usepackage{amsthm}
\usepackage{times}
\usepackage{euler}
\usepackage{natbib}
\usepackage{hyperref}
\usepackage{tikz-cd}
\usepackage{tensor}
\usepackage{rotating}
\usepackage[font=small,labelfont=bf]{caption}

\newcommand{\kernel}{\mathrm{Ker}}
\newcommand{\Span}{\mathrm{Span}}
\newcommand{\image}{\mathrm{Im}}
\newcommand{\proj}{\mathrm{P}}
\renewcommand{\vec}[1]{\boldsymbol{\mathbf{#1}}} 
\newcommand{\mat}[1]{\underline{\boldsymbol{\mathbf{#1}}}}

\newtheorem{theorem}{Theorem}
\newtheorem{lemma}[theorem]{Lemma}
\newtheorem{proposition}[theorem]{Proposition}
\newtheorem{definition}[theorem]{Definition}
\newtheorem{corollary}[theorem]{Corollary}
\newtheorem{exercise}{Exercise}
\theoremstyle{remark}
\newtheorem{remark}[theorem]{Remark}

\newcommand{\End}{\mathrm{End}}

\newcommand{\reals}{\mathbb{R}}
\newcommand{\complex}{\mathbb{C}}
\newcommand{\Section}{\Gamma\,}
\newcommand{\id}{\mathrm{id}}

\newcommand{\Lor}{\mathrm{Lor}}
\newcommand{\ILor}{\mathrm{ILor}}
\newcommand{\Gal}{\mathrm{Gal}}
\newcommand{\IGal}{\mathrm{IGal}}
\newcommand{\group}[1]{\mathrm{#1}}
\newcommand{\som}{\mathcal{S}}
\newcommand{\vplus}{\oplus}

\newcommand{\U}{\mathcal{U}}

\begin{document}

\title{Reconsidering Velocity Addition/Subtraction\\
in Special Relativity}

\author{Domenico Giulini\\
Institute for Theoretical Physics\\
Leibniz University of Hannover, Germany\\
\emph{Email: giulini@itp.uni-hannover.de}\\
\emph{ORCID: 0000-0003-3123-7257}
}
\date{}
\maketitle
\begin{abstract}
\noindent
We reconsider velocity addition/subtraction 
in Special Relativity (SR) and re-derive 
its well-known non-commutative and non-associative 
algebraic properties in a self-contained way, 
including various explicit expressions for the 
Thomas angle, the derivation of which will be 
seen to be not as challenging as often suggested. 
All this is based on the polar-decomposition 
theorem in the traditional component language, 
in which Lorentz transformations are ordinary 
matrices.
  In the second part of this paper we offer a
less familiar alternative geometric view, 
that leads to an invariant definition of the 
concept of relative velocity between two states 
of motion, which is based on the boost-link-theorem,
of which we also offer an elementary proof that does 
not seem to be widely known in the relativity literature.  
  Finally we compare this to the corresponding 
geometric definitions in Galilei-Newton spacetime,
emphasising similarities and differences. 
Regarding the presentation of the material we will 
pursue an uncompromising pedagogical strategy, 
willingly accepting repetitions and occasional 
redundancies if deemed beneficial for clarity and 
the avoidance of anticipated misunderstandings.
An appendix with four sections includes some
mathematical details on results needed in the 
main text, as well some recollections on notions 
like semi-direct products of groups and affine 
spaces. 
\end{abstract}

\emph{Mathematics Subject Classification (MSC):} 
83A05, 
20N05, 
57K32  
\\
\emph{Physics and Astronomy Classification Scheme(PACS)}:
03.30.+p

\newpage
\setcounter{tocdepth}{2}
{\small \tableofcontents}
\newpage

\section{%
Introduction 
\label{sec:Introduction}
}
This paper attempts to give a comprehensive 
account of the algebraic and geometric aspects
connected with the composition and linking properties 
of boost transformations in Special Relativity 
(henceforth abbreviated by SR). 
It contains many known results\footnote{See \citet{Gourgoulhon:SR} for a comprehensive 
modern account.} but also offers new aspects. 
Some of the proofs of the known results appear 
to be new and - hopefully - easier than those 
in the published literature. The style of our 
presentation is largely shaped by an educational 
intent and clearly reveals that the following 
remarks originate from a lecture 
manuscript.\footnote{Lectures on 
Special Relativity, delivered at Leibniz 
University Hannover at irregular intervals 
since 2009.}

In classical mechanics we are used to the fact 
that the relative velocity between two 
states of motion is independent of the 
third state of motion (of the ``observer'') 
from which it is reckoned. In Special Relativity
this ceases to be true in a twofold way: 
regarding its measure and regarding the space 
of which it is an element of. Nevertheless, 
all assignments are naturally fully covariant and 
hence consistent with the principle of relativity.

All of these aspects will be clarified in 
this paper, whose main result is stated and 
fully proved in 
Section\,\ref{sec:RelativeVelocities}. It may be 
summarised as follows: Let $\som$ denote the 
set of \emph{states of motion}, which is a 
maximally symmetric 3-dimensional Riemannian 
manifold of constant negative curvature 
(given by the set of unit timelike future 
pointing vectors). For any ordered pair 
$(s_1,s_2)$ of elements in $\som$ the 
\emph{relative velocity} between them is 
a section 
$s\mapsto\vec\beta(s,s_1,s_2)\in T\som$
in the tangent bundle over $\som$. 
The expression $\vec\beta(s,s_1,s_2)$ is 
a rational function of its arguments (formula
\eqref{eq:BoostLink-7a}) which is 
equivariant under the action of the 
Lorentz group, in the sense that for any 
Lorentz transformation $L$ it obeys   
$\vec\beta(Ls,Ls_1,Ls_2)=L\vec\beta(s,s_1,s_2)$. 
Its geometric interpretation is that of 
the velocity $\vec\beta$ (in units of $c$) of 
the unique boost $B(\vec\beta)$ relative to $s$ 
that links $s_1$ with $s_2$. The uniqueness 
statement here is known as \emph{boost-link-theorem}, 
of which we give an elementary and constructive 
proof. The analytic expression of the corresponding 
boost map (formula \eqref{eq:BoostLink-11})
is likewise a rational function of 
$(s,s_1,s_2)$.

The organisation of the paper is as follows.
Section\,\ref{sec:OldStory} we review in a 
self-contained fashion the full story concerning 
Einstein's law of velocity addition/sub\-traction 
in component form. Polar decomposition is used to 
decompose elements in the Lorentz group into a 
boost and a rotation. It is emphasised that this 
decomposition is not natural, but rather 
dependent of a distinguished state of motion. 
Explicit formulae for the polar decomposition 
of the composition of two boosts are given to 
derive once more Einstein addition and also 
Thomas Rotation, the properties of which are 
derived. In particular, uniqueness of polar 
decomposition is used to prove that Einstein 
addition endows the open unit ball in $\reals^3$ 
with the algebraic structure of a non-associative 
quasigroup with identity, called a \emph{loop}. 

Section\,\ref{sec:RelativeVelocities} 
replaces the common matrix notation by a 
proper geometric setting in which the 
analytic expressions merely contain states 
of motion and their scalar products. 
The boost-link-theorem is proved and taken 
as basis for the definition of ``relative 
velocity'', namely as the velocity of the 
boost that links the two given states $(s_1,s_2)$. 
We will call it the ``link velocity''. Next to 
the two states it also depends on the state $s$ 
relative to which we identify the linking 
Lorentz transformation as a pure boost. 
Since the group generated by pure boosts is 
the entire Lorentz group, that reference to 
$s$ is necessary. We emphasise that this 
latter reference renders the notion of relative
velocity a ternary (rather than binary) relation
and that this fact is \emph{not} in conflict with
the relativity principle, as sometimes claimed.

Section\,\ref{sec:GalileiNewton} compares the 
findings of Section\,3 to the case of Galilei-Newton 
spacetime, which we again characterise in proper 
geometric terms. Boosts now form indeed a 
3-dimensional subgroup, which is even abelian 
and normal. Hence the relative velocity between 
an ordered set $(s_1,s_2)$ of states does 
\emph{not} require the specification of a 
third reference state relative to which the 
linking Galilei transformation is identified 
as pure boost.  

Finally, an Appendix of four parts contains 
the statement and proof of the polar-decomposition 
theorem, a proof of the statement (used in the 
main text) that parallel-transport along geodesics 
on the space $\som$ of states is equivalent to 
boost transformations, and two summaries of 
elementary mathematical concepts used in the 
text, namely semi-direct products (for groups) 
and affine structures (on sets).

\noindent
\textbf{Notation and conventions:}
Throughout spacetime is a 4-dimensional real 
affine space with associated vector space $V$ 
and dual space $V^*$.  By \emph{the Minkowski metric}, 
denoted by $\eta$, we understand a non-degenerate 
symmetric bilinear form on $V$ of signature 
$(-,+,+,+)$. Definitions, lemmas, 
propositions, theorems, and remarks are 
numbered by a single counter in their 
order of appearance. An exception is 
made by two exercises on 
p.\,\pageref{ex:1}, which are 
numbered separately. 

\noindent
\textbf{Acknowledgements:}
I thank Philip Schwartz and Lea Zybur for 
many conversations that helped me develop 
and refine some of the ideas presented here.

\section{%
Velocity addition/subtraction: the old story
\label{sec:OldStory}
}
In this section we recall the standard, 
matrix-representation-based derivation of 
velocity addition/subtraction, including 
an analysis of the less widely known 
algebraic structure it defines. We shall 
give complete proofs of all relevant 
statements. Some of these proofs appear 
shorter and yet easier to follow than 
those in the published literature. 
\subsection{%
A reminder on the elementary text-book 
derivation
\label{sec:TextBookDer}
}
A frequently-seen derivation of the velocity 
addition in SR ist this: A particle $P$ 
moves along a spatial trajectory $\vec x'(t')$ 
relative to an inertial system $I'$ with whose
coordinates are labelled by $(t',\vec x')$. 
Relative to another inertial system, $I$,
whose coordinates are labelled by 
$(t,\vec x)$, the system $I'$ moves with 
constant velocity $\vec v$. According to SR, 
the boost transformation between $I$ and $I'$
is 
\begin{subequations}
\label{eq:BoostTrans-1}
\begin{alignat}{1}
\label{eq:BoostTrans-1a}
t&\,=\,\gamma(t'+\vec v\cdot\vec x'/c^2)\,,\\
\label{eq:BoostTrans-1b}
\vec x&\,=\,\vec x'+(\gamma-1)\vec n(\vec n\cdot\vec x')
+\gamma\vec v t'\,,
\end{alignat}
\end{subequations}
where 
\begin{equation}
\label{eq:Abbreviations1}
\gamma:=(1-\Vert\vec v\Vert^2/c^2)^{-1/2}\,,\quad
\vec n:=\vec v/\Vert\vec v\Vert\,.
\end{equation}
The particle's $P$ trajectory relative to 
$I$ is then obtained by inserting $\vec x'(t')$ 
for $\vec x'$ into these equations. Both expressions 
then become a functions of $t'$. Taking the 
differential, setting $d\vec x'/dt':=\vec u'$ and
$d\vec x/dt:=\vec u$, we get:  
\begin{subequations}
\label{eq:BoostDifferentials-1}
\begin{alignat}{1}
\label{eq:BoostDifferentials-1a}
dt&\,=\,\gamma(1+\vec v\cdot\vec u'/c^2)\,dt'\,,\\
\label{eq:Boostdifferentials-1b}
d\vec x&\,=\,\bigl[\vec u'+(\gamma-1)\vec n(\vec n\cdot\vec u')
+\gamma\vec v\bigr]\,dt'\,.
\end{alignat}
\end{subequations}
Hence, the particle's velocity relative to $I$ is  
\begin{equation}
\label{eq:SR-VelAdd-1}
\vec u=
\frac{\vec v+\vec u'_{\Vert}+\gamma^{-1}\vec u'_\perp}{1+
\vec v\cdot \vec u'/c^2}\,,
\end{equation}
where 
\begin{equation}
\label{eq:ParallelPerpComp}
{\vec u}'_\Vert:=\vec n(\vec n\cdot\vec u')
\quad\text{and}\quad
\vec u'_\perp:=\vec u'-\vec u'_\Vert
\end{equation}
are the orthogonal projections of $\vec u'$ parallel
and perpendicular to $\vec v$, respectively. 
Alternative ways to write \eqref{eq:SR-VelAdd-1} are 
\begin{subequations}
\label{eq:SR-VelAdd-2}
\begin{alignat}{2}
\label{eq:SR-VelAdd-2a}
\vec u&\,=\ 
\frac{1}{1+\frac{\vec v\cdot\vec u'}{c^2}}
\left\{
\vec v+\vec u'+\frac{\gamma}{c^2(\gamma +1)}\,
\vec v\times(\vec v\times\vec u')
\right\}\\
\label{eq:SR-VelAdd-2b}
&\,=\ 
\frac{1}{1+\frac{\vec v\cdot\vec u'}{c^2}}
\left\{
\vec v+ \gamma^{-1}\vec u'+\frac{\gamma}{c^2(\gamma +1)}\,
\vec v\,(\vec v\cdot\vec u')
\right\}\,.
\end{alignat}
\end{subequations} 
The first term takes a simple form if $\vec v$
is parallel to $\vec u'$ (in which  case 
the term involving the $\times$-products
vanishes) and the second an even simpler 
form if $\vec v$ is perpendicular to 
$\vec u'$ (in which case the term in the 
numerator as well as that in the denominator 
involving the inner product 
$\vec v\cdot\vec u'$ vanishes). These 
formulae express $\vec u$ as function 
of $\vec v$ and $\vec u'$ which is rational 
in  $\vec u'$ but not rational in $\vec v$
due to the ocurrence of $\gamma$. The 
function is $C^\infty$ as long as 
$\Vert\vec v\Vert\leq c$ 
and  $\Vert\vec u'\Vert\leq c$ without 
simultaneous equality, which then also 
imply $\vec v\cdot\vec u'>-c^2$. We have then 
which is the case as long as $\Vert\vec v\Vert\leq c$ 
and $\Vert\vec u'\Vert\leq c$ without simultaneous 
equality. In that case one has $\Vert\vec u\Vert\leq c$ 
with equality if and only if either 
$\Vert\vec v\Vert=c$ or $\Vert\vec u'\Vert=c$. 
This immediately follows from the equation 
\begin{equation}
\label{eq:SR-VelAdd-3}
\gamma(\vec u)=\gamma(\vec v)\gamma(\vec u')\,\bigl(1+\vec v\cdot\vec u'/c^2\bigr)
\end{equation}
where the $\gamma$-factors are now defined 
for any of the three velocities involved; 
i.e., $\gamma(\vec u):=(1-\Vert\vec u\Vert^2/c^2)^{-1/2}$, etc. Equation 
\eqref{eq:SR-VelAdd-3} follows easiest from 
taking the square of \eqref{eq:SR-VelAdd-1}.
We shall encounter another proof below. 

Formulae \eqref{eq:SR-VelAdd-1}, \eqref{eq:SR-VelAdd-2}, 
and \eqref{eq:SR-VelAdd-3} summarise Einstein's law of 
``velocity addition'' as usually discussed. The 
interesting algebraic structure behind it will be the 
subject of the next section. It has been analysed in 
detail before by others  \citep{Ungar:1988,Ungar:1989,Ungar:1997,Ungar:2005,Urbantke:2003}, 
but our treatment here will be different, self-contained, 
and direct. The third and main part of 
this paper addresses and resolves the following 
concern: How can the \emph{addition} of different 
relative velocities possibly make sense in view of 
the fact that the vector space of relative velocities 
depends on the inertial frame it refers to? In our 
example, $\vec u'$ seems to be defined relative to 
$I'$, whereas  $\vec v$ and $\vec u$ are relative to 
$I$. However, there is no obvious way to add vectors 
in the space of velocities relative to $I'$ to vectors 
in the (different) space of velocities relative to 
$I$; and yet, this is precisely what seems to have 
been done above. How can that be a meaningful operation? 
The answer we shall give is that, albeit $\vec u'$ 
does represent the velocity of $I''$ against $I'$, 
it does so with reference to $I$ rather than $I'$. 
How this additional reference to $I$, which renders 
the notion of ``relative velocity'' a ternary rather 
than binary relation, is to be understood properly 
will be explained in detail. Related observations 
had been made before in \citep{Matolcsi.Goher:2001,Matolcsi:2005,Urbantke:2003}
but not in the manifestly covariant form that we shall 
present here.

\subsection{Lorentz transformations 
in the  traditional matrix representation}
\label{sec:BoostTrans}
We begin by slightly reformulating the previous 
formulae in a way that is easier to memorise. 
First, we take all relative velocities in units of 
$c$ and call them (as is usual in the SR-literature) 
by $\vec\beta$. Setting $\vec\beta_1:=\vec v/c$, $\vec\beta_2:=\vec u'/c$ and $\vec\beta:=\vec u/c$, and also $\beta=\Vert\vec\beta\Vert$
etc. for the modulus, as well as $\gamma:=\gamma(\beta)=1/\sqrt{1-\beta^2}$ and correspondingly for 
$\gamma_i:=\gamma(\beta_i)$. Then \eqref{eq:SR-VelAdd-1} and 
\eqref{eq:SR-VelAdd-2} read
\begin{subequations}
\label{eq:SR-VelAdd-4}
\begin{alignat}{1}
\label{eq:SR-VelAdd-4a}
\vec\beta=:\vec\beta_1\vplus\vec\beta_2
&\,=\,\frac{1}{1+\vec\beta_1\cdot\vec\beta_2}\,
\Bigl\{\vec\beta_1+\vec\beta_2^\Vert+
\gamma_1^{-1}\vec\beta_2^\perp\Bigr\}\\
\label{eq:SR-VelAdd-4b}
&\,=\,\frac{1}{1+\vec\beta_1\cdot\vec\beta_2}
\left\{\vec\beta_1+\vec\beta_2+\frac{\gamma_1}{1+\gamma_1}
\vec\beta_1\times(\vec\beta_1\times\vec\beta_2)\right\}\\
\label{eq:SR-VelAdd-4c}
&\,=\,\frac{1}{1+\vec\beta_1\cdot\vec\beta_2}
\left\{\vec\beta_1+\gamma_1^{-1}\vec\beta_2+\frac{\gamma_1}{1+\gamma_1}
\vec\beta_1\,(\vec\beta_1\cdot\vec\beta_2)\right\}\,,
\end{alignat}
\end{subequations} 
where we defines the binary $\vplus$-operation -- the 
``Einstein addition law'' -- by the expressions 
on the right hand side, which, apart from the
square-root in $\gamma_1$, involves the 
velocities only in a rational fashion.
The superscripts $\Vert$ and $\perp$ now refer to the 
Euclidean orthogonal projections parallel and perpendicular 
to $\vec\beta_1$, which, using $\vec n_1:=\vec\beta_1/\beta_1$, 
are given in analogy to \eqref{eq:ParallelPerpComp} by 
\begin{equation}
\label{eq:SR-VelAdd-5}
{\vec\beta_2}^\Vert:=\vec n_1(\vec n_1\cdot\vec\beta_2)
\quad\text{and}\quad
\vec\beta_2^\perp:=\vec\beta_2-\vec\beta_2^\Vert\,.
\end{equation}
Equation \eqref{eq:SR-VelAdd-3} reads 
\begin{equation}
\label{eq:SR-VelAdd-6}
\gamma=\gamma_1\gamma_2(1+\vec\beta_1\cdot\vec\beta_2)\,,
\end{equation}
showing that $\gamma<\infty$ if $\gamma_i<\infty$ ($i=1,2$), and that $\gamma\rightarrow\infty$ 
if and only if at least one 
$\gamma_i\rightarrow\infty$. This implies that 
\begin{equation}
\label{eq:StarMap}
\vplus: \mathring{B}_1(\reals^3)\times \mathring{B}_1(\reals^3)
\rightarrow\mathring{B}_1(\reals^3)\,,
\end{equation}
where $\mathring{B}_1(\reals^3)$ denotes the 
open ball of unit radius in $\reals^3$. In 
this section we will clarify the properties 
of the binary operation that $\vplus$ puts on 
$\mathring{B}_1(\reals^3)$. This would be a 
tedious problem to do from the algebraic 
form alone. In fact, we will see that this 
algebraic form derives from an underlying 
group law, even though $\vplus$ is not itself 
a group multiplication. This results in  
$\vplus$ being -- in some  sense -- a minimal 
non-associative generalisation of group 
multiplication. 

In order to appreciate this later insight 
we encourage the reader to work on 
the following two exercises (doing the first,
trying the second):

\begin{exercise}
\label{ex:1}
Use the expression \eqref{eq:SR-VelAdd-4a} 
for $\vec\beta={\vec\beta}_1\vplus{\vec\beta}_2$
and solve it for ${\vec\beta}_2$, showing that 
\begin{equation}
\label{eq:SR-VelAdd-Inv}
\vec\beta_2
=(-\vec\beta_1)\vplus\vec\beta
=\frac{1}{1-\vec\beta_1\cdot\vec\beta}\,
\Bigl\{-\vec\beta_1+\vec\beta^\Vert+
\gamma_1^{-1}\vec\beta^\perp\Bigr\}\,.
\end{equation}   
Hint: Equation \eqref{eq:SR-VelAdd-4a} 
can be used to express ${\vec\beta}^\Vert$ 
and ${\vec\beta}^\perp$ as functions 
of $\vec\beta_1$ and ${\vec\beta}_2$, and 
also to express $\vec\beta\cdot{\vec\beta}_1$
in terms of ${\vec\beta}_1\cdot{\vec\beta}_2$
and $\gamma_1$. The rest is straightforward 
computation. 
\end{exercise}
\begin{definition}
\label{def:VelDifference}
The expression \eqref{eq:SR-VelAdd-Inv} is called
the \textbf{velocity difference} between 
$\vec\beta$ and $\vec\beta_1$. 
\end{definition}

\begin{exercise}
\label{ex:2}
Try to repeat Exercise\ref{ex:1}, now  
solving $\vec\beta={\vec\beta}_1\vplus{\vec\beta}_2$
for $\vec\beta_1$ (as function of $\vec\beta$ and 
$\vec\beta_2$) rather than $\vec\beta_2$
(as function of $\vec\beta$ and 
$\vec\beta_1$) as above. There is a unique 
solution!
\end{exercise} 
It is likely that you will not find the 
solution to Exercise\,\ref{ex:2} at this point
(it is given in \eqref{eq:AlgStr-14} below). 
It is \emph{not} given by 
$\vec\beta\vplus(-\vec\beta_2)$ (as one might 
have expected at first) unless $\vec\beta$ 
and $\vec\beta_2$ are parallel. Hence there is 
no immediate analogue of 
Definition\,\ref{def:VelDifference} 
for the ``other'' velocity difference, i.e.  
that between $\vec\beta$ and $\vec\beta_2$.

\subsection{The generalised orthogonal group and its Lie algebra}
\label{sec:GenOrthGroupAlg}
Let us first consider the general case where $V$ 
be a real $n$-dimensional vector space, $V^*$ 
its dual space, and $\eta\in V^*\otimes V^*$ a 
symmetric, non-generate, bilinear form of any
signature. Recall that the 
\emph{signature} of such a form is given by a 
pair $(n_-,n_+)$ of non-negative integers which 
denote the dimensions of maximal linear 
subspaces $V_-$ and $V_+$  of $V$ restricted to 
which $\eta$ is negative- and positive-definite, 
respectively.\footnote{Our convention is mostly 
shared in the physics literature, whereas in the 
mathematics literature The word ``signature'' is 
often used differently. For example, 
\citet[p.\,269]{Greub:LinearAlgebra} 
calls $n_+$ the \emph{index} and $n_+-n_-$ the 
\emph{signature}.}  
Whereas $n_\pm$ are uniquely defined by $\eta$, 
$V_\pm$ are not.  Note that $V=V_-\oplus V_+$, 
hence $n=n_-+n_+$, with $V_+$ and $V_-$ 
$\eta$-orthogonal to each other. If $n_-=1$ 
and $n_+=n-1$ (or vice versa; though we will 
stick to the ``mostly plus'' convention 
in this paper) $\eta$ is called a 
\emph{Minkowski metric} in $n$ dimensions. 
All statements in this subsection apply to 
general $n>1$ and $(n_-,n_+)$, unless stated 
otherwise. 
\begin{definition}
\label{def:GenOrthGroup}
The \textbf{(generalised) orthogonal group} of 
$(V,\eta)$, denoted by $\group{O}(V,\eta)$, 
is defined by the following subgroup of the 
group $\group{Gl(V)}$ of linear isomorphsims
of $V$:
\begin{equation}
\label{eq:OrthogonalLieGroup}
\group{O}(V,\eta)=
\bigl\{
L\in\group{GL}(V):
\eta(Lv,Lw)=\eta(v,w)\,;\, \forall v,w\in V
\bigr\}\,.
\end{equation}
It follows from this definition that $\det(L)=\pm 1$
for all $L\in\group{O}(V,\eta)$. The \emph{special 
(generalised) orthogonal group} is the subgroup of 
orientation preserving elements in  
$\group{O}(V,\eta)$:
\begin{equation}
\label{eq:SpecOrthogonalLieGroup}
\group{SO}(V,\eta)=
\bigl\{
L\in\group{O}(V,\eta):\det(L)=1
\bigr\}\,.
\end{equation}
If $\eta$ is either positive or negative definite, 
i.e. if either $n_-=0$ or $n_+=0$, the group 
$\group{O}(V,\eta)$ contains two connected 
components with $\group{SO}(V,\eta)$ being 
the component containing the group identity
(usually called the ``identity component'').
In all other cases, i.e. for 
$n_\pm\geq 1$, the group 
$\group{O}(V,\eta)$ decomposes (as set) into
the disjoint union (denoted by $\sqcup$) of  
four connected components (see, e.g., the 
book by  \citet[pp.\,237]{O-Neill:SRG}), 
\begin{equation}
\label{eq:OrthogonalLieGroupDecomp}
\group{O}(V,\eta)=\underbrace{
\group{O}_{(+,+)}(V,\eta)\sqcup
\group{O}_{(-,-)}(V,\eta)}_{\group{SO}(V,\eta)}\sqcup\,
\group{O}_{(-,+)}(V,\eta)\sqcup
\group{O}_{(+,-)}(V,\eta)\,,
\end{equation}
where the $\pm$ in the (first, second)  slot 
indicates whether the orientation amongst the 
(negative, positive)---definite subspaces is 
preserved or reversed.
\end{definition}

If $\{e_1,\cdots, e_n\}$ is a basis for $V$
with dual basis $\{\theta^1,\cdots, \theta^n\}$
for $V^*$, so that $\theta^a(e_b)=\delta^a_b$, 
we can write 
\begin{equation}
\label{eq:LorentzTrans-1}
L=\tensor{L}{^a_b}\ e_a\otimes\theta^b
\end{equation}
and 
\begin{equation}
\label{eq:LorentzTrans-2}
\eta=\eta_{ab}\,\theta^a\otimes\theta^b
\end{equation}

We recall that the metric $\eta$ defines an isomorphism 
\begin{equation}
\label{eq:DefEtaDown}
\eta_\downarrow:V\rightarrow V^*\,,\quad
v\mapsto\eta_\downarrow(v):=\eta(v,\cdot)\,,
\end{equation}
that has an inverse  
\begin{equation}
\label{eq:DefEtaUp}
\eta_\uparrow: V^*\rightarrow V\,,\quad 
\lambda\mapsto\eta_\uparrow(\lambda):=
(\eta_\downarrow)^{-1}(\lambda)\,.
\end{equation}
Hence $\eta$ defines a likewise non-degenerate 
symmetric bilinear form in the dual space $V^*$,
which we call $\eta^*$. Its basis-idependent definition is 
\begin{equation}
\label{eq:DefInverseMetric}
\eta^*(\lambda,\sigma):=\eta\bigl(\eta_\uparrow(\lambda),\eta_\uparrow(\sigma\bigr)
=\lambda\bigl(\eta_\uparrow(\sigma)\bigr)
=\sigma\bigl(\eta_\uparrow(\lambda)\bigr)\,.
\end{equation}
With respect to the pair auf dual bases this 
leads to
\begin{equation}
\label{eq:LorentzTrans-3}
\eta^*=\eta^{ab}\, e_a\otimes e_b
\end{equation} 
where $\eta^{ab}$ are the components of the 
matrix that is inverse to $\eta_{ab}$. Hence 
we can write $\eta_{ab}=\eta(e_a,e_b)$ and 
$\eta^{ab}=\eta^*(\theta^a,\theta^b)$. 
Equivalent to \eqref{eq:OrthogonalLieGroup}
is then 
\begin{equation}
\label{eq:OrthogonalLieGroupDual}
\group{O}(V,\eta)=
\bigl\{
L\in\mathrm{GL}(V):
\eta^*\bigl(L^\top\lambda ,L^\top\sigma\bigr)=\eta^*(\lambda,\sigma)\,;\, \forall \lambda,\sigma\in V^*
\bigr\}\,,
\end{equation}
where $L^\top:V^*\rightarrow V^*$ is the transposed map 
naturally associated to $L$.
In terms of components   \eqref{eq:OrthogonalLieGroup} and \eqref{eq:OrthogonalLieGroupDual} read, 
respectively,
\begin{subequations}
\label{eq:DefLorentzTransComp}
\begin{alignat}{1}
\label{eq:DefLorentzTransComp-a}
\tensor{L}{^a_c}
\tensor{L}{^b_d}\ \eta_{ab}
=\eta_{cd}\,,\\
\label{eq:DefLorentzTransComp-b}
\tensor{L}{^a_c}
\tensor{L}{^b_d}\ \eta^{cd}
=\eta^{ab}\,.
\end{alignat}
\end{subequations}  

Taking the $s$-derivative at $s=0$ of \eqref{eq:OrthogonalLieGroup}, where $L=L(s)$ 
is a differentiable curve in $\group{O}(V,\eta)$ with 
$L(s=0)=\id_V$, we obtain for 
$\ell=dL/ds\vert_{s=0}$ the defining relation 
for the Lie algebra:
\begin{equation}
\label{eq:OrthogonalLieAlgebra}
\mathfrak{o}(V,\eta)=
\bigl\{\ell\in\mathrm{End}(V):
\eta(\ell v,w)=-\eta(v,\ell w)\,;\, \forall v,w\in V \bigr\}\,.
\end{equation}
In other words, $\mathfrak{o}(V,\eta)$ is given 
by the $\eta$-antisymmetric endomorphisms of 
$V$.
 
The given isomorphism \eqref{eq:DefEtaDown}
provides a specific identification of $V^*$ 
with $V$ via replacing any $\alpha\in V^*$ 
with $v=\eta_\uparrow(\alpha)$. Hence we may
also identify $\mathrm{End}(V)$, which 
is \emph{naturally} isomorphic to $V\otimes V^*$, 
with $V\otimes V$. This will notationally 
simplify later expressions.  In particular, 
$\mathfrak{o}(V,\eta)$ is then identified with 
$V\wedge V$. The natural Lie product on 
$\mathrm{End}(V)$, which is just given by 
the commutator, induces a Lie product 
on $V\otimes V$, which on pure tensor 
products is given by  
\begin{equation}
\label{eq:DefLieAlgEnd}
\bigl[v\otimes w\,,v'\otimes w'\bigr]
=\eta(w,v')\,v\otimes w'
-\eta(w',v)\,v'\otimes w
\end{equation}
with unique bilinear extension to all of 
$V\otimes V$. On $\mathfrak{o}(V,\eta)$ 
this becomes (recall 
$v\wedge w=v\otimes w-w\otimes v$):
\begin{equation}
\label{eq:DefLieAlgOrth}
\begin{split}
\bigl[v\wedge w\,,v'\wedge w'\bigr]
= &\ \eta(v,w')\,w\wedge v'
+    \eta(w,v')\,v\wedge w'\\
-&\ \eta(v,v')\ w\wedge w'
 -\eta(w,w')\,v\wedge v'\,.
\end{split}
\end{equation}

If $\{e_1,\cdots, e_n\}$ is a basis for $V$ with 
$\eta_{ab}:=\eta(e_a,e_b)$ and 
$M_{ab}:=e_a\wedge e_b$, \eqref{eq:DefLieAlgOrth}
takes the form found in many physics-textbooks:
\begin{equation}
\label{eq:DefLieAlgOrthComp}
\begin{split}
\bigl[M_{ab},M_{cd}\bigr]
=\eta_{ad}M_{bc}
+\eta_{bc}M_{ad}
-\eta_{ac}M_{bd}
-\eta_{bd}M_{ac}\,.
\end{split}
\end{equation}
The set $\{M_{ab}:1\leq a<b\leq n\}$ forms a basis 
for the $\frac{1}{2}n(n-1)$ - dimensional Lie 
algebra $\mathfrak{o}(V,\eta)$.

\subsection{Restriction to four dimensions and Lorentzian signature}
\label{sec:LorentzGroupMatrix}
From now on we restrict to the special-relativistic 
case in which $(n_-,n_+)=(1,3)$.\footnote{Or, 
alternatively, $(n_-,n_+)=(3,1)$. But in this paper 
we shall adopt the ``mostly plus'' convention.}
In this case the (full) \emph{orthogonal group}, 
$\group{O}(V,\eta)$, is called the \emph{Lorentz group}, abbreviated by $\Lor$. As we have seen above, it 
has four components. The component containing the 
identity is $\group{O}_{(+,+)}(V,\eta)$ and usually 
called the \emph{proper orthochronous Lorentz group}, 
where the ``proper'' stands for ``overall-orientation preserving'', meaning that one restricts to 
$L\in \group{SO}(V,\eta)=\group{O}_{(+,+)}(V,\eta)\,
\sqcup\,\group{O}_{(-,-)}(V,\eta)$, i.e. to those 
$L$ with $\det(L)=1$. ``Orthochronous'' stands for 
``time-orientation preserving'', which means 
that one restricts to $L\in\group{O}_{(+,+)}(V,\eta)\,
\sqcup\,\group{O}_{(+,-)}(V,\eta)$, i.e. to those 
$L$ with $\eta(Lv,v)<0$ for all timelike $v$. Note 
that $\group{O}_{(+,+)}(V,\eta)\,\sqcup\,\group{O}_{(-,+)}(V,\eta)$, too, is a two-component subgroup of 
$\group{O}(V,\eta)$ which is sometimes called the 
\emph{orthochorous Lorentz group}; compare, e.g., 
\citep[p.\,11]{Streater.Wightman:AllThat}.

\begin{remark}
In what follows we will restrict attention to the 
component of the identity only and simply write 
\begin{equation}
\Lor:=\group{O}_{(+,+)}(V,\eta)\,,
\quad\text{for\ $(n_-,n_+)=(1,3)$}\,. 
\end{equation}
Note that often $\Lor$ (or just $L$) stands for the whole 
Lorentz group whereas the proper orthochronous subgroup 
is denoted by $\Lor_+^\uparrow$ (or $L_+^\uparrow$), where 
the subscript $+$ represents the ``proper'' and the 
superscript $\uparrow$ the ``orthochronous''. But since 
here we do not consider time- or space-reversing 
maps we simplify the notation as indicated. We remark 
that everything we are going to say simply extends to 
\emph{all} time-orientation preserving Lorentz transformations,
i.e. to $\group{O}_{(+,+)}(V,\eta)\,\sqcup\,\group{O}_{(+,-)}(V,\eta)$ (or to $\Lor_+^\uparrow\sqcup\Lor_-^\uparrow$). 
In order to also include time-orientation reversing transformations we would have to generalise our later Definition\,\ref{def:StateOfMotion-1} 
of the set of states of motion to also include the 
other connected component (called $V_1^-$) of the 
hyperboloid of unit timelike vectors. We avoid these complications since we will not need this generalisation
in the sequel.  
\end{remark}

According to \eqref{eq:OrthogonalLieGroup} 
and \eqref{eq:OrthogonalLieGroupDual},
the matrices representing Lorentz 
transformations must satisfy 
\eqref{eq:DefLorentzTransComp-a} and \eqref{eq:DefLorentzTransComp-b}, respectively. 
We choose a basis $\{e_0,e_1,e_2,e_3\}$ so that 
\begin{equation}
\label{eq:MinkMetrikComp}
1=-\eta_{00}
=\eta_{11}
=\eta_{22}
=\eta_{33}\quad\text{and}\quad
\eta_{ab}=0\quad\text{for}\quad a\ne b\,.
\end{equation}
If we write a general matrix with components 
$\tensor{L}{^a_b}$ in (1+3)-decomposed form 
as\footnote{The symbol $c$ for the upper-left 
matrix entry in \eqref{eq:LorentzMatrix-1} 
should not be confused with the velocity of 
light, which does not explicitly appear in 
any of the following formulae.}
\begin{equation}
\label{eq:LorentzMatrix-1}
\bigl\{\tensor{L}{^a_b}\bigr\}=
\begin{pmatrix}
c & \vec a^\top\\
\vec b &\mat{M}
\end{pmatrix}\,,
\end{equation}
where $c\in\reals$,  
$\vec a,\vec b\in \reals^3$, and $\mat{M}\in\mathrm{End}(\reals^3)$ is a 
$(3\times 3)$-matrix\footnote{Generally, in this section 
elements in $\reals^3$ are written by bold-faced letters 
and elements in $\group{End}(\reals^3)$ (i.e. $3\times 3$ 
matrices) by underlined bold-faced letters. The transposed
object of either of these is indicated by a 
superscript $\top$.}, equations  
\eqref{eq:DefLorentzTransComp-a} and 
 \eqref{eq:DefLorentzTransComp-b} are 
equivalent to, respectively, 
\begin{subequations}
\label{eq:LorentzTransRel-1}
\begin{alignat}{1}
\label{eq:LorentzTransRel-1a}
\Vert\vec b\Vert^2&\,=\,c^2-1\,,\\
\label{eq:LorentzTransRel-1b}
\mat{M}^\top\vec b &\,=\,c\vec a\,,\\
\label{eq:LorentzTransRel-1c}
\mat{M}^\top\mat{M} &\,=\,\mat{E}_3+\vec a\otimes\vec a^\top\,,	
\end{alignat}
\end{subequations}	
and 
\begin{subequations}
\label{eq:LorentzTransRel-2}
\begin{alignat}{1}
\label{eq:LorentzTransRel-2a}
\Vert\vec a\Vert^2&\,=\,c^2-1\,,\\
\label{eq:LorentzTransRel-2b}
\mat{M}\vec a &\,=\,c\vec b\,,\\
\label{eq:LorentzTransRel-2c}
\mat{M}\,\mat{M}^\top &\,=\,\mat{E}_3+\vec b\otimes\vec b^\top\,,
\end{alignat}
\end{subequations}
where $\mat{E}_3$ denotes the $(3\times 3)$
unit matix.  
Note that \eqref{eq:DefLorentzTransComp-b} follows 
from \eqref{eq:DefLorentzTransComp-a} by 
replacing $L$ with its transposed (due to the 
numerical equality of $\eta_{ab}$ with 
$\eta^{ab}$). Accordingly, the set \eqref{eq:LorentzTransRel-2} follows 
from the set \eqref{eq:LorentzTransRel-1} by exchanging $\vec a$ with $\vec b$ and $\mat{M}$ 
with $\mat{M}^\top$. 

Clearly, spatial rotations and pure boosts \eqref{eq:BoostTrans-1} satisfy these equations.
In fact, ``spatial rotation'' here means to 
embed  $R:\group{SO}(3)\hookrightarrow\Lor$, 
$\mat{D}\mapsto R(\mat{D})$, according to 
\begin{equation}
\label{eq:EmbedRotation}
R(\mat{D}):=
\begin{pmatrix}
1&\vec 0^\top\\
\vec 0 & \mat{D}
\end{pmatrix}\,.
\end{equation}
With respect to the same basis (for which 
$x^0=ct$), a boost \eqref{eq:BoostTrans-1} 
corresponds to 
\begin{equation}
\label{eq:EmbedBoost}
B(\vec\beta):=
\begin{pmatrix}
\gamma &\gamma\vec\beta^\top\\
\gamma\vec\beta & \mat{E}_3
+(\gamma-1)\vec n\otimes\vec n^\top
\end{pmatrix}\,,
\end{equation}
where we abbreviate $\vec n:=\vec\beta/\beta$ and each entry 
in the matrix on the right-hand side is to be regarded 
as function of $\vec\beta$. In particular, since 
$\beta^2=1-\gamma^{-2}$, we have $(\gamma-1)\vec n\otimes\vec n^\top=\frac{\gamma^2}{1+\gamma}\vec\beta\otimes\vec\beta^\top$, which is a continuously 
differentiable function of $\vec\beta$.

For \eqref{eq:EmbedRotation} satisfaction of 
\eqref{eq:LorentzTransRel-1} is immediate
(with \eqref{eq:LorentzTransRel-2} following by 
replacing $\mat{D}$ with its transposed 
(= inverse)). Likewise, noting that 
$\mat{E}_3+(\gamma-1)\vec n\otimes\vec n^\top=
P_\perp+\gamma P_\Vert$, with $P_\perp$ and  
$P_\Vert$ the projectors in $\reals^3$ 
perpendicular and parallel to $\vec n$, 
respectively, we immediately see that 
\eqref{eq:LorentzTransRel-1b} and 
\eqref{eq:LorentzTransRel-1c} are satisfied. 
Equation \eqref{eq:LorentzTransRel-1a} is 
equivalent to the identity 
$\beta^2\gamma^2=\gamma^2-1$.
The converse to this result is the following: 

\begin{theorem}
Any Lorentz transformation 
\eqref{eq:LorentzMatrix-1} that is `proper', 
i.e. satisfies $\det(L)=+1$ and 
`orthochronous', i.e. $c \geq 1$, can be 
composed into the product of a rotation and 
a boost. This decomposition is just the 
``polar decomposition'' with respect to the 
standard Euclidean metric in $\reals^4$ 
(compare Appendix\,\ref{sec:PolarDecomposition})
and hence unique if the order of 
rotation-boost-multiplication is fixed. 
The reversed order has the same rotation but a 
boost velocity  that differs from the first by 
that rotation. That is,  
\begin{subequations}
\begin{equation}
\label{eq:LT-PolarDec-a}
\bigl\{\tensor{L}{^a_b}\bigr\}=
\begin{pmatrix}
c & \vec a^\top\\
\vec b &\mat{M}
\end{pmatrix}
=B(\vec\beta)R(\mat{D})
=R(\mat{D})B(\vec\beta')\,,
\end{equation} 
where $\mat{D}\in\group{SO}(3)$,  
$\vec\beta\in\mathring{B}_1(\reals^3)$, and 
 are $\vec\beta'\in\mathring{B}_1(\reals^3)$
are determined as rational functions of 
$\gamma,\vec a,\vec b$, and $\mat{M}$ by
\begin{alignat}{1}
\label{eq:LT-PolarDec-b}
\mat{D}&:\,=\,\mat{M}-\frac{\vec b\otimes\vec a^\top}{\gamma+1}\,,\\
\label{eq:LT-PolarDec-c}
\vec\beta &:\,=\,\vec b/c\,,\quad
\vec\beta'\,=\,\vec a/c\,.
\end{alignat}
\end{subequations}
We further have $\vec\beta=\mat{D}\vec\beta'$; 
compare \eqref{eq:ProofPolarDec-2} and 
\eqref{eq:ProofPolarDec-6} below.
\end{theorem}
\begin{proof}
We only need to show the second equality in 
\eqref{eq:LT-PolarDec-a}, as the third then 
follows from the proof of the first. We proceed 
by proving three things: 
1)~That the product $B(\vec\beta)R(\mat{D})$
with $\mat{D}$ and $\vec\beta$ as in 
\eqref{eq:LT-PolarDec-b} and \eqref{eq:LT-PolarDec-c}
indeed equals $\bigl\{\tensor{L}{^a_b}\bigr\}$;
2)~that $B(\vec\beta)\in\group{GL}(\reals^4)$
is symmetric and positive definite; and
 3)~that $\mat{D}\in\group{SO}(3)$.
Now, the product of $B(\vec\beta)$ given 
by \eqref{eq:EmbedBoost} and $R(\mat{D})$ 
given by \eqref{eq:EmbedRotation} is 
\begin{equation}
\label{eq:ProofPolarDec-1}
\begin{pmatrix}
\gamma 
 & \gamma(\mat{D}^\top \vec\beta)^\top\\
\gamma\vec\beta 
 & \mat{D}+(\gamma-1)\vec n\otimes(\mat{D}^\top\vec n)^\top
\end{pmatrix}\,.
\end{equation} 
This must equal \eqref{eq:LorentzMatrix-1},
which for the upper-left entry implies 
$\gamma=c$ and the lower-left entry 
$\gamma\vec\beta=\vec b$. Since 
$\Vert\gamma\vec\beta\Vert^2
=\gamma^2\beta^2=\gamma^2-1$ this is 
only consistent with the latter equation 
if $\Vert\vec b\Vert^2=c^2-1$, which is 
just guaranteed by \eqref{eq:LorentzTransRel-1a}. 
Given $\gamma\vec\beta=\vec b$, equality of the upper-right entry of \eqref{eq:ProofPolarDec-1} and \eqref{eq:LorentzMatrix-1} is equivalent 
to $\mat{D}^\top\vec b=\vec a$. This is indeed 
satisfied by \eqref{eq:LT-PolarDec-b} since then
\begin{equation}
\label{eq:ProofPolarDec-2}
\mat{D}^\top\vec b=\mat{M}^\top\vec b-\vec a\frac{\Vert\vec b\Vert^2}{\gamma+1}
= \vec a\left(c-\frac{c^2-1}{\gamma+1}\right)=\vec a
\end{equation}
where we used \eqref{eq:LorentzTransRel-1a}
and \eqref{eq:LorentzTransRel-1b} for the 
second and $c=\gamma$for the third equality. 
This also shows equality of the lower-right 
entries, which requires       
\begin{equation}
\label{eq:ProofPolarDec-3}
\mat{D}
=\mat{M}-(\gamma-1)
 \vec n\otimes (\mat{D}^\top\vec n)^\top
=\mat{M}-(\gamma-1)\frac{\vec b\otimes(\mat{D}^\top\vec b)^\top}{\beta^2\gamma^2}\,,
\end{equation}
which is equivalent to \eqref{eq:LT-PolarDec-b}
in view of \eqref{eq:ProofPolarDec-2} and 
$\beta^2\gamma^2=\gamma^2-1$. Next we check orthogonality of $\mat{D}$:
\begin{equation}
\label{eq:ProofPolarDec-4}
\begin{split}
\mat{D}^\top \mat{D}
&=
\left(
\mat{M}^\top-\frac{\vec a\otimes\vec b^\top}{\gamma+1}
\right)
\left(
\mat{M}-\frac{\vec b\otimes\vec a^\top}{\gamma+1}
\right)\\
&=\mat{M}^\top\mat{M}
-\frac{\mat{M}^\top \vec b\otimes\vec a+\vec a\otimes (\mat{M}^\top \vec b)^\top}{\gamma+1}
+\Vert\vec b\Vert^2\frac{\vec a\otimes\vec a^\top}{(\gamma+1)^2}\\
&=\mat{E}_3+\vec a\otimes\vec a^\top\left(1-
\frac{2c}{\gamma+1}+\frac{c^2-1}{(\gamma+1)^2}\right)\\
&=\mat{E}_3\,,
\end{split}
\end{equation}
where we used \eqref{eq:LorentzTransRel-1b}
and \eqref{eq:LorentzTransRel-1c} in the third and 
$c=\gamma$ in the last step. Similarly we 
could have proven $\mat{D}\,\mat{D}^\top=\mat{E}_3$ 
using the relations \eqref{eq:LorentzTransRel-2}. 
Finally, symmetry of $B(\vec\beta)$ is obvious 
and positive definiteness follows from its 
eigenvalues, which depend on the modulus $\beta$
of $\vec\beta$ but not on its direction. 
Hence we may choose, e.g.,  $\vec\beta=\beta\vec e_1$
to calculate its characteristic polynomial, 
which reads 
\begin{equation}
\label{eq:ProofPolarDec-5}
P(\lambda)=
\bigl[(\gamma-\lambda)^2-\beta^2\gamma^2\bigr](1-\lambda)^2\,.
\end{equation}
Hence the eigenvalues are all positive and given by    
\begin{equation}
\label{eq:BoostEigenvalues}
\lambda_{1}=\sqrt{\frac{1+\beta}{1-\beta}}\,,\qquad
\lambda_{2}=\sqrt{\frac{1-\beta}{1+\beta}}\,,\qquad
\lambda_{3{,}4}=1\,.
\end{equation}
Finally, given orthogonality of $R(\mat{D})$
and symmetry as well as positive-definiteness 
of $B(\vec\beta)$, we see that the right-hand 
side of \eqref{eq:LT-PolarDec-a} ist just the 
polar decomposition of $\{\tensor{L}{^a_b}\}\in\group{GL}(\reals^4)$, which is unique -- if we agree on 
the order of the two terms on the right hand 
side, i.e. the orthogonal matrix $R(\mat{D})$ 
to the right of the positive-definite symmetric 
one. Had we chosen the opposite order, as in 
the third equality in \eqref{eq:LT-PolarDec-a}, 
the two matrices would again be uniquely determined
with the same rotation but a different (rotated) 
boost $\vec\beta'=\mat{D}^{-1}\vec\beta$. This is 
a simple consequence of the general relation 
\begin{equation}
\label{eq:ProofPolarDec-6}
R(\mat{D})B(\vec\beta)R(\mat{D}^{-1})
=B(\mat{D}\vec\beta)\,,
\end{equation}
that is an immediate consequence of 
\eqref{eq:EmbedRotation} and 
\eqref{eq:EmbedBoost}. 
Now, since 
$\beta=\vec b/c$ and $\mat{D}^{-1}=\mat{D}^\top$,
equation \eqref{eq:ProofPolarDec-2} just shows  
the second equation of \eqref{eq:LT-PolarDec-c}.
\end{proof}

\subsection{Polar decomposition of boost products}
\label{sec:BoostProducts}
As in section \eqref{sec:LorentzGroupMatrix}
we pick a state of motion $s\in\som$ which we 
take as our zeroth (time-like) basis vector 
for $V$; i.e. we choose $e_0=s$ and complement 
this to an orthonormal basis 
$\{e_0,e_1,e_2,e_3\}$ of $V$. With respect to 
that choice we shall now speak of pure boosts, 
whose matrix representatives look like 
\eqref{eq:EmbedBoost} and, and also of 
``polar decomposition''. The task is to polar decompose
the product of boosts $B(\vec\beta_1)$ and 
 $B(\vec\beta_2)$:
\begin{equation}
\label{eq:ProductBoosts-1}
B(\vec\beta_1)B(\vec\beta_2)=B(\vec\beta)R(\mat{D})
\,.
\end{equation}   
The boost parameter $\vec\beta$ and the rotation 
matrix $\mat{D}$ on the right-hand side are 
uniquely determined by $\vec\beta_1$ and 
$\vec\beta_2$. In other words, there are 
functions
\begin{subequations}
\label{eq:ProductBoosts-2}
\begin{alignat}{2}
\label{eq:ProductBoosts-2a}
\vec\beta:
 \mathring{B}_1(\reals^3)\times
 \mathring{B}_1(\reals^3)
&\rightarrow \mathring{B}_1(\reals^3)\,,\\
\label{eq:ProductBoosts-2b}
\mat{D}:
 \mathring{B}_1(\reals^3)\times
 \mathring{B}_1(\reals^3)
&\rightarrow \group{SO}(3)\,,
\end{alignat}
\end{subequations}
which we now determine. For that we just need to 
follow the procedure outlined in 
section\,\ref{sec:LorentzGroupMatrix}. Using the matrix representation \eqref{eq:EmbedBoost} for
$B(\vec\beta_1)$ and  $B(\vec\beta_2)$ we obtain 
a matrix of the form \eqref{eq:LorentzMatrix-1}
with 
\begin{subequations}
\label{eq:ProductBoosts-3}
\begin{alignat}{1} 
\label{eq:ProductBoosts-3a}
c &\,=\,
\gamma_1\gamma_2(1+\vec\beta_1\vec\beta_2)\,,\\
\label{eq:ProductBoosts-3b}
\vec a&\,=\,
\gamma_1\gamma_2\bigl(\vec\beta_2+\vec\beta_1^\Vert+\gamma^{-1}_2\vec\beta_1^\perp\bigr)\,,\\
\label{eq:ProductBoosts-3c}
\vec b&\,=\,
\gamma_1\gamma_2\bigl(\vec\beta_1+\vec\beta_2^\Vert+\gamma^{-1}_1\vec\beta_2^\perp\bigr)\,,\\
\label{eq:ProductBoosts-3d}
\mat{M}&\,=\,
\mat{E}_3+
(\gamma_1-1)\,\vec n_1\otimes \vec n_1^\top+
(\gamma_2-1)\,\vec n_2\otimes \vec n_2^\top
\nonumber\\
&\qquad\quad+
\bigl[\beta_1\gamma_1\beta_2\gamma_2+
(\gamma_1-1)(\gamma_2-1)(\vec n_1\cdot\vec n_2)\bigr]\,\vec n_1\otimes\vec n_2^\top\,,
\end{alignat}
\end{subequations}
where the superscripts $\Vert$ and $\perp$ on 
$\vec\beta_1$ and $\vec\beta_2$ refer to the 
projections parallel and perpendicular to the ``other'' $\vec\beta$, i.e. $\vec\beta_2$ and 
$\vec\beta_1$ respectively. From $c=\gamma$ and 
\eqref{eq:LT-PolarDec-c} we get 
\begin{subequations}
\label{eq:ProductBoosts-4}
\begin{alignat}{3}
\label{eq:ProductBoosts-4a}
\vec\beta &\,=\,\vec b/c&&\,=\,\frac{\vec\beta_1+\vec\beta_2^\Vert+\gamma_1^{-1}\vec\beta_2^\perp}{1+\vec\beta_1\cdot\vec\beta_2}&&\,=\,\vec\beta_1\vplus\vec\beta_2\,,\\
\label{eq:ProductBoosts-4b}
\vec\beta' &\,=\,\vec a/c&&\,=\,\frac{\vec\beta_2+\vec\beta_1^\Vert+\gamma_2^{-1}\vec\beta_1^{\perp}}{1+\vec\beta_1\cdot\vec\beta_2}&&\,=\,\vec\beta_2\vplus\vec\beta_1\,,
\end{alignat}
and, since $c=\gamma$, 
\begin{equation}
\label{eq:ProductBoosts-4c}
\gamma=\gamma(\beta)
=\gamma(\beta')
=\gamma_1\gamma_2(1+\vec\beta_1\cdot\vec\beta_2)
=\gamma_1\gamma_2\bigl(1+\beta_1\beta_2\,\cos(\varphi)\bigr)\,,
\end{equation}
\end{subequations}
if $\varphi$ denotes the angle between 
$\vec\beta_1$ and $\vec\beta_2$.
Hence the boost contained in the polar 
decomposition of the product of boosts with 
parameters $\vec\beta_1$ and $\vec\beta_2$ is 
just the boost with Einstein added parameters 
$\vec\beta=\vec\beta_1\vplus\vec\beta_2$.

Next we turn to the rotation $\mat{D}$. 
Like $\vec\beta$ it is a function of 
$\vec\beta_1$ and $\vec\beta_2$ and is 
given a special name:
\begin{definition}
\label{def:ThomasRotation}
The rotation matrix $\mat{D}$ resulting from 
polar decomposition \eqref{eq:ProductBoosts-1}
of the product of two boosts 
$B(\vec\beta_1)\,B(\vec\beta_2)$ is called
the \textbf{Thomas rotation} and denoted by
$\mat{T}(\vec\beta_1,\vec\beta_2)$.
\end{definition}
The Thomas rotation can be written down by 
forming the expression on the right-hand 
side of \eqref{eq:LT-PolarDec-b} using 
\eqref{eq:ProductBoosts-4}, but the algebraic 
expression is complex and not immediately 
telling. It can be understood in simple terms as 
follows. First we observe that it is a linear 
combination of $\mat{E}_3$, 
$\vec n_1\otimes\vec n_1$, 
$\vec n_2\otimes\vec n_2$,  
$\vec n_1\otimes\vec n_2$, and 
$\vec n_2\otimes\vec n_1$. Hence it maps the plane 
$\Span\{\vec n_1,\vec n_2\}= \Span\{\vec\beta_1,\vec\beta_2\}$ into itself and 
leaves the orthogonal complement pointwise fixed; 
in other words: it is a rotation in the plane 
spanned by the two boost velocities. That plane 
contains $\vec a$ and $\vec b$ and we know from 
 \eqref{eq:ProofPolarDec-2} that 
$\vec b=\mat{D}\vec a$.\footnote{Note that 
whereas $\vec b=\mat{D}\vec a$ is always 
true, $\mat{D}$  need generally not be a 
rotation in the plane spanned by 
$\vec a$ and $\vec b$. For the Lorentz transformation resulting from the composition 
of two boosts, however, this is the case.} 
Hence the rotation angle $\theta$ is the 
angle between $\vec a$ and $\vec b$,
\begin{equation}
\label{eq:ProductBoosts-5}
\cos(\theta)
=\frac{\vec a\cdot\vec b}{\Vert\vec a\Vert\,\Vert\vec b\Vert}=\frac{\vec a\cdot\vec b}{\gamma^2-1}\,,
\end{equation}
counted positively with respect to the 
orientation given by the ordered pair 
$\{\vec a,\vec b\}$. Here we used  
$\Vert\vec a\Vert=\Vert\vec b\Vert=\sqrt{\gamma^2-1}$ (compare 
\eqref{eq:LorentzTransRel-1a} and \eqref{eq:LorentzTransRel-1b}, using 
$c=\gamma$) in the last step. 

On the other hand, according to the general formula for the rotation angle, we have  
\begin{equation}
\label{eq:ProductBoosts-6}
1+2\cos(\theta)
=\mathrm{trace}(\mat{D})
=\mathrm{trace}(\mat{M})-\frac{\vec a\cdot\vec b}{\gamma+1}\,,
\end{equation}
using \eqref{eq:LT-PolarDec-b} in the last step.
Replacing $\vec a\cdot\vec b$ according to 
\eqref{eq:ProductBoosts-5} leads to an equation 
for $\cos(\theta)$ that we can solve:
\begin{equation}
\label{eq:ProductBoosts-7}
\cos(\theta)
=\frac{\mathrm{trace}(\mat{M})-1}{\gamma+1}\,.
\end{equation}
Now, \eqref{eq:ProductBoosts-3d} gives 
\begin{equation}
\begin{split}
\label{eq:ProductBoosts-8}
\mathrm{trace}(\mat{M})-1
&=3
+(\gamma_1-1)
+(\gamma_2-1)
+\beta_1\gamma_1\beta_2\gamma_2\,\cos(\varphi)\\
&\qquad+(\gamma_1-1)(\gamma_2-1)\,\cos^2(\varphi)\\
&=2+\gamma_1\gamma_2\bigl(1+\beta_1\beta_2\cos(\varphi)\bigr)\\&\qquad-(\gamma_1-1)(\gamma_2-1)\,\sin^2(\varphi)\\
&=1+(\gamma+1)-(\gamma_1-1)(\gamma_2-1)\,\sin^2(\varphi)\,.
\end{split}
\end{equation}
where, as in \eqref{eq:ProductBoosts-4c},
$\varphi$ denotes the angle between the 
velocities, i.e. 
\begin{equation}
\label{eq:ProductBoosts-9}
\cos(\varphi):=\vec n_1\cdot\vec n_2\,.
\end{equation}
Hence \eqref{eq:ProductBoosts-7} becomes
\begin{equation}
\label{eq:ProductBoosts-10}
\cos(\theta)
=1-\frac{(\gamma_1-1)(\gamma_2-1)\sin^2(\varphi)}{1+\gamma_1\gamma_2+\sqrt{(\gamma_1^2-1)(\gamma_2^2-1)}\ \cos(\varphi)}\,.
\end{equation}
We have deliberately written the denominator in 
the second term, which is just $1+\gamma$, in 
terms of $\gamma_1$, $\gamma_2$, and $\varphi$,
so as to explicitly display $\theta$ as function 
of the moduli of the two boost velocities (i.e. 
their $\gamma$-factors) and the angle $\varphi$ between them.  

Alternatively, instead of 
$(\gamma_1,\gamma_2,\varphi)$ we may express 
$\cos(\theta)$ as function of $(\gamma_1,\gamma_2,\gamma)$. This is achieved by 
writing $\sin^2(\varphi)=1-\cos^2(\varphi)$
and replacing all occurrences of $\cos(\varphi)$ in \eqref{eq:ProductBoosts-10} with the expression 
that follows from  \eqref{eq:ProductBoosts-4c}:
\begin{equation}
\label{eq:ProductBoosts-11}
\cos(\varphi)=\frac{\gamma-\gamma_1\gamma_2}{\sqrt{(\gamma^2_1-1)(\gamma^2_2-1)}}\,.
\end{equation}
This gives after a few elementary steps
\begin{equation}
\label{eq:ProductBoosts-12}
\cos(\theta)=\frac{(1+\gamma+\gamma_1+\gamma_2)^2}{(1+\gamma)(1+\gamma_1)(1+\gamma_2)}-1\,.
\end{equation}
This expression is remarkable for its simple 
structure and permutation symmetry in 
$\{\gamma,\gamma_1,\gamma_2\}$.
An alternative way to write it is in terms 
of $\cos(\theta/2)=\sqrt{(\cos(\theta)+1)/2}$ and 
$\cosh(\rho):=\gamma$, so that 
$(1+\gamma)=2\,\cosh^2(\rho/2)$, as well as the 
corresponding equations for $\gamma_i$ in 
terms of $\rho_i$ ($i=1,2$)\footnote{The greek letter $\rho$ is chosen because this quantity 
is usually referred to as ``rapidity'', 
introduced by \citet{Robb:OptGeom}.} 
\begin{equation}
\label{eq:ProductBoosts-13}
\cos(\theta/2)=\frac{1+\cosh(\rho)+\cosh(\rho_1)+\cosh(\rho_2)}{4\,\cosh(\rho/2)\cosh(\rho_1/2)\cosh(\rho_2/2)}\,.
\end{equation}
These formulae, as as well as \eqref{eq:ProductBoosts-10}, are well 
known in the literature; an early 
appearance of \eqref{eq:ProductBoosts-13} 
being \citep[formula (124)]{Macfarlane:1962}. 
Elementary derivations are often claimed to 
be unduly tedious or even a ``herculean'' 
task'' \citep[p.\,28]{Ungar:BeyondEinstein}, 
and more elegant ways using Clifford algebras
have been given \citep{Urbantke:1990}. 
However, we believe that the derivation given 
here is sufficiently easy for presentation 
in, say, a basic lecture on SR. 

\begin{figure}
\noindent
\centering
\includegraphics[width=0.45\linewidth]{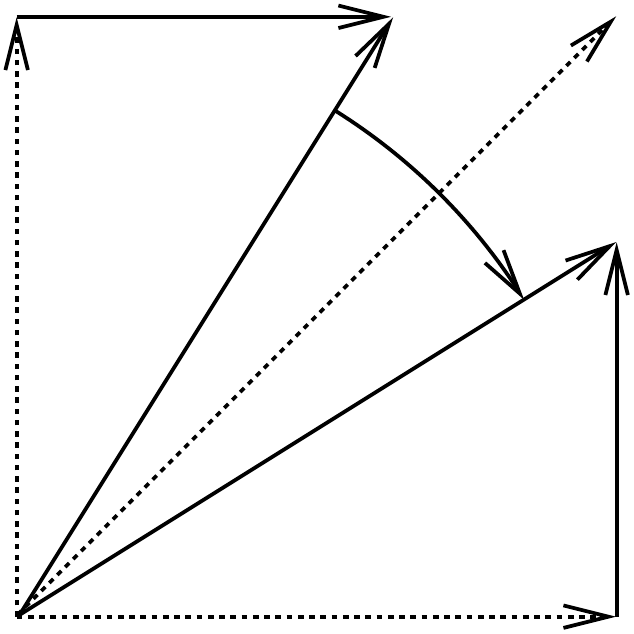}
\put(-17,-10){$\vec\beta_1$}
\put(9,33){\begin{rotate}{90}$\gamma^{-1}\vec\beta_2$\end{rotate}}
\put(-174,150){$\vec\beta_2$}
\put(-130,165){$\gamma^{-1}\vec\beta_1$}
\put(-80,39){\begin{rotate}{31}$\vec\beta_1\vplus\vec\beta_2$\end{rotate}}
\put(-122,74){\begin{rotate}{59}$\vec\beta_2\vplus\vec\beta_1$\end{rotate}}
\put(-20,162){$\vec\beta_1+\vec\beta_2$}
\put(-45,85){$\theta$}
\caption{\label{fig:Fig-ThomasRotation}
\small Non commutativity of velocity addition 
due to Thomas rotation. Shown ist the addition 
of two perpendicular velocities of equal 
magnitude $\beta=0.8$ corresponding to  
$\gamma^{-1}=0.6$. The Thomas rotation
$\mat{T}(\vec\beta_1,\vec\beta_2)$ rotates by an
positive angle $\theta$ in the oriented plane 
spanned by the ordered pair $\{\vec a,\vec b\}$.
Since $\vec a$ is proportional to 
$\vec\beta_2\vplus\vec\beta_1$ and $\vec b$
to $\vec\beta_1\vplus\vec\beta_2$, this 
rotation is in the clockwise -- i.e. negative
-- orientation with respect to the 
ordered pair $\{\vec\beta_1,\vec\beta_2\}$.}
\end{figure}

It is important to note that the right hand side 
of \eqref{eq:ProductBoosts-12} depends on the 
reference state which respect to which the 
$\vec\beta$'s and hence the $\gamma$'s are 
defined. In fact, the $(-\gamma)$'s are just the 
Minkowskian scalar products between the corresponding states of motion with the 
reference state. In this way  \eqref{eq:ProductBoosts-12} can be written as 
rational function of scalar products of states 
of motion with the reference state, thereby 
displaying the dependence on the latter in a 
basis independent way. We will later adopt 
precisely that strategy: to reduce all 
statements to functions of scalar products 
of states or motion.

\subsection{On the magnitude of Thomas rotation}
\label{sec:MoreThomasRotation}
In Fig.\,\ref{fig:Fig-ThomasRotation} we 
show the Thomas rotation for the special 
case of mutually perpendicular velocities of 
equal magnitude; i.e. $\cos(\varphi)=0$
and $\gamma_1=\gamma_2=\gamma_*$. In this 
case \eqref{eq:ProductBoosts-10} reduces to 
\begin{equation}
\label{eq:ThomasRot-Example}
\cos(\theta)=\frac{2\gamma_*}{1+\gamma^2_*}\,,
\end{equation}
which is a monotonically decreasing function 
of $\gamma_*\in(1,\infty)$ ranging from $1$ 
to zero and corresponding to a monotonically 
increasing angle from $0$ to $\pi/2$. 

More generally, we may ask for the angle 
$\varphi$ at which $\theta$ is largest for 
given $\gamma_{1{,}2}$. This can be answered 
using either expression 
\eqref{eq:ProductBoosts-10} or 
\eqref{eq:ProductBoosts-12}, if in the latter 
case we regard $\gamma$ again as function of
$\gamma_{1{,}2}$ and $\varphi$ according to 
\eqref{eq:ProductBoosts-4c}. This latter 
possibility turns out to be slightly more 
convenient, so let's follow this strategy and 
consider \eqref{eq:ProductBoosts-12}. Since 
$\cos$ is monotonically decreasing in 
$[0,\pi]$, a maximum of $\theta$ corresponds 
to a minimum of $\cos(\theta)$, which we now 
seek. The only dependence of $\cos(\theta)$
on $\varphi$ is through $\gamma$. The 
$\varphi$ derivative of the latter is 
$\gamma':=d\gamma/d\varphi=-\beta_1\gamma_1\beta_2\gamma_2\sin(\varphi)$,
and $d\cos(\theta)/d\varphi=
\bigl(d\cos(\theta)/d\gamma\bigr)\gamma'$.
Hence stationary points exist either for 
$\gamma'=0$, which is for the boundary values 
$\varphi=0$ or $\varphi=\pi$ that correspond 
to  aligned and anti-alingned velocities and 
which correspond to $\theta=0$ (hence a minimum), 
or for $d\cos(\varphi)/d\gamma=0$, containing 
the maxima. In view of \eqref{eq:ProductBoosts-12}, 
the latter equation is equivalent to 
\begin{equation}
\label{eq:ThomasAngleMax-1}
\frac{d}{d\gamma}\,
\frac{(1+\gamma+\gamma_1+\gamma_2)^2}{(1+\gamma)}
=\frac{1+\gamma+\gamma_1+\gamma_2}{(1+\gamma)^2}\,
(1+\gamma-\gamma_1-\gamma_2)=0\,.
\end{equation}
Since all $\gamma$'s are positive, this is 
in turn equivalent to 
\begin{equation}
\label{eq:ThomasAngleMax-2}
\gamma=\gamma_{\rm max}:=\gamma_1+\gamma_2-1
\end{equation}
which, according to \eqref{eq:ProductBoosts-11},
corresponds to an angle $\varphi_{\rm max}$ between 
the velocities that satisfies
\begin{equation}
\label{eq:ThomasAngleMax-3}
\cos(\varphi_{\rm max})=\frac{\gamma_{\rm max}-\gamma_1\gamma_2}{\sqrt{(\gamma^2_1-1)(\gamma^2_2-1)}}
=-\ \sqrt{\frac{(\gamma_1-1)(\gamma_2-1)}{(\gamma_1+1)(\gamma_2+1)}}\,.
\end{equation}
The fact that this is negative means that 
$\varphi_{\rm max}>\pi/2$. This means that the maximal 
Thomas angle is only obtained for obtuse angles 
between the velocities. The value of the maximal 
Thomas angle is given by inserting 
\eqref{eq:ThomasAngleMax-2} into \eqref{eq:ProductBoosts-12}:
\begin{equation}
\begin{split}
\label{eq:ThomasAngleMax-4}
\cos(\theta_{\max})
=&\ \frac{(1+\gamma_{\rm max}+\gamma_1+\gamma_2)^2}%
  {(1+\gamma_{\rm max})(1+\gamma_1)(1+\gamma_2)}-1\\
=&\ \frac{3(\gamma_1+\gamma_2)-\gamma_1\gamma_2-1}
     {(1+\gamma_1)(1+\gamma_2)}\\
=&\ 1-2\,\frac{(\gamma_1-1)(\gamma_2-1)}{(\gamma_1+1)(\gamma_2+1)}\\
&=\ 1-2\cos^2(\varphi_{\rm max})\\
&=\ -\ \cos(2\varphi_{\rm max})\,.
\end{split}
\end{equation}
From this we infer 
\begin{equation}
\label{eq:ThomasAngleMax-5}
\cos(\theta_{\rm max})+\cos(2\varphi_{\rm max})
=2\,
\cos\left(\frac{\theta_{\rm max}+2\varphi_{\rm max}}{2}\right)
\cos\left(\frac{\theta_{\rm max}-2\varphi_{\rm max}}{2}\right)
=0\,,
\end{equation}
which in turn implies that either 
$\theta_{\rm max}+2\varphi_{\rm max}$ or 
$\theta_{\rm max}-2\varphi_{\rm max}$ is 
an odd-integer multiple of $\pi$. But 
$\varphi_{\rm max}$ is obtuse, i.e. between 
$\pi/2$ and $\pi$, as follows from  \eqref{eq:ThomasAngleMax-3}. If either 
$\gamma_1$ and/or $\gamma_2$ approach the 
value $1$ from above $\cos(\varphi_{\rm max})$
approaches the value zero from below and hence 
$\varphi_{\rm max}$ approaches $\pi/2$ from 
above. In this case $\theta_{\rm max}$ 
approaches zero from below if we refer both 
angles to the orientation given to the 
plane of rotation by the ordered pair 
$\{\vec\beta_1,\vec\beta_2\}$ 
(compare Fig.\,\ref{fig:Fig-ThomasRotation}).
This shows that with that convention our unique solution to \eqref{eq:ThomasAngleMax-5} is  
given by 
\begin{equation}
\label{eq:ThomasAngleMax-6}
\theta_{\rm max}=\pi-2\varphi_{\rm max}\,.
\end{equation}
So $\theta_{\rm max}$ varies between $0$ 
and $-\pi$ and its modulus exceeds the right 
angle $-\pi/2$ if $\cos(\theta_{\rm max})$
turns negative. This, according to the third
equation in \eqref{eq:ThomasAngleMax-4}, is the 
case if  
\begin{equation}
\label{eq:ThomasAngleMax-7}
\frac{(\gamma_1-1)(\gamma_2-1)}{(\gamma_1+1)(\gamma_2+1)}>\frac{1}{2}\,.	
\end{equation}
For equal velocities, i.e. 
$\gamma_1=\gamma_2=\gamma_*=1/\sqrt{1-\beta^2_*}$
this happens for
\begin{equation}
\label{eq:ThomasAngleMax-8}
\gamma_*>\frac{\sqrt{2}+1}{\sqrt{2}-1}\approx 5.828
\end{equation}
corresponding to 
\begin{equation}
\label{eq:ThomasAngleMax-9}
\beta_*>\frac{2^{5/4}}{2^{1/2}+1}
\approx 0.985\,,
\end{equation}
that is, 98.5\% of the velocity of light.

\subsection{On the algebraic structure 
of Einstein addition}
\label{sec:AlgebraicStructure}
Using Definition\,\ref{def:ThomasRotation} we shall now write 
\begin{equation}
\label{eq:AlgStr-1}
B(\vec\beta_1)B(\vec\beta_2)
=B(\vec\beta_1\vplus\vec\beta_2)\, 
 R\bigl(\mat{T}(\vec\beta_1,\vec\beta_2)\big)\,.
\end{equation}
Taking the transposed of that equation
and using the symmetry of $B$, the orthogonality 
of $R$, and equivariance property  \eqref{eq:ProofPolarDec-6}, we 
get\footnote{Historically it is interesting 
to note that this formula has already been 
written down (in different but easy-to-translate 
notation) by 
\citet[p.\,169, formula\,(7)]{Silberstein:1914}.
Without explicitly relating it to any 
orthogonal transformation, mere non-commutativity has been noted from the very beginning; e.g.,  
by \citet[pp.\,905-6]{Einstein-SRT:1905}
and \citet[p.\,44]{Laue-SRT:1911}. Einstein 
remarked (p.\,905) that the parallelogram-law 
for velocity addition is only approximately 
valid and at the same time finds it ``remarkable'' 
(p.\,906) that the expression for the modulus 
of the composed velocities is a symmetric 
function in these velocities. Laue, after 
writing down the law of velocity addition, 
states: ``The strangest (merkw\"urdigste) 
thing is that the two velocitied to be added do not enter 
on equal footing (gleichberechtigt)''. Laue's
text also contains an explicit discussion of 
the special case of orthogonal velocities and 
a figure equivalent to our 
Fig.\,\ref{fig:Fig-ThomasRotation}
\citep[p.\,44, Fig.\,4]{Laue-SRT:1911}.}
\begin{equation}
\label{eq:AlgStr-2}
\vec\beta_1\vplus\vec\beta_2=\mat{T}(\vec\beta_1,\vec\beta_2)\bigr(\vec\beta_2\vplus\vec\beta_1\bigl)\,,
\end{equation}
showing non-commutativity for non-collinear 
velocities. That we already know from 
$\vec b=\mat{D}\vec a$ -- see line above \eqref{eq:ProductBoosts-5} -- and
\eqref{eq:ProductBoosts-4}. An immediate 
consequence of \eqref{eq:AlgStr-2} is 
\begin{equation}
\label{eq:AlgStr-3}
\mat{T}(\vec\beta_2,\vec\beta_1)=
\mat{T}^{-1}(\vec\beta_1,\vec\beta_2)
\end{equation}
We also note the following: The vector 
$\vec\beta_1\vplus\vec\beta_2$ is a linear combination of $\vec n_1$ and 
$\vec n_2$ with coefficients that only involve 
scalar products of these vectors; that is, the coefficients are invariant under 
$(\vec n_1,\vec n_2)\rightarrow (-\vec n_1,-\vec n_2)$ (even functions). The rotation matrix $T(\vec\beta_1,\vec\beta_2)$ is a linear 
combination of the identity and terms 
proportional to $\vec n_a\otimes\vec n_b^\top$
($a,b\in\{1,2\}$) with coefficients also 
depending only on the scalar products of 
these vectors. This implies 
\begin{subequations}
\label{eq:AlgStr-4}
\begin{alignat}{1}
\label{eq:AlgStr-4a}
(-\vec\beta_1)\vplus(-\vec\beta_2)
&\,=\,-\,(\vec\beta_1\vplus\vec\beta_2)\,,\\
\label{eq:AlgStr-4b}
\mat{T}(-\vec\beta_1\,,\,-\vec\beta_2)
&\,=\,\mat{T}(\vec\beta_1,\vec\beta_2)\,.
\end{alignat}	
\end{subequations}
Other immediate consequences of the same 
remark are:
\begin{subequations}
\label{eq:AlgStr-5}
\begin{alignat}{1}
\label{eq:AlgStr-5a}
(\mat{D}\vec\beta_1)\vplus(\mat{D}\vec\beta_2)
&\,=\,\mat{D}(\vec\beta_1\vplus\vec\beta_2)\,,\\
\label{eq:AlgStr-5b}
\mat{T}[\mat{D}\vec\beta_1,\mat{D}\vec\beta_2]
&\,=\,\mat{D}\,\mat{T}[\vec\beta_1,\vec\beta_2]\mat{D}^{-1}\,,
\end{alignat}	
\end{subequations}
for any $\mat{D}\in\group{SO}(3)$.

Now, a general Lorentz transformation 
is written as 
\begin{equation}
\label{eq:GenLT}
L(\vec\beta,\mat{D}):=B(\vec\beta)R(\mat{D})\,.
\end{equation}   
The composition of two such transformations is 
then given by 
\begin{equation}
\label{eq:CompositionLT-1}
\begin{split}
L(\vec\beta,\mat{D})
&=L(\vec\beta_1,\mat{D}_1)\,L(\vec\beta_2,\mat{D}_2)\\
&=B(\vec\beta_1)\,R(\mat{D}_1)\,B(\vec\beta_2)\,R(\mat{D}_2)\\
&=B(\vec\beta_1)\,R(\mat{D}_1)\,B(\vec\beta_2)\,
  R(\mat{D}^{-1}_1)\,R(\mat{D}_1)\,R(\mat{D}_2)\\
&=B(\vec\beta_1)\,B(\mat{D_1}\vec\beta_2)\,R(\mat{D}_1\mat{D}_2)\\
&=B(\vec\beta_1\vplus\mat{D}_1\vec\beta_2)\,R\bigl(\mat{T}[\vec\beta_1,\mat{D}_1\vec\beta_2]\,\mat{D}_1\,\mat{D}_2\bigr)\,,
\end{split}
\end{equation} 
so that 
\begin{subequations}
\label{eq:CompositionLT-2}
\begin{alignat}{1}
\label{eq:CompositionLT-2a}
\vec\beta
&\,=\,\vec\beta_1\vplus\mat{D}_1\vec\beta_2\,,\\
\label{eq:CompositionLT-2b}
\mat{D}
&\,=\,\mat{T}[\vec\beta_1,\vec\beta_2]\,\mat{D}_1\,\mat{D}_2\,.
\end{alignat}
\end{subequations}
This should be compared with the group 
composition law of the semi-direct product 
$\reals^3\rtimes\group{SO}(3)$, which would be 
given by 
$\vec\beta=\vec\beta_1+\mat{D}_1\vec\beta_2$
and $\mat{D}=\mat{D}_1\,\mat{D}_2$, i.e. $\vplus$ 
replaced by $+$ and the Thomas rotation always 
the identity, as it it is the case for the 
Galilei group (with $\vec\beta$ replaced by $\vec v$). 

In order to deduce the parameters for the inverse
transformation $L^{-1}(\vec\beta,\mat{D})$ we first 
note that $B^{-1}(\vec\beta)=B(-\vec\beta)$, as is already obvious from \eqref{eq:EmbedBoost}.
Taking the inverse of both sides of 
\eqref{eq:GenLT} and using \eqref{eq:ProofPolarDec-6}
then gives 
\begin{equation}
\label{eq:AlgStr-6}
L^{-1}(\vec\beta,\mat{D})
:= R(\mat{D}^{-1})B(-\vec\beta)
= B(-\mat{D}^{-1}\vec\beta)\,R(\mat{D}^{-1})\\
=L(-\mat{D}^{-1}\vec\beta,\mat{D}^{-1})\,
\end{equation} 
Again this can be compared with the inverse 
of a semi-direct product $\reals^3\rtimes\group{SO}(3)$, which would be given by just the same 
formula.   

The Thomas rotation, which we already 
identified in \eqref{eq:AlgStr-2} as 
directly responsible for the non-commutativity 
of velocity composition, is also responsible 
for other remarkable properties.  One of them 
is the fact that the inverse of a composed 
velocity differs from the transposed 
composition of the inverse velocities.  
In fact, the inverse of the combination 
$B(\vec\beta_1)B(\vec\beta_2)$ is obviously 
$B(-\vec\beta_2)B(-\vec\beta_1)$, the polar decomposition of which will then,
contain a boost with velocity 
$(-\vec\beta_2)\vplus(-\vec\beta_1)$. 
The latter differs from the inverse velocity 
$-(\vec\beta_1\vplus\vec\beta_2)$ of the original
combination $B(\vec\beta_1)B(\vec\beta_2)$ by a
Thomas rotation:
\begin{equation}
\label{eq:MocanuParadox}
\begin{split}
-(\vec\beta_1\vplus\vec\beta_2)
&=(-\vec\beta_1)\vplus(-\vec\beta_2)\\
&=\mat{T}(-\vec\beta_1,-\vec\beta_2)\bigl((-\vec\beta_2)\vplus(-\vec\beta_1)\bigr)\\
&=\mat{T}(\vec\beta_1,\vec\beta_2)
  \bigl((-\vec\beta_2)\vplus(-\vec\beta_1)\bigr)\\
& \ne (-\vec\beta_2)\vplus(-\vec\beta_1)\,,
\end{split}
\end{equation}
Here we used \eqref{eq:AlgStr-4a} in the first,
\eqref{eq:AlgStr-2} in the second, and \eqref{eq:AlgStr-4b} in the third equality.
This has often been considered paradoxical, following \citet{Mocanu:1986}. It is  
known in the literature as ``Mocanu Paradox'' \citep{Ungar:1989}.

Another consequence of the Thomas rotation is the 
failure of associativity of Einstein 
addition.\footnote{The first to explicitly 
note non-associativity and the r\^ole of Thomas 
rotation in it seems to have been \citet{Ungar:1989}.} 
To see this let us start from associativity of 
composition of Lorentz transformations with three
pure boosts:
\begin{equation}
\label{eq:AlgStr-7}
B(\vec\beta_1)
\bigl(B(\vec\beta_2)
B(\vec\beta_3)\bigr)
=
\bigl(B(\vec\beta_1)
B(\vec\beta_2)\bigr)
B(\vec\beta_3)	
\end{equation}
Applying polar decomposition 
\eqref{eq:ProductBoosts-1} with 
$\vec\beta=\vec\beta_1\vplus\vec\beta_2$
and $\mat{D}=\mat{T}(\vec\beta_1,\vec\beta_2)$ 
on each side gives for the left-hand side
\begin{subequations}
\label{eq:AlgStr-8}
\begin{equation}
\label{eq:AlgStr-8a}
\begin{split}
& B(\vec\beta_1)\,B(\vec\beta_2\vplus\vec\beta_3)\,
R\bigl(\mat{T}(\vec\beta_2,\vec\beta_3)\bigr)\\
=&\ 
B\bigl(\vec\beta_1\vplus(\vec\beta_2\vplus\vec\beta_3)
\bigr)\ R\bigl(\mat{T}(\vec\beta_1,\vec\beta_2\vplus\vec\beta_3)\bigr)
R\bigl(\mat{T}(\vec\beta_2,\vec\beta_3)\bigr)\,.
\end{split}
\end{equation}
and, slightly more complicated, for the 
right-hand side
\begin{equation}
\label{eq:AlgStr-8b}
\begin{split}
& B(\vec\beta_1\vplus\vec\beta_2)\,
R\bigl(\mat{T}(\vec\beta_1,\vec\beta_2)\bigr)\,
B(\vec\beta_3)\\
=\ &
B(\vec\beta_1\vplus\vec\beta_2)\,
B\bigl(\mat{T}(\vec\beta_1,\vec\beta_2)\vec\beta_3\bigr)\,
R\bigl(\mat{T}(\vec\beta_1,\vec\beta_2)\bigr)\\
= \ & B\bigl((\vec\beta_1\vplus\vec\beta_2)\vplus \mat{T}(\vec\beta_1,\vec\beta_2)\vec\beta_3\bigr)\
R\bigl(\mat{T}(\vec\beta_1\vplus\vec\beta_2\,,\,\mat{T}(\vec\beta_1,\vec\beta_2)\vec\beta_3)\bigr)
R\bigl(\mat{T}(\vec\beta_1,\vec\beta_2)\bigr)\,,
\end{split}
\end{equation}
\end{subequations}
where we again used \eqref{eq:ProofPolarDec-6}
from the first to the second line.
As both sides are in polar decomposed form,
boost and rotation parts must separately be 
equal, leading for the boosts to 
\begin{subequations}
\label{eq:AlgStr-9}
\begin{alignat}{1}
\label{eq:AlgStr-9a}
\vec\beta_1\vplus(\vec\beta_2\vplus\vec\beta_3)
&\,=\,
(\vec\beta_1\vplus\vec\beta_2)\vplus \mat{T}(\vec\beta_1,\vec\beta_2)\vec\beta_3	\,,\\
\label{eq:AlgStr-9b}
(\vec\beta_1\vplus\vec\beta_2)\vplus\vec\beta_3
&\,=\,
\vec\beta_1\vplus\bigl(\vec\beta_2\vplus \mat{T}(\vec\beta_2,\vec\beta_1)\vec\beta_3\bigr)\,,
\end{alignat}
\end{subequations}
where \eqref{eq:AlgStr-9b} follows immediately 
from \eqref{eq:AlgStr-9a} by setting 
$\vec\beta'_3:=\mat{T}(\vec\beta_1,\vec\beta_2)\vec\beta_3$, so 
that, according to \eqref{eq:AlgStr-3}, 
$\vec\beta_3:=\mat{T}(\vec\beta_2,\vec\beta_1)\vec\beta'_3$. 
Dropping the prime on $\vec\beta'_3$ then 
gives \eqref{eq:AlgStr-9b}. 

Equations \eqref{eq:AlgStr-9} show explicitly 
how the existence of the Thomas precession 
obstructs associativity. Formula 
\eqref{eq:AlgStr-9a} and  \eqref{eq:AlgStr-9b}
are identical with the ``right weak associative 
law of velocity composition'' and the ``left 
weak associative law of velocity composition'',
respectively, stated by \citet[p.\,71, expression iia,b]{Ungar:1988}. There it is also stated that 
the proof of such identities ``is lengthy and, hence, requires the use of computer algebra'' \citep[p.\,72]{Ungar:1988}. As we have just 
seen, this is an exaggeration. 

Another interesting and immediate consequence 
from \eqref{eq:AlgStr-7} is obtained by 
specialising to $\vec\beta_1=\vec\beta_3$. In 
this case the triple product is $B(\vec\beta_1)B(\vec\beta_2)B(\vec\beta_1)$, which is 
symmetric and positive definite, hence already 
polar decomposed. Therefore the rotational 
part on the right-hand side of \eqref{eq:AlgStr-8a} must be the identity, which leads to 
\begin{equation}
\label{eq:AlgStr-10}
\mat{T}(\beta_1,\beta_2)
=\mat{T}(\beta_1,\beta_2\vplus\beta_1)
=\mat{T}(\beta_1\vplus\beta_2,\beta_2)\,,
\end{equation}
where we also used \eqref{eq:AlgStr-3}
and the second equation follows from the first 
by inversion and exchange the indices 1 and 2.

Even though Einstein addition fails 
associativity, it is  remarkable that it 
does maintains a property that is usually 
implied by it. To explain this, let us first 
make the obvious observation that, like in 
ordinary vector addition, the neutral element 
for Einstein addition is still the zero 
velocity and the unique left- and right 
inverse of $\vec\beta$ is $(-\beta)$. Hence for 
all $\vec\beta$ we have
\begin{equation}
\label{eq:AlgStr-11}
\vec\beta\vplus\vec 0=\vec 0\vplus\vec\beta=\vec\beta	
\end{equation}
and 
\begin{equation}
\label{eq:AlgStr-12}
\vec\beta\vplus(-\vec\beta)=(-\vec\beta)\vplus\vec\beta=\vec 0\,.
\end{equation}
Assume for a moment that associativity did 
hold. We could then uniquely solve an equation like
\begin{equation}
\label{eq:AlgStr-13}
\vec\beta_1\vplus\vec\beta_2=\vec\beta_3
\end{equation}
for $\vec\beta_1$ given $\vec\beta_2$ and 
$\vec\beta_3$, or for  $\vec\beta_2$ given 
$\vec\beta_1$ and $\vec\beta_3$. The way to 
achieve this would in the first case be to 
$\vplus$-multiply \eqref{eq:AlgStr-13} from 
the right with $(-\vec\beta_2)$ and then 
use associativity to show that the left-hand 
side is just $\vec\beta_1$ whereas the 
right-hand side is 
$(\vec\beta_3)\vplus(-\vec\beta_2)$. 
Alternatively, left $\vplus$-multiplication 
with $-\vec\beta_1$ would determine 
$\vec\beta_2$ as 
$(-\vec\beta_1)\vplus\vec\beta_3$. Now, in 
reality, we do not have associativity and we 
cannot proceed in this way. But -- and that 
is a remarkable fact --  we can still write 
down explicit expressions solving 
\eqref{eq:AlgStr-13} for either $\beta_1$ 
or $\beta_2$. Moreover, at least for 
$\beta_2$, the expression is just that we 
would have derived on account of associativity, 
as just discussed. We have 
\begin{theorem}
\label{thm:CompositionSolvability}
The unique solutions to \eqref{eq:AlgStr-13} 
are
\begin{subequations}
\label{eq:AlgStr-14}	
\begin{alignat}{1}
\label{eq:AlgStr-14a}	
\vec\beta_1
&\,=\, \vec\beta_3\vplus 
    \bigl(
    -\mat{T} (\vec\beta_3,\vec\beta_2)
    \vec\beta_2
    \bigr)\,,\\
\label{eq:AlgStr-14b}	
\vec\beta_2
&\,=\, (-\vec\beta_1)\vplus\vec\beta_3\,.
\end{alignat}
\end{subequations}
\end{theorem}
\begin{proof}
The proof of \eqref{eq:AlgStr-14b}
is just given by left $\vplus$-multiplication 
with $(-\vec\beta_1)$. We now use formula \eqref{eq:AlgStr-9} in which 
we replace $\vec\beta_1$ with $(-\vec\beta_1)$,
$\vec\beta_2$ with $\vec\beta_1$, and 
$\vec\beta_3$ with $\vec\beta_2$.
The Thomas term $\mat{T}(\vec\beta_1,\vec\beta_2)$
then turns into $\mat{T}(-\vec\beta_1,\vec\beta_1)$
which is the identity. Hence, in this special 
case, we may proceed as if associativity holds 
and get \eqref{eq:AlgStr-14b}.
For the proof of \eqref{eq:AlgStr-14a} we have to 
go a little further and start from the general relation 
\begin{subequations}
\label{eq:AlgStr-15}
\begin{equation}
\label{eq:AlgStr-15a}
L(\vec\beta_1,\mat{D}_1)\,
L(\vec\beta_2,\mat{D}_2)=
L(\vec\beta_3,\mat{D}_3)\,,	
\end{equation}
which reads in terms of parameters, according to 
\eqref{eq:CompositionLT-2},
\begin{alignat}{1}
\label{eq:AlgStr-15b}	
\vec\beta_3&\,=\,\vec\beta_1\vplus \mat{D}_1\vec\beta_2\,,\\
\label{eq:AlgStr-15c}
\mat{D}_3&\,=\,T(\vec\beta_1,\mat{D}_1\vec\beta_2)\mat{D}_1\mat{D}_2\,.
\end{alignat}
\end{subequations}
On the group level we know how to solve 
\eqref{eq:AlgStr-15a}
for $L(\vec\beta_1,\mat{D}_1)$ through right-multiplication 
with $L^{-1}(\vec\beta_2,\mat{D}_2)$. From 
\eqref{eq:AlgStr-6}  and 
\eqref{eq:CompositionLT-2} we also know the 
respective parameter expressions for inversion 
and multiplication. Hence we get 
\begin{subequations}
\label{eq:AlgStr-16}
\begin{equation}
\label{eq:AlgStr-16a}
\begin{split}
L(\vec\beta_1,\mat{D}_1)
&=L(\vec\beta_3,\mat{D}_3)L(-\mat{D}^{-1}_2\vec\beta_2\,,\,\mat{D}_2^{-1})\\
&=L\bigl(\vec\beta_3\vplus(-\mat{D}_3\mat{D}_2^{-1}\vec\beta_2)\,,\,\mat{T}(\vec\beta_3,-\mat{D}_2^{-1}\vec\beta_2)\mat{D}_3\mat{D}_2^{-1}\bigr)\,,
\end{split}
\end{equation}
so that for the parameters we get
\begin{alignat}{1}
\label{eq:AlgStr-16b}	
\vec\beta_1&\,=\,\vec\beta_3\vplus 
(-\mat{D}_3\mat{D}_2^{-1}\vec\beta_2)\,,\\
\label{eq:AlgStr-16c}
\mat{D}_1&\,=\,T(\vec\beta_3,-\mat{D}_3\mat{D}_2^{-1}\vec\beta_2)\mat{D}_3\mat{D}^{-1}_2\,.
\end{alignat}
\end{subequations}  
We note that the equations \eqref{eq:AlgStr-15}
and \eqref{eq:AlgStr-16} are equivalent sets and 
valid for all $(\vec\beta_i,\mat{D}_i)$, $i=1,2,3$.
Next we observe that \eqref{eq:AlgStr-13}
is just \eqref{eq:AlgStr-15b} for the special 
case in which $\mat{D}_1$ is the identity. Hence 
we consider \eqref{eq:AlgStr-15} and \eqref{eq:AlgStr-16} for $\mat{D}_1$ being the 
identity. In 
that case \eqref{eq:AlgStr-15c} becomes  
\begin{equation}
\label{eq:BoostCompositionSolvability-4}
\begin{split}
\mat{D}_3\mat{D}^{-1}_2
&=\mat{T}(\vec\beta_1,\vec\beta_2)\\
&=\mat{T}(\vec\beta_1\vplus\vec\beta_2,\vec\beta_2)
\quad\text{(because of \eqref{eq:AlgStr-10})}\\
&=\mat{T}(\vec\beta_3,\vec\beta_2)
\quad\,\ \qquad\text{(because of \eqref{eq:AlgStr-13})}
\,.	
\end{split}	
\end{equation}
Inserting this into \eqref{eq:AlgStr-16b} proves
\eqref{eq:AlgStr-14a}.	
\end{proof}   

We can summarise the algebraic structure 
realised by $\vplus$ on the open ball 
$\mathring{B}_1(\reals^3)$ in modern 
terminology as follows; see also 
Fig.\,\ref{fig:Fig-AlgStr}:
\begin{itemize}
\item
Let $M$ be a set and $\phi:M\times M$ a map,
denoted by $(a,b)\mapsto\phi(a,b)=:a\cdot b$. 
Then the pair $(M,\phi)$ is called a 
\emph{magma} (sometimes also 
\emph{groupoid}). Hence 
$\bigl(\mathring{B}_1(\reals^3),\vplus\bigr)$ 
is a magma. 
\item
A magma $(M,\phi)$ is called a 
\emph{semigroup} if $\phi$ is associative, 
i.e. $a\cdot (b\cdot c)=(a\cdot b)\cdot c$
for all $a,b,c\in M$. Hence 
$\bigl(\mathring{B}_1(\reals^3),\vplus\bigr)$ 
is \emph{not} a semigroup.
\item
A magma $(M,\phi)$ is called a 
\emph{quasigroup} if for each pair 
$(a,b)\in M\times M$ there is a unique pair 
$(x,y)\in M\times M$ such that $a\cdot x=b$
and $y\cdot a=b$. This property is also 
called \emph{divisibility}. 
Theorem\,\ref{thm:CompositionSolvability}
then shows that $\bigl(\mathring{B}_1(\reals^3),\vplus\bigr)$ is a quasigroup.
\item
A magma $(M,\phi)$ is said to have an  \emph{identity} if there exists an $e\in M$ 
such that $e\cdot a=a\cdot e=a$ for all 
$a\in M$. Such an $e$ is necessarily 
unique (even if  $(M,\phi)$ is not a 
quasigroup). As for Einstein addition 
$\vec 0$ is an identity (compare \eqref{eq:AlgStr-11}), we infer that 
$\bigl(\mathring{B}_1(\reals^3),\vplus\bigr)$ 
is a \emph{quasigroup with identity}.
\item
In a quasigroup with identity each element
$a$ has a unique left and a unique right 
inverse, namely that $x$ and that $y$ satsifying
$x\cdot a=e$ and $a\cdot y=e$. $x$ and $y$ 
need not be identical. If they are, we write 
$x=y=:a^{-1}$. From \eqref{eq:AlgStr-12} we 
see that 
$\bigl(\mathring{B}_1(\reals^3),\vplus\bigr)$ 
is a \emph{quasigroup with identity and 
coinciding left-right inverses}. 
\item
A quasigroup with identity is called a 
\emph{loop}. An associative loop is a group. 
Hence $\bigl(\mathring{B}_1(\reals^3),\vplus\bigr)$ is a loop that is not a group, but is special insofar as left-and right-inverse elements 
coincide.   
\end{itemize}
\begin{figure}[ht]
\centering
\includegraphics[width=0.5\linewidth]{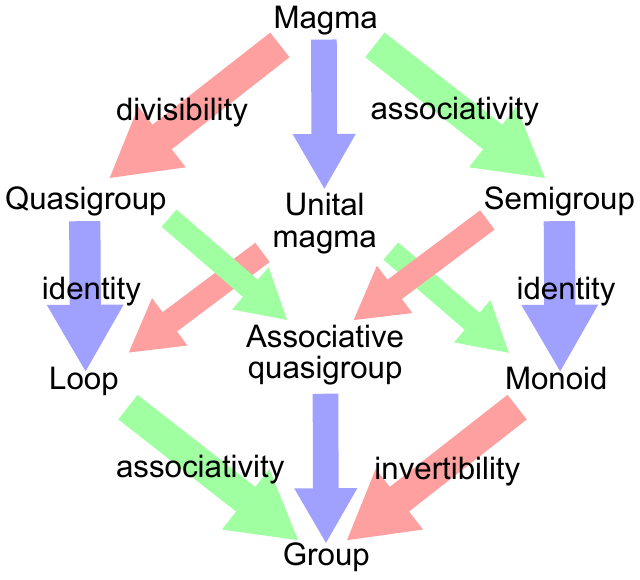}
\caption{\label{fig:Fig-AlgStr}
\small Hierarchy of algebraic structures. 
Einstein addition $\vplus$ endows the open 
ball $\mathring{B}_1(\reals^3)$ with the 
structure of a loop, which is just short of 
being a group by its failure to satisfy associativity. 
(Picture source \url{https://commons.wikimedia.org/wiki/File:Magma_to_group3.svg}. Picture attribution: 
Tomruen, CC0, via Wikimedia Commons)}  
\end{figure}
\newpage

\section{Velocity subtraction: 
the new (geometric) story}
\label{sec:RelativeVelocities}
In this section, which is the heart of 
this paper, we will present a geometric 
view of Lorentz transformations and how to 
decompose them as boosts and rotations. This 
will clarify the invariant meaning behind the Einstein law of velocity addition. All this 
will result in a geometrically satisfying 
definition of the notion of 
``relative velocity'' between two states of 
motion which we will call their 
``link velocity'', indicating that this 
definition rests on the so-called-called 
``boost link theorem'' that we state and 
prove. The conceptually 
important point to keep in mind is that 
the link-velocity between two states of 
motion needs to be referred to a third state 
of motion $s$. All constructions are 
``geometric''  insofar as all expressions 
involve only vectors and their scalar products. 

\subsection{
States of motion and the non-naturalness of 
polar decoposition
\label{sec:PolDecNonNat}
}
Above we used the polar decomposition to 
decompose a Lorentz transformation
into a spatial rotation and a boost. It is 
important to realise that this operation is 
\emph{not} natural. It depends on a preferred 
state of motion which in the matrix formulation 
above was given by the choice of the 
timelike vector $e_0$ of the chosen basis. Any 
rotation $R(\mat{D})$ takes place in the corresponding ``rest space'' that is the 
orthogonal complement of $e_0$. Hence, any 
of our rotations $R(\mat{D})$ acts in a 
spacelike 2-plane within that ``space'' and 
pointwise fixes the timleike 2-plane orthogonal 
to it. The latter always contains $e_0$. 
Similarly, any Boost $B(\vec\beta)$ takes place 
in a timelike 2-plane containing $e_0$ and 
pointwise fixes the spacelike orthogonal 
complement. Decompositions with respect 
to different choices of states of motion 
are \emph{a~priori} incomparable, as they 
refer to different ``spaces'' between which 
no natural identification (isomorphism) exists.  
  
Before we continue let us give the definition 
of the term ``state of motion'' that we just 
used:  
\begin{definition}
\label{def:StateOfMotion-1}
Let 
\begin{equation}
\label{eq:DefStateOfMotion}
V_1:=\{v\in V:\eta(v,v)=-1\}
\end{equation}
be the set of unit timelike vectors in $V$. 
It consists of two connected components, i.e. 
\begin{equation}
\label{eq:DefStateOfMotion-2}
V_1:=V_1^+\cup V_1^-\,,\quad 
V_1^+\cap V_1^-=\emptyset\,.
\end{equation}
If $v\in V_1^+$ then $-v\in V_1^-$. On $V_1$
we consider the equivalence relation 
$v\sim w\Leftrightarrow v=\pm w$. Elements in 
the quotient set 
\begin{equation}
\label{eq:DefStateOfMotion-3}
\som:=V_1/\!\sim
\end{equation}
are called \textbf{states of motion}. They may be 
faithfully represented by, say, elements of 
$V_1^+$, e.g. as follows: pick any element 
$s\in=V_1^+$; then
\begin{equation}
\label{eq:DefStateOfMotion-4}
\som=\{v\in V_1: \eta(v,s)\leq -1\}\,.
\end{equation} 
\end{definition}

As remarked, polar decomposition ist not 
a natural operation on $\Lor$. 
The additional structure that is needed is a 
state of motion or, equivalently, a \emph{Euclidean} inner product 
$g$ on $V$. The equivalence is seen as follows:
We identify $\som$ with $V_1^+$; then any 
$s\in\som$ defines a Euclidean inner product on 
$(V,\eta)$ by 
\begin{equation}
\label{eq:DefStateOfMotion-5}
g:=\eta+2\sigma\otimes\sigma\,,
\end{equation} 
where $\sigma:=\eta_{\downarrow}(s)$. 
Conversely, given $V$ and a Lorentzian metric 
$\eta$ as well as a Euclidean metric $g$, 
we consider the Euclidean sphere 
$S^3:=\{v\in V:g(v,v)=1\}$ (which is compact) 
and on it the function $Q:S^3\rightarrow\reals$,
$v\mapsto\eta(v,v)$. It has precisely 
two negative minima on $S^3$ corresponding to 
a pair of antipodal timelike vectors 
$\pm v\in V$ and hence, upon normalisation, determine a unique state of motion. 

We now give the general definition 
of ``polar decomposition''\footnote{We 
review the general theory of polar 
decomposition for finite-dimensional real 
vector spaces in Appendix\,\ref{sec:PolarDecomposition}.} 
for Lorentz transformations: 
\begin{definition}
\label{def:PolarDecomposition}
Given $L\in\Lor$, 
its \textbf{polar decomposition relative to a 
state of motion $s\in\som$} is 
as follows: Let $\sigma:=\eta_{\downarrow}(s)$
and $g$ defined by means of $s$ and $\eta$ 
as in \eqref{eq:DefStateOfMotion-5}. 
Write $L=BR$ where $R$ is orthogonal with 
respect to $g$ and where $B$ is 
symmetric and positive definite with 
respect to $g$; that is,  
$g(Bv,w)=g(v,Bw)$ for all $v,w\in V$ and 
$g(Bv,Bv)>0$ for all  $v\in V\backslash\{0\}$. 
For given $L$ the factors $R$ and $B$ always 
exist, are unique, and are both again elements 
of $\Lor$.  
\end{definition}
Note that the map $R$ fixes $s$; in fact, 
it pointwise fixes a timelike subspace of $V$ containing $s$ 
and acts non-trivially (if $R\ne\id_V$) on the 
two-dimensional spacelike plane orthogonal to 
it (the plane of rotation). Likewise, $B$ 
pointwise fixes a two-dimensional spacelike 
subspace of $V$ within the orthogonal complement 
of $s$ and acts non-trivially (if $B\ne\id_V$) 
on the two-dimensional timelike subspace of $V$ 
containing $s$. This leads us to 
\begin{definition}
\label{def:BoostRotation}
We call $L\in\Lor$ a \textbf{(spatial) rotation  
relative to $s\in\som$} iff it pointwise fixes 
a two-dimensional timelike subspace of $V$ 
containing $s$.  We call $L$ a 
\textbf{boost relative to $s\in\som$} iff 
it  pointwise fixes a two-dimensional 
spacelike subspace of $V$ within the 
orthogonal complement of $s$. 
\end{definition}
Suppose $L\in\Lor$ is a pure rotation relative 
to $s\in\som$. Then it is easy to see that $L$ is 
also a pure rotation relative $s'\ne s$ iff
$s'$ lies in the timelike orthogonal complement 
of the spacelike plane of rotation, i.e. if 
the timelike plane $\Span\{s,s'\}$ is pointwise 
fixed. Likewise, if $L\in\Lor$ is a pure boost 
relative to $s\in\som$, it is also a pure boost 
relative $s'\ne s$ iff $s'$ lies in the timelike 
plane of the boost, i.e. if the spacelike orthogonal 
complement of $\Span\{s,s'\}$ is poinwise fixed 
by $L$. 

To say that the polar decomposition of $\Lor$ 
is ``non-natural'' means that it only exists 
relative to an additional structural input, here 
the choice of some $s\in\som$. We will see below that 
given that choice the decomposition into boost and 
rotation can be easily formulated without ever 
talking about polar decomposition. This alternative
formulation will be better suited for comparison 
with the Galilei-Newton case that we perform in the final
section. 

\subsection{The kinematical setting}
\label{sec:KinematicalSetting}
We recall that the space of states $\som$ 
is identified with a connected component of the 
spacelike hyperboloid of unit timelike vectors
in $(V,\eta)$. Elements of $\som$ shall be 
denoted by the letter $s$, possibly with 
lower-case indices for distinction. We shall 
simplify notation by denoting the Minkowski 
inner product by a dot, i.e. $\eta(u,v)=:
u\cdot v$.  Further we write $u^2:=u\cdot u$ 
and $\Vert u\Vert:=\sqrt{\vert u^2\vert}$.
The latter does not define a norm in $V$ 
due to the indefiniteness of $\eta$, but it
does on any spacelike subspace $T_s\som$. 

We shall identify the tangent space to 
$\som$ at $s$ with the $\eta$-orthogonal 
complement of $s$ in $V$:
\begin{equation}
\label{eq:DefTangentSpace}
T_s\som=R_s:=\{u\in V:u\cdot s=0\}\,.
\end{equation}
In order to notationally distinguish tangent 
vectors to $\som$ from other vectors in $V$ 
we shall set them in bold. 

Associated to any $s\in\som$ are projection 
endomorphisms in $V$ parallel $(\Vert)$ and
 perpendicular $(\perp)$ to $s$: 
\begin{subequations}
\label{eq:DefProjEndo}
\begin{alignat}{2}
\label{eq:DefProjEndo-a}
&\proj_s^\Vert 
&&\,=\,-s\otimes s\,,\\
\label{eq:DefProjEndo-b}
&\proj_s^\perp 
&&\,=\,\id_v+s\otimes s\,.
\end{alignat}
\end{subequations}
As already explained in 
section\,\ref{sec:GenOrthGroupAlg} below
equation \eqref{eq:OrthogonalLieAlgebra},
 we shall identify 
$\mathrm{End}(V)=V\otimes V^*$ with 
$V\otimes V$, so that, e.g., an element 
$u\otimes v\in V\otimes V$ corresponds 
to the endomorphism 
$V\ni w\mapsto (v\cdot w)\,u\in V$. Hence 
\begin{equation}
\label{eq:ProjEndoActing}
\proj_s^\Vert(v)=-(v\cdot s)\,s
\quad\text{and}\quad
\proj_s^\perp(v)=v+(v\cdot s)\,s\,.
\end{equation}

\begin{definition}
\label{def:RelVel-1}
Given two states of motion $s$ and $s_1$. 
The \textbf{relative velocity between $s_1$ 
and $s$, judged from $s$,} is defined by 
\begin{equation}
\label{eq:DefRelVel-1}
\vec\beta(s,s_1):=\frac{\proj_s^\perp(s_1)}{\Vert\proj_s^\Vert(s_1)\Vert}\ \in\ T_s\som\,.
\end{equation}
This, in geometric terms, is just the 
ordinary definition of velocity (in 
units of $c$) in SR.
\end{definition}
The apparently redundant second reference to 
$s$ expressed in the phrase ``judged from $s$''
will be justified below. It turns out to be 
necessary because relative velocities 
between two states need reference to a third 
one, as already emphasised above, and that 
third one needs not be any of the given two 
ones. 

Noting that $\Vert P^\Vert_s(s_1)\Vert=-(s\cdot s_1)>0$, the expression in \eqref{eq:DefRelVel-1}
is just 
\begin{equation}
\label{eq:DefRelVel-2}
\vec\beta(s,s_1)=-s-\frac{s_1}{s\cdot s_1}\,.
\end{equation}
The squared modulus is 
\begin{equation}
\label{eq:DefRelVel-3}
\beta^2(s,s_1):=\Vert\vec\beta(s,s_1)\Vert^2=1-(s\cdot s_1)^{-2}\,,
\end{equation}
or (again taking into account that 
$(s\cdot s_1)<0$)
\begin{equation}
\label{eq:DefRelVel-4}
\gamma(s,s_1):=-\,(s\cdot s_1)=\frac{1}{\sqrt{1-\beta^2(s,s_1)}}\,.
\end{equation}
This is just the usual ``gamma-factor'' 
associated to any relative velocity. 
Note that the modulus  
\eqref{eq:DefRelVel-3} is symmetric in 
its arguments, i.e. $\beta(s,s_1)=\beta(s_1,s)$,
whereas the vectors 
$\vec\beta(s,s_1)\in T_s\som$ and 
$\vec\beta(s_1,s)\in T_{s_1}\som$ lie in 
different vector spaces and cannot be 
compared directly. In particular, 
a reciprocity statement, like  
``$\vec\beta(s,s_1)=-\vec\beta(s_1,s)$'',
would make no sense. We will see below  
how the reference to a common reference 
state will eventually render such 
a statement meaningful. 

It follows from our earlier discussion of boost transformation that there is a unique boost 
relative to $s$ transforming $s$ to $s_1$. 
We will denote it either by $B(s,s_1)$ or 
$B(s,\vec\beta)$, with $\vec\beta\in T_s\som$ 
given by \eqref{eq:DefRelVel-2}. In fact, its 
form can just be read off (by abstraction) 
from \eqref{eq:EmbedBoost}, 
now setting $\vec n:=\vec\beta/\beta$ and 
$\gamma=(1-\beta^2)^{-1/2}$:
\begin{subequations}
\label{eq:RelBoost-1}
\begin{alignat}{1}
\label{eq:RelBoost-1a}
B(s,\vec\beta)
&\,=\, \id_V+(\gamma-1)
 (-s\otimes s+\vec n\otimes\vec n)
 + \beta\gamma(s\otimes\vec n-\vec n\otimes s)\\
\label{eq:RelBoost-1b}
&\,=\,\proj^\perp_{(s,\vec n)}
    +\gamma \proj^\Vert_{(s,\vec n)}
    +\beta\gamma\,s\wedge\vec n\,.
\end{alignat}
\end{subequations}
The second line results from the first 
by observing that 
$(-s\otimes s+\vec n\otimes\vec n)$
is just the $\eta$-orthogonal projector 
onto $\Span\{s,\vec n\}$ which we 
denoted by $\proj^\Vert_{(s,\vec n)}$. 
Accordingly, $\proj^\perp_{(s,\vec n)}:=\id_V-\proj^\Vert_{(s,\vec n)}$. 
Now, any of the expressions 
\eqref{eq:RelBoost-1} immediately implies 
\begin{subequations}
\label{eq:RelBoost-2}
\begin{alignat}{1}
\label{eq:RelBoost-2a}
B(s,\vec\beta)s&\,=\,
\gamma(s+\beta\vec n)\,,\\
\label{eq:RelBoost-2b}
B(s,\vec\beta)\vec n&\,=\,
\gamma(\vec n+\beta s)\,,\\
\label{eq:RelBoost-2c}
B(s,\vec\beta)\vec v&\,=\,\vec v\ (\forall\vec v: \vec v\cdot s=\vec v\cdot\vec n=0)\,,
\end{alignat}
\end{subequations}
which clearly qualifies $B(s,\vec\beta)$
uniquely as the boost in the $s$-$\vec n$ 
plane with velocity $\vec\beta=\beta\vec n$.  
Replacing $\vec n$ according to \eqref{eq:DefRelVel-2} by $\vec n=(1/\beta\gamma)(s_1-\gamma s)$ in \eqref{eq:RelBoost-1a} leads
after a short calculation to 
\begin{equation}
\label{eq:RelBoost-3}
B(s,s_1)=\id_V+
\frac{%
 s\otimes s
+s_1\otimes s_1
+s\wedge s_1
-2\gamma s_1\otimes s}{\gamma+1}\,.
\end{equation}
Since $\gamma=-(s\cdot s_1)$, this shows that 
$B(s_1,s)$ is a rational function of $s$ and 
$s_1$.  This is in contrast to 
\eqref{eq:RelBoost-1a}, which, if expressed in 
terms of $\vec\beta=\beta\vec n$, still has 
terms proportional to $\gamma$ (rather than 
$\gamma^2$) which is not rational in $\beta$. 

Even though it is obvious from its derivation, we can verify directly 
that  \eqref{eq:RelBoost-3} is a 
boost in the plane spanned by $s$ and 
$s_1$ mapping $s$ to $s_1$. In fact, 
expression \eqref{eq:RelBoost-3} 
immediately leads to  
\begin{subequations}
\label{eq:RelBoost-4}
\begin{alignat}{2}
\label{eq:RelBoost-4a}
& B(s,s_1)s&&\,=\,s_1\,,\\
\label{eq:RelBoost-4b}
& B(s,s_1)s_1&&\,=\,-s+2\gamma s_1\,,\\
\label{eq:RelBoost-4c}
& B(s,s_1)u&&\,=\,u\ (\forall u: u\cdot s=u\cdot s_1=0)\,.
\end{alignat}
\end{subequations}
From (\ref{eq:RelBoost-4a},\ref{eq:RelBoost-4b}) 
we easily infer by short calculations that all 
three scalar products $s\cdot s=-1$, $s_1\cdot s_1=-1$, and $s\cdot s_1=-\gamma$ are left 
invariant under $B(s,s_1)$, which together with 
\eqref{eq:RelBoost-4c} shows that 
indeed all scalar products are left invariant and that hence $B(s_1,s)$
is the said boost.

Equations 
(\ref{eq:RelBoost-4a},\ref{eq:RelBoost-4b}) 
also allow us to easily determine the 
action of $B(s,s_1)$ onto 
$\vec\beta(s,s_1)\in T_s\som$, which 
necessarily results in an element of 
$T_{s_1}\som$. From \eqref{eq:DefRelVel-2} 
we infer:
\begin{equation}
\label{eq:RelBoost-5}
\begin{split}
B(s,s_1)\vec\beta(s,s_1)
&=B(s,s_1)\left(-s-\frac{s_1}{(s\cdot s_1)}\right)\\
&=\left(-s_1-\frac{-s-2(s\cdot s_1)s_1}{s\cdot s_1}\right)\\
&=-\left(-s_1-\frac{s}{s_1\cdot s}\right)\\
&=-\vec\beta(s_1,s)\,,
\end{split}
\end{equation}  
where we used once more \eqref{eq:DefRelVel-2} in the last step.
This  is the right form to state a  
reciprocity of relative velocities at this point\footnote{We will later state another 
one in which both ``velocities'' 
(more precisely: ``link-velocities'',
see below) refer to the same reference
state; see Corollary\ref{thm:Reciprocity}.}  if both refer to different reference 
states. Namely, whereas 
$\vec\beta(s,s_1)\in T_s\som$ and 
$\vec\beta(s_1,s)\in T_{s_1}\som$ 
are incomparable because they lie in 
different tangent spaces, 
$B(s,s_1)\vec\beta(s,s_1)$ can be 
compared to  $\vec\beta(s_1,s)$ (both
lying in $T_{s_1}\som$) and 
likewise $B(s_1,s)\vec\beta(s_1,s)$
can be compared to $\vec\beta(s,s_1)$
(both lying in $T_s\som$) with the 
result that these are, respectively, 
the negative of each other. Moreover, 
as we will show in 
Appendix\,\ref{sec:ParallelTransport},
the linear isometry $B(s,s_1):T_s\som\rightarrow
T_{s_1}\som$ is just that resulting 
from parallel transport along the 
unique geodesic within $\som$ (with 
respect to the metric $\eta$ restricted 
to $T\som$) connecting $s$ with $s_1$.
Hence reciprocity can also be stated with 
$B(s,s_1)$ being interpreted as parallel 
transport along connecting geodesics.  

At this point we recall that any boost 
transformation \eqref{eq:RelBoost-1}
is in the image of the exponential map. 
In fact 
\begin{proposition}
Let $\vec\beta=\beta\vec n$ and 
$\vec\rho=\rho\vec n$; then
\begin{subequations}
\label{eq:BoostExponential1}
\begin{alignat}{1}
\label{eq:BoostExponential1-a}
\exp\bigl(s\wedge \vec   
\rho\bigr)&\,=\,B(s,\vec\beta)\,,\\
\label{eq:BoostExponential1-b}
\text{where}\quad \rho &\,=\,\tanh^{-1}(\beta)\,.
\end{alignat}
\end{subequations}
\end{proposition}
\begin{proof}
We start by noting that 
\begin{equation}
\label{eq:BoostExponential2}
\begin{split}
(s\wedge\vec n)^2
&=
(s\otimes\vec n-\vec n\otimes s)\circ
(s\otimes\vec n-\vec n\otimes s)
=-s\otimes s+\vec n\otimes\vec n\\ &=\proj^\Vert_{(s,\vec n)}\,. 
\end{split}
\end{equation} 
Hence, decomposing the exponential series 
into even and odd powers gives
\begin{equation}
\label{eq:eq:BoostExponential3}
\begin{split}
\exp(\rho\,s\wedge\vec n)
&=\sum_{k=0}^\infty\frac{\rho^k}{k!}(s\wedge\vec n)^k\\
&=\sum_{k=0}^\infty\frac{\rho^{2k}}{(2k)!}(s\wedge\vec n)^{2k}
 +\sum_{k=0}^\infty\frac{\rho^{2k+1}}{(2k+1)!}(s\wedge\vec n)^{2k+1}
\end{split}
\end{equation} 
In view of \eqref{eq:BoostExponential2}, each term in the 
first (even) sum is proportional to $\proj^\Vert_{(s,\vec n)}$
except the first ($k=0$), which equals 
$\id_V=\proj^\Vert_{(e,\vec n)}+\proj^\perp_{(e,\vec n)}$. 
In the second (odd) sum, again because of  \eqref{eq:BoostExponential2} and also because of 
$\proj^\Vert_{(s,\vec n)}\circ (s\wedge\vec n)
=(s\wedge\vec n)\circ\proj^\Vert_{(s,\vec n)}
=s\wedge\vec n$ we have 
$(s\wedge\vec n)^{2k+1}=(s\wedge\vec n)$  for all 
$k\geq 0$. Hence 
\begin{equation}
\label{eq:eq:BoostExponential4}
\exp(\rho\,s\wedge\vec n)
=\proj^\perp_{(s,\vec n)}+\cosh(\rho)\,\proj^\Vert_{(s,\vec n)}
+\sinh(\rho)\,s\wedge\vec n\,,
\end{equation} 
which is just \eqref{eq:RelBoost-1b} taking into 
account \eqref{eq:BoostExponential1-b}, i.e.
$\gamma=1/\sqrt{1-\beta^2}=\cosh(\rho)$ and 
$\beta\gamma=\sinh(\rho)$.
\end{proof}
 
Finally we show how to repeat the initial 
matrix-calculation that led to the addition 
formula, now in a geometric fashion making 
manifest that all the velocities appearing 
in it refer to the same reference state 
$s$ and are hence elements of the same 
vector space. Using \eqref{eq:RelBoost-1a} 
and $s\cdot\vec\beta_i=0$ ($i=1,2$) we get 

\begin{equation}
\label{eq:VelocityAddGeom-1}
\begin{split}
\bigl(B(s,\vec\beta_1)\circ B(s,\vec\beta_2)\bigr)s
&=B(s,\vec\beta_1)
  \bigl[\gamma_2(s+\vec\beta_2)\bigr]\\
&=\gamma_2B(s,\vec\beta_1)s+
  \gamma_2B(s,\vec\beta_1)\vec\beta_2\\
&=\gamma_1\gamma_2(s+\vec\beta_1)\\
&\quad+ \gamma_2\bigl[
\vec\beta_2+(\gamma_1-1)(\vec n_1\cdot\vec \beta_2)\,\vec n_1+\gamma_1(\vec\beta_1\cdot\vec\beta_2)\,s
\bigr]\\
&=\gamma_1\gamma_2(1+\vec\beta_1\cdot\vec\beta_2)\,s\\
&\quad +\gamma_1\gamma_2\bigl[\vec\beta_1
 +(\vec n_1\cdot\vec\beta_2)\,\vec n_1\bigr]
+\gamma_2\bigl[\vec\beta_2-
(\vec n_1\cdot\vec\beta_2)\,\vec n_1\bigr]\\
&=\gamma_1\gamma_2(1+\vec\beta_1\cdot\vec\beta_2)
\left(s+\frac{\vec\beta_1+\vec\beta_2^\Vert+\gamma_1^{-1}\vec\beta_2^\perp}{1+\vec\beta_1\cdot\vec\beta_2}\right)\,,
\end{split}
\end{equation}
where, as before, the superscripts $\Vert$ and $\perp$ refer to the projections parallel and 
perpendicular to $\vec n_1$. This is of the 
form of a pure boost relative to $s$ acting 
on $s$:
\begin{equation}
\label{eq:VelocityAddGeom-2}
B(s,\vec\beta)s=\gamma(s+\vec\beta)
\end{equation}
with 
\begin{subequations}
\label{eq:VelocityAddGeom-3}
\begin{alignat}{1}
\label{eq:VelocityAddGeom-3a}
\gamma &\,=\,\gamma_1\gamma_2
 (1+\vec\beta_1\cdot\vec\beta_2)\,,\\
\label{eq:VelocityAddGeom-3b}
\vec\beta &\,=\,\frac{\vec\beta_1
 +\vec\beta_2^\Vert+\gamma_1^{-1}\vec\beta_2^\perp}{1+\vec\beta_1\cdot\vec\beta_2}=\vec\beta_1\vplus\vec\beta_2\,.
\end{alignat}
\end{subequations}
Note that we only calculated the action of 
$B(\vec\beta_1,s)B(\vec\beta_2,s)$ on $s$, 
not the map as such, which also contains the 
Thomas rotation in the spacelike plane 
perpendicular to $s$ and which hence acts as 
the identity on $s$. This is why the Thomas 
rotation does not appear here.  

As a final remark in this subsection we wish to 
point out an interesting relation with hyperbolic 
geometry. We fix a reference state $s$ and set  
$s_i:=B(s,\vec\beta_i)s$ for $i=1,2$. We consider 
the open unit ball in the tangent space of $\som$
at $s$,
\begin{equation}
\label{eq:BallTangent-1}
\U:=\{\vec v\in T_s\som: \vec v\cdot\vec v<1\}\,,
\end{equation}
in which the velocities relative to $s$ take their
values. On $\U$ we define the hyperbolic distance 
function $d:\U\times\U\rightarrow\reals_{\geq 0}$
which assigns to each pair $\vec\beta_{1{,}2}\in\U$ 
the Riemannian geodesic distance\footnote{Equivalently: hyperbolic 
angle or rapidity.} between $s_{1{,}2}=\gamma_{1{,}2}(s+\vec\beta_{1{,}2})$. In this way $(\U,d)$ becomes 
isometric to $(\som,d_h)$, where $d_h$ is the distance 
function induced by the Riemannian metric $h$ on $\som$.
Explicitly, the distance function $d$ is given by  
\begin{equation}
\label{eq:BallTangent-2}
d(\vec\beta_1,\vec\beta_2)
:=\text{arcosh}\bigl(-s_1\cdot s_2\bigr)
=
\text{arcosh}\bigl(
\gamma_1\gamma_2(1-\vec\beta_1\cdot\vec\beta_2)
\bigr)\,.
\end{equation}
Now, (\ref{eq:VelocityAddGeom-1}-\ref{eq:VelocityAddGeom-3})
show that for any boost $B(s,\vec\beta)$ relative to 
$s$ we have 
\begin{equation}
\label{eq:BallTangent-3}
s'_i:=B(s,\vec\beta)s_i=B(s,\vec\beta\vplus\vec\beta_i)s
\end{equation}
But as boosts preserve scalar products we have
$s'_1\cdot s'_2=s_1\cdot s_2$, which implies
\begin{equation}
\label{eq:BallTangent-4}
d(\vec\beta\vplus\vec\beta_1\,,\,\vec\beta\vplus\vec\beta_2)
=d(\vec\beta_1,\vec\beta_2)\,,
\end{equation}
for any triple $(\beta,\beta_1,\beta_2)$ 
of points in $\U$. This equations says that, 
with respect to $d$, velocity addition $\vplus$ 
defines an isometric ``action'' of $\U$ on itself. 
The origin of this lies of course in the proper 
isometric action of the Lorentz group on $(\som,h)$.
One may ask to what extent relativistic addition
is characterised  by \eqref{eq:BallTangent-4}. 
In other words: what is the most general map 
$f:\U\times\U\rightarrow\U$ that satisfies  
\begin{equation}
\label{eq:BallTangent-5}
d\bigl(f(\vec\beta,\vec\beta_1),f(\vec\beta,\vec\beta_2)\bigr)
=d(\vec\beta_1,\vec\beta_2)\,,
\end{equation} 
for any triple $(\beta,\beta_1,\beta_2)$ of 
points in $\U$? This has been solved by 
\citet[Theorem\,2]{Benz:2000} who proved that $f$ 
equals $\vplus$ iff $f$ satisfies the following two 
conditions: 
1)~$\Vert f(\vec\beta,\vec 0)\Vert=\Vert\vec\beta\Vert$ and
2)~$\gamma(f(\vec\beta_1,\vec\beta_2))\,f(\vec\beta_1,\vec\beta_2)-\gamma(\vec\beta_2)\vec\beta_2=
c\vec\beta_1$, where $c>0$. The first seems obvious and the 
second acquires a straightforward physical meaning if we 
recall that the momentum relative to $s$ of a particle 
at velocity $\vec\beta\in T_s\som$ is proportional to 
$\gamma(\vec\beta)\vec\beta$. The second condition 
then just says that boosting a particle with 
boost-velocity $\vec\beta$ adds to it a momentum 
positively proportional to $\vec\beta$ (all with reference 
to $s$). Using this momentum-interpretation, the two 
conditions may be replaced by others that also include
the non-relativistic case \citep{Benz:2002}.

\subsection{Defining link velocity}
\label{sec:BoostLinkProblem}
Given three states of motion, $s$, $s_1$, and 
$s_2$. It turns out that there exists a unique 
boost $B(\vec\beta,s)$ relative to $s$ that maps 
$s_1$ to $s_2$:
\begin{equation}
\label{eq:BLP1}
B(s,\vec\beta)s_1=s_2\,.
\end{equation}
This is the affirmative answer to the so-called 
``boost-link-problem'' that we shall prove in an 
elementary fashion in the next subsection. More 
precisely, the results of the next subsection can be 
summarised as follows:  
\begin{theorem}
\label{thm:BoostLink-1}
For any three given states of motion 
$(s,s_1,s_2)$ there exists a unique 
$\vec\beta\in T_s\som$, 
given by a rational function 
$\vec\beta=\vec\beta(s,s_1,s_2)$ of $(s,s_1,s_2)$,
satisfying \eqref{eq:BLP1} (see
\eqref{eq:BoostLink-7a} below). The boost 
$B(s,s_1,s_2):=
B\bigl(s,\vec\beta(s,s_1,s_2)\bigr)$
is then also a rational function of 
$(s,s_1,s_2)$ (see \eqref{eq:BoostLink-11} 
below). The function $\vec\beta$ is Lorentz
equivariant; i.e. for any $\som$-preserving 
(no time reversal) Lorentz transformation $L$ 
we have
\begin{equation}
\label{eq:LorentzEquivariance}
\vec\beta(Ls,Ls_1,Ls_2)=L\vec\beta(s,s_1,s_2)\,.
\end{equation}
\end{theorem}
As an immediate consequence we note that for 
special $s$ we have already encountered 
the expression for the boost linking $s_1$ 
and $s_2$: 
\begin{corollary}
\label{thm:BoostLink-2} 
Let $(s,s_1,s_2)$ and $(s',s_1,s_2)$ be two 
triple of states where $s$ as well as $s'$ 
lie in the plane $\Span\{s_1,s_2\}$.
Then the boosts $B(s,\vec\beta)$ and 
$B(s',\vec\beta')$ which according to 
Theorem\,\ref{thm:BoostLink-1}
satisfy \eqref{eq:BLP1} are identical and 
given by 
\begin{equation}
\label{eq:BLP2}
B(s,\vec\beta)=B(s',\vec\beta')=\id_V+
\frac{%
 s_1\otimes s_1
+s_2\otimes s_2
+s_1\wedge s_2
-2\gamma_{12} s_2\otimes s_1}{
1+\gamma_{12}}\,,
\end{equation}
where 
$\gamma_{12}:= -(\vec s_1\cdot\vec s_2)$.
Note that the right-hand side of 
\eqref{eq:BLP2} is independent of $s$ 
and given by the expression 
\eqref{eq:RelBoost-3} with arguments 
$(s,s_1)$ changed to $(s_1,s_2)$. 
Hence all boosts linking $s_1$ and 
$s_2$ coincide as long as 
their reference state $s$ lies in the
plane spanned by $s_1$ and $s_2$. 
\end{corollary}
\begin{proof}
We already know that expression 
\eqref{eq:RelBoost-3} with
arguments $(s,s_1)$ changed to 
$(s_1,s_2)$ gives the unique boost 
in the plane $\Span\{s_1,s_2\}$
mapping $s_1$ to $s_2$. But that 
plane contains $s$ and hence is also 
a boost relative to $s$. 
Uniqueness then implies the statement.  
\end{proof}

Based on Theorem\,\ref{thm:BoostLink-1}
we now make the following 
\begin{definition}
\label{def:RelVel-2}
Given three states of motion 
$s$, $s_1$, and $s_2$. 
The \textbf{link-velocity between $s_1$ 
and $s_2$ relative to $s$} is defined 
as that (unique!) $\vec\beta\in T_s\som$ 
solving \eqref{eq:BLP1}. We will speak of it 
as ``the velocity of $s_2$ against $s_1$ 
relative to $s$'' or ``the velocity of $s_2$ 
relative to $s_1$ judged from $s$'', or similar, 
so as to in any case avoid a double appearance 
of the word ``relative''.
\end{definition} 

In view of Theorem\,\ref{thm:BoostLink-1} we can 
characterise the link-velocity in the following way:
\begin{theorem}
\label{thm:LinkVelSection}
The link-velocity is a function 
\begin{equation}
\label{eq:LinkVelSection-1}
\begin{split}
\vec\beta:
\som\times\som &\rightarrow\Gamma T\som\\
(s_1,s_2)&\mapsto\vec\beta[s_1,s_2]
\end{split}
\end{equation} 
that maps any ordered pair $(s_1,s_2)$ of states 
to a section $\vec\beta[s_1,s_2]$ in the 
tangent-bundle of $\som$. That section is such 
that the value of $\vec\beta[s_1,s_2]$ at 
$s\in\som$ equals $\vec\beta(s,s_1,s_2)$ as 
defined in Theorem\,\ref{thm:BoostLink-1}, 
the explicit expression of which is given 
below in \eqref{eq:BoostLink-7a}. 
\end{theorem}

\begin{remark}
The condition of equivariance 
\eqref{eq:LorentzEquivariance} is equivalent 
to the statement that the map \eqref{eq:LinkVelSection-1} 
is invariant under $\group{Lor}$. Indeed, 
$\group{Lor}$ acts on the domain 
$\som\times\som$ by taking the cartesian-product 
of its natural action on $\som$: 
\begin{subequations}
\label{eq:LorActionsDomains}
\begin{equation}
\label{eq:LorActionsDomains-a}
\begin{split}
\Phi: \Lor\times (\som\times\som)
&\rightarrow (\som\times\som)\\
\bigl(L,(s_1,s_2)\bigr)
&\mapsto\Phi_L(s_1,s_2):=(Ls_1,Ls_2)\,.
\end{split}
\end{equation}
Moreover, $\group{Lor}$ also acts on the co-domain 
(the target) $\Section T\som$ as folllows: 
\begin{equation}
\label{eq:LorActionsDomains-b}
\begin{split}
\Psi: \Lor\times\Section T\som
&\rightarrow \Section T\som\\
(L,\sigma)
&\mapsto\Psi_L(\sigma):=L\circ\sigma\circ L^{-1}\,.
\end{split}
\end{equation}
\end{subequations}
Note that $L$ acts directly on the image of $\sigma$
via its defining representation on $V$, since according 
to \eqref{eq:DefTangentSpace} we identified $T_s\som$
with $R_s:=\{u\in V: u\cdot s=0\}$. Restricting the 
defining action of $L$ on $V$ to $R_s\subset V$ 
results in an isometry between $R_s\som$ and 
$R_{Ls}\som$, i.e. between $T_s\som$ and 
$T_{Ls}\som$. Now, by standard construction (already 
applied in \eqref{eq:LorActionsDomains-b}), actions on 
domains and codomains always combine to an action on 
the set of maps. Applied to the map $\vec\beta$ in 
\eqref{eq:LinkVelSection-1} with actions 
\eqref{eq:LorActionsDomains-a} on the domain 
and \eqref{eq:LorActionsDomains-b} on the 
co-domain, we get an action of $\Lor$ on the set of 
maps $\som\times\som\rightarrow\Section T\som$, given 
by   
\begin{equation}
\label{eq:LinkVelLorInv-1}
T_L(\vec\beta):=\Psi_L\circ\vec\beta\circ\Phi^{-1}_L\,.
\end{equation}
Equation \eqref{eq:LorentzEquivariance}
is then equivalent to the statement that 
the map \eqref{eq:LinkVelSection-1} is $\Lor$-invariant:  
\begin{equation}
\label{eq:LinkVelLorInv}
T_L(\vec\beta)=\vec\beta \quad(\forall\,L\in\Lor)\,.
\end{equation} 
\end{remark}

The fact that the link-velocity $\vec\beta(s,s_1,s_2)$
is a ternary has been discussed before by a particular 
school following \citet{Oziewicz2006,Oziewicz2007,Oziewicz2011,Oziewicz.Page2011}, 
who interpret this fact as an ``astonishing conflict 
of the Lorentz group with relativity'' \citep{Oziewicz2011}; 
see also \citep{Celakoska:2008,Celakoska.EtAl:2015,Koczan:2023}.
But \eqref{eq:LinkVelLorInv} shows that this is an unwarranted 
complaint.

\subsection{Solving the boost-link problem}
\label{sec:BoostLinkProof}
In this section we shall prove and 
elaborate on Theorem\,\ref{thm:BoostLink-1},
using the following notation: 
The element $\vec\beta\in T_s\som$
that we wish \eqref{eq:BLP1} to solve for 
is again written as $\vec\beta=\beta\vec n$
with $\vec n\in T_s\som$ a unit vector.
The norm $\beta$ of $\vec\beta$ can equivalently be expresses as usual by 
$\gamma:=1/\sqrt{1-\beta^2}$, i.e. 
$\beta=\sqrt{1-\gamma^{-2}}$.
Next to that we define the other 
``gammas'' via the three possible scalar products between 
the $\{s,s_1,s_2\}$:
\begin{subequations}
\label{eq:gammas}
\begin{alignat}{3}
\label{eq:gammas-1}
\gamma_1
&\,:=\,-s\cdot s_1\quad
&&\Rightarrow\quad\beta_1 
&&\,:=\,\sqrt{1-\gamma_1^{-2}}\,,\\
\label{eq:gammas-2}
\gamma_2
&\,:=\,-s\cdot s_2\quad
&&\Rightarrow\quad\beta_2 
&&\,:=\,\sqrt{1-\gamma_2^{-2}}\,,\\
\label{eq:gammas-12}
\gamma_{12}
&\,:=\,-s_1\cdot s_2\quad
&&\Rightarrow\quad\beta_{12} 
&&\,:=\,\sqrt{1-\gamma_{12}^{-2}}\,.
\end{alignat}
\end{subequations} 
In what follows $(s,s_1,s_2)$ and hence 
$(\gamma_1,\gamma_2,\gamma_{12})$ and 
$(\beta_1,\beta_2,\beta_{12})$ are 
considered given, whereas $(\beta,\vec n)$,
or equivalently $(\gamma,\vec n)$, are to
be determined as functions of the former. 
This is the task to which we now turn. 

Inserting the expression \eqref{eq:RelBoost-1a} for 
$B(\vec\beta,s)$ into \eqref{eq:BLP1}
leds to 
\begin{equation}
\label{eq:BoostLink-1}
\begin{split}
s_2-s_1
=s\,&\bigl[(\gamma-1)\gamma_1+\beta\gamma(\vec n\cdot s_1)\bigr]\\
+\,\vec n &\bigl[(\gamma-1)(\vec n\cdot s_1)+\beta\gamma\gamma_1\bigr]\,.
\end{split}
\end{equation}
The right-hand side gives the components 
of $(s_2-s_1)$ parallel and perpendicular
to $s$. The parallel component of the 
left-hand side is 
$\proj_s^\Vert(s_2-s_1)=s(\gamma_2-\gamma_1)$. 
Equating this to the $s$-term of the 
right-hand side allows to express 
$(\vec n\cdot s_1)$ as follows:
\begin{equation}
\label{eq:BoostLink-2}
\vec n\cdot s_1=
\frac{\gamma_2-\gamma\gamma_1}{\beta\gamma}\,.
\end{equation}
The perpendicular component of the 
left-hand side is $\proj_s^\perp(s_2-s_1)$
which we equate to the $\vec n$-term on the 
right-hand side. In the latter we 
replace $\vec n\cdot s_1$ with the 
expression just found in  \eqref{eq:BoostLink-2}. This leads, 
after a short calculation, to 
\begin{equation}
\label{eq:BoostLink-3}
\vec n=\sqrt{\frac{\gamma+1}{\gamma-1}}\
\frac{\proj_s^\perp(s_2-s_1)}{\gamma_1+\gamma_2}\,.
\end{equation}
This is not yet the solution since an 
unknown, $\gamma$, still appears on 
the right-hand side. But we can 
determine $\gamma$ by using the fact 
that $\vec n$ has unit norm. Recalling
that  
\begin{subequations}
\label{eq:BoostLink-4}
\begin{alignat}{1}
\label{eq:BoostLink-4a}
\proj_s^\perp(s_2-s_1)
&\,=\,
s_2-s_1-(\gamma_2-\gamma_1)\,s\,,\\
\label{eq:BoostLink-4b}
\Vert\proj_s^\perp(s_2-s_1)\Vert^2
&\,=\,\gamma_1^2+\gamma_2^2+2(\gamma_{12}-\gamma_1\gamma_2-1)\,,
\end{alignat}
\end{subequations}
we can take the square of \eqref{eq:BoostLink-3} 
and solve the ensuing equation for $\gamma$, 
which gives 
\begin{equation}
\label{eq:BoostLink-5}
\gamma=\gamma(s,s_1,s_2):=\frac{\gamma_1^2+\gamma_2^2+\gamma_{12}-1}{1+2\gamma_1\gamma_2-\gamma_{12}}\,,
\end{equation}
where we think of the right-hand side as 
rational function in the scalar products 
\eqref{eq:gammas}.

Note that now all terns of the right-hands side 
of \eqref{eq:BoostLink-5}, and hence also on the 
right-hand side of  \eqref{eq:BoostLink-3}, are 
expressed in terms of given quantities. Hence we 
succeeded in proving existence and uniqueness 
of the solution for the boost-link-problem and 
also justified Definition\,\ref{def:RelVel-2}.

In the sequel we shall use the abbreviation
\begin{equation}
\label{eq:BoostLink-6}
\vec{(s_2-s_1)_\perp}:=\proj^\perp_s(s_2-s_1)
\end{equation}
which we write in bold so stress that this 
is an element in $T_s\som$. It is just a 
simple linear combination of $s$, $s_1$, 
and $s_2$ as shown in \eqref{eq:BoostLink-4a}, 
but it will be notationally more compact 
and also geometrically more transparent 
to write $\vec{(s_2-s_1)_\perp}$. Now,  
multiplying \eqref{eq:BoostLink-3} with 
$\beta=\gamma^{-1}\sqrt{\gamma^2-1}$ we 
get, using \eqref{eq:BoostLink-5}, 
\begin{subequations}
\label{eq:BoostLink-7}
\begin{alignat}{1}
\label{eq:BoostLink-7a}
\vec\beta
&\,=\,\vec\beta(s,s_1,s_2):= \frac{\gamma_1+\gamma_2}{\gamma_1^2+\gamma_2^2+\gamma_{12}-1}\,\vec{(s_2-s_1)_\perp}\,,\\
\label{eq:BoostLink-7b}
\gamma\vec\beta&\,=\,\gamma(s,s_1,s_2)\vec\beta(s,s_1,s_2)
=
\frac{\gamma_1+\gamma_2}{1+2\gamma_1\gamma_2-\gamma_{12}}\,\vec{(s_2-s_1)_\perp}\,.
\end{alignat}
\end{subequations}
Note that in view of \eqref{eq:BoostLink-4a}
the right-hand sides are linear combinations of 
$s, s_1$, and $s_2$ with coefficients that are 
rational functions in the scalar products between 
these states. This is what we mean by saying that  
\eqref{eq:BoostLink-7a} is itself a rational function 
of $(s,s_1,s_2)$. It is the first such function 
mentioned in Theorem\,\ref{thm:BoostLink-1}. The 
equivariance condition  \eqref{eq:LorentzEquivariance}
is obvious from these remarks. 

We also note that the
fraction of the right-hand side of 
\eqref{eq:BoostLink-7a} is a symmetric 
function in $(s_1,s_2)$, whereas 
$\vec{(s_2-s_1)_\perp}$ is clearly 
antisymmetric. Hence we have 
\begin{corollary}
\label{thm:Reciprocity}
Link-velocities obey the reciprocity relation
\begin{equation}
\label{eq:BoostLink-8}
\vec\beta(s,s_2,s_1)=-\,\vec\beta(s,s_1,s_2)\,.
\end{equation}
\end{corollary}
Note that \eqref{eq:BoostLink-8} makes
sense since both sides refer to the same 
reference-state $s$, i.e. both sides are 
elements of the same vector space $T_s\som$. 

We now compute the boost as function
of $(s,s_1,s_2)$, i.e. we insert the expressions \eqref{eq:BoostLink-3} for 
$\vec n$, \eqref{eq:BoostLink-5} for 
$\gamma$, and \eqref{eq:BoostLink-7b}
for $\gamma\vec\beta$ into \eqref{eq:RelBoost-1a} and obtain
\begin{equation}
\begin{split}
\label{eq:BoostLink-9}
 B(s,s_1,s_2)
:=\ & B\bigl(s,\vec\beta(s,s_1,s_2))\\
=\ & \id_V+(\gamma-1)
   (-s\otimes s+\vec n\otimes\vec n)
   + s\wedge\gamma\vec\beta\,.
\end{split}
\end{equation}
The various terms simplify as follows: 
From \eqref{eq:BoostLink-5} we get
\begin{subequations}
\label{eq:BoostLink-10}
\begin{equation}
\label{eq:eq:BoostLink-10a}
-\,(\gamma-1)s\otimes s=-\,\frac{(\gamma_1-\gamma_2)^2+2(\gamma_{12}-1)}{1+2\gamma_1\gamma_2-\gamma_{12}}\,s\otimes s\,.
\end{equation}
From \eqref{eq:BoostLink-3} and again \eqref{eq:BoostLink-5}
\begin{equation}
\label{eq:BoostLink-10b}
\begin{split}
(\gamma-1)\vec n\otimes\vec n
&\,=\,\frac{\gamma+1}{(\gamma_1+\gamma_2)^2}\ 
  \vec{(s_2-s_1)_\perp}\otimes\vec{(s_2-s_1)_\perp}\\
&\,=\,\frac{\vec{(s_2-s_1)_\perp}\otimes\vec{(s_2-s_1)_\perp}}{1+2\gamma_1\gamma_2-  
            \gamma_{12}} \,.
\end{split}
\end{equation}
And from \eqref{eq:BoostLink-7b} 
\begin{equation}
\label{eq:BoostLink-10c}
s\wedge\gamma\vec\beta
=\frac{\gamma_1+\gamma_2}{1+2\gamma_1\gamma_2-\gamma_{12}}\, s\wedge\vec{(s_2-s_1)_\perp}\,.
\end{equation}
Using \eqref{eq:BoostLink-4a}, the sum of \eqref{eq:BoostLink-10b} and 
\eqref{eq:BoostLink-10c} is, abbreviating 
$s_{21}:=(s_2-s_1)$,
\begin{equation}
\label{eq:eq:BoostLink-10d}
\begin{split}
& (\gamma-1)\vec n\otimes\vec n+s\wedge\gamma\vec\beta
=\\
&\quad\frac{
(\gamma_1-\gamma_2)^2\, s\otimes s
+s_{21}\otimes s_{21}
+2\gamma_1\,s\otimes s_{21}
-2\gamma_2\,s_{21}\otimes s}{1+2\gamma_1\gamma_2-\gamma_{12}}\,.
\end{split}
\end{equation}
\end{subequations}
Finally, adding this to
\eqref{eq:eq:BoostLink-10a} we obtain 
for \eqref{eq:BoostLink-9} 
\begin{equation}
\begin{split}
\label{eq:BoostLink-11}
 & B(s,s_1,s_2)=\\
&\quad \id_V+\frac{
2(1-\gamma_{12})\, s\otimes s
+s_{21}\otimes s_{21}
+2\gamma_1\,s\otimes s_{21}
-2\gamma_2\,s_{21}\otimes s}{1+2\gamma_1\gamma_2-\gamma_{12}}\,.
\end{split}
\end{equation}
This is the second rational function mentioned 
in Theorem\,\ref{thm:BoostLink-1}.

A few easy checks reassure us that 
\eqref{eq:BoostLink-11} is indeed the 
right expression. First of all, the 
tensor structure of the right-hand side 
makes it immediately evident that  
$B(s,s_1,s_2)$ maps the 2-dimensional
timelike plane $\Span\{s,s_{12}\}$ 
to itself and fixes points in the 
2-dimensional spacelike orthogonal 
complement. Second, a simple computation 
gives $B(s,s_1,s_2)s_1=s_2$ and, third, 
another simple computation for the 
special case where $s=s_1$ (hence 
$\gamma_1=1$ and $\gamma_2=\gamma_{12}$) turns expression \eqref{eq:BoostLink-11} into \eqref{eq:BLP2}. 

Finally we mention an alternative way to 
derive \eqref{eq:BoostLink-11}, based on the 
Cartan-Dieudonné theorem\footnote{Let $(V,\eta)$
be an $n$-dimensional real vector space with 
non-degenerate symmetric bilinear form 
$\eta\in V^*\otimes V^*$.
Let $\group{O}(V,\eta)$ be the corresponding 
orthogonal group, defined as in 
\eqref{eq:OrthogonalLieGroup}. 
The \emph{Cartan-Dieudonné theorem} states 
that any $L\in \group{O}(V,\eta)$ is the 
composition of at most $n$ reflections 
at non-null (non-degenerate) hyperplanes.
We recall that if $u\in V$ is non-null, i.e.
$\eta(u,u)\ne 0$, so that we may assume
$\eta(u,u)=\varepsilon\in\{1,-1\}$, the reflection $\rho_u\in\group{O}(V,\eta)$ at 
the hyperplane $u^\perp\subset V$ perpendicular to $u$ is defined by 
$\rho_u(v):=v-2\varepsilon \eta(v,u)\,u$.
We refer to \citep[Chap.\,6.6]{Jacobson:BasicAlgebraI}
for a proof of the general Cartan-Dieudonné 
theorem. If one relaxes the upper bound on 
the number of reflections from $n$ to 
$2n-1$, the proof becomes much easier; 
see, e.g., 
\citep[p.\,101, Theorem\,6]{Giulini:2006b}.}, according to which we 
can write the linking boost by the composition of two reflections; see \citep{Urbantke:2003}.

\subsection{Base-point dependence of 
the link-velocity}
\label{sec:BasePointDependence}
We have just seen that the boost relative 
to $s$ that links $s_1$ with $s_2$ takes place in the plane spanned by $s$ and 
$(s_2-s_1)$, the velocity being given 
by \eqref{eq:BoostLink-7a}. 
For fixed $s_1$ and $s_2$ the velocities 
$\vec\beta(s,s_1,s_2)$ vary with $s$ 
in a twofold way. First, their ``directions'' differ in the 
sense that they are elements on the 
tangent spaces $T_s\som$ depending on 
$s$. Second, the magnitude also depends 
on $s$ in an interesting way that we 
now wish to elaborate on. We do this 
by studying $\gamma(s_2,,s_1;s)$ as a 
function of $s$.

From Corollary\,\ref{thm:BoostLink-2} we know that $B(s,s_1,s_2)$, and hence 
$\gamma(s,s_1,s_2)$ does not depend 
on $s$ as long as 
$s\in\Span\{s_1,s_2\}$. In that 
case we obviously have $\gamma=\gamma_{12}$.
This suggests to compute the deviation 
of $\gamma$ from $\gamma_{12}$ in 
dependence of the amount by which $s$ 
``sticks out'' of 
$\Span\{s_1,s_2\}$; that is, by 
the norm of the $\eta$-orthogonal 
projection of $s$ into the orthogonal 
complement of $\Span\{s_1,s_2\}$.
In order to determine the latter, we note 
\begin{lemma}
\label{thm:ParallelProjectionPlane}
Let $s_1,s_{2}\in \som$; the $\eta$-orthogonal projector into 
$\Span\{s_1,s_2\}$ is 
\begin{equation}
\label{eq:ParallelProjectionPlane}
\proj_{(s_1,s_2)}^\Vert=
\frac{s_1\otimes s_1+s_2\otimes s_2
-\gamma_{12}(s_1\otimes s_2+s_2\otimes s_1)}{\gamma^2_{12}-1}\,.
\end{equation}
\end{lemma}
\begin{proof}
We note that $\proj_{(s_1,s_2)}^\Vert$
is a linear combination of tensor 
products of elements taken from the set 
$\{s_1,s_2\}$ which is symmetric under 
exchange of $s_1$ and $s_2$. Hence it is 
a $\eta$-symmetric linear map whose kernel 
contains the $\eta$-orthogonal complement
of $\Span\{s_1,s_2\}$. An easy 
calculation shows 
$\proj_{(s_1,s_2)}^\Vert s_1=s_1$, hence  
also $\proj_{(s_1,s_2)}^\Vert s_2=s_2$ (symmetry) and therefore
 $\proj_{(s_1,s_2)}^\Vert\circ \proj_{(s_1,s_2)}^\Vert=\proj_{(s_1,s_2)}^\Vert$. 
\end{proof}

In passing we note that expression 
\eqref{eq:ParallelProjectionPlane}
immediately leads to 
\begin{equation}
\label{eq:ParallelProjectionPlane-s}
s_\Vert:=\proj_{(s_1,s_2)}^\Vert s=
\frac{(\gamma_{12}\gamma_2-\gamma_1) s_1+
      (\gamma_{12}\gamma_1-\gamma_2) s_2}
{\gamma_{12}^2-1}\,.
\end{equation}
Moreover, the $\eta$-orthogonal projector into the $\eta$-orthogonal complement 
of $\Span\{s_1,s_2\}$ is 
\begin{equation}
\label{eq:OrthogonalProjectionPlane-1}
\proj_{(s_1,s_2)}^\perp=\id_V-
\proj_{(s_1,s_2)}^\Vert\,,
\end{equation}
so that the squared norm of the 
projection 
\begin{equation}
\label{eq:OrthogonalProjectionPlane-2}
s_\perp:=\proj^\perp_{(s_1,s_2)}s
\end{equation}
becomes 
\begin{equation}
\label{eq:OrthogonalProjectionPlane-3}
\begin{split}
\Vert s_\perp\Vert^2
&=s_\perp\cdot s_\perp
=s\cdot s_\perp
=s\cdot s-s\cdot s_\Vert\\
&=\frac{1-\gamma_1^2-\gamma_2^2-\gamma_{12}^2+2\gamma_1\gamma_2\gamma_{12}}{\gamma_{12}^2-1}\,,
\end{split}
\end{equation}
where we used 
\eqref{eq:ParallelProjectionPlane-s}
in the last step. We rewrite this as 
\begin{equation}
\label{eq:OrthogonalProjectionPlane-4}
\gamma_1^2+\gamma_2^2-1=
-\Vert s_\perp\Vert^2(\gamma_{12}^2-1)
-\gamma^2_{12}+2\gamma_{12}\gamma_1\gamma_2\,,
\end{equation}
which we use to eliminate 
$\gamma_1^2+\gamma_2^2-1$ in the 
denominator on the right-hand side of 
\eqref{eq:BoostLink-5}. This leads to the 
desired final formula 
\begin{equation}
\label{eq:OrthogonalProjectionPlane-5}
\gamma=\gamma_{12}-\Vert s_\perp\Vert^2\
\frac{\gamma_{12}^2-1}{1+2\gamma_1\gamma_2-\gamma_{12}}\,.
\end{equation}

We note that the fraction on the 
right-hand side is non-negative 
and zero iff $\gamma_{12}=1$, i.e. iff 
$s_1=s_2$. Indeed, writing 
\begin{equation}
\label{eq:OrthogonalProjectionPlane-6}
s_1=\gamma_1(s+\vec\beta_1)
\quad\text{and}\quad
s_2=\gamma_2(s+\vec\beta_2)\,,
\end{equation}  
with $\beta_{1{,}2}\in T_s\som$, 
we have $\gamma_{12}=-s_1\cdot s_2=
\gamma_1\gamma_2(1-\vec\beta_1\cdot\vec\beta_2)$ and hence 
\begin{equation}
\label{eq:OrthogonalProjectionPlane-7}
1+2\gamma_1\gamma_2-\gamma_{12}=
1+\gamma_1\gamma_2(1+\vec\beta_1\cdot\vec\beta_2)>1\,.
\end{equation}  
Equation \eqref{eq:OrthogonalProjectionPlane-5}
therefore shows that $\gamma\leq\gamma_{12}$
with equality iff $s_\perp=0$, i.e. 
$s\in\Span\{s_1,s_2\}$.

A specific example may be helpful to 
develop some intuition of how $\gamma$
may deviate from $\gamma_{12}$. 
For that we assume that $s$ is tilted 
symmetrically against $s_1$ and $s_2$,
i.e that  
\begin{equation}
\label{eq:TiltDependence-1}
\gamma_1=\gamma_2=:\gamma_*\,.
\end{equation}
In that case \eqref{eq:BoostLink-5}
immediately leads to
\begin{equation}
\label{eq:TiltDependence-2}
\gamma=\frac{%
 2\gamma^2_*+\gamma_{12}-1}
{2\gamma^2_*-\gamma_{12}+1}\,.
\end{equation}
Note that from \eqref{eq:OrthogonalProjectionPlane-7}
we already know that the denominator is $>1$. For fixed $\gamma_{12}$ this gives
$\gamma$ as monotonically decreasing function of $\gamma_*$ that 
asymptotically approaches $\gamma=1$
(and hence the corresponding 
$\beta=\sqrt{1-\gamma^{-2}}$ approaching
$\beta=0$) as $\gamma_*$ tends to 
infinity; the graph is shown to the left 
in Fig.\,\ref{fig:GammaVersusTilt}. 
$\gamma$ assumes its maximum 
for the minimal $\gamma_*$, which is that 
where $s\in\Span\{s_1,s_2\}$ and 
due to $\gamma_1=\gamma_2$ just bisects 
$s_1$ and $s_2$. Hence  
$s=(s_1+s_2)/\sqrt{2(1+\gamma_{12})}$ 
and the minimal $\gamma_*$ is 
$\gamma_*=-s\cdot s_1=\sqrt{(1+\gamma_{12})/2}$
at which $\gamma$ takes its maximal value $\gamma=\gamma_{12}$. 

\begin{figure}[ht]
\includegraphics[width=.44\textwidth]{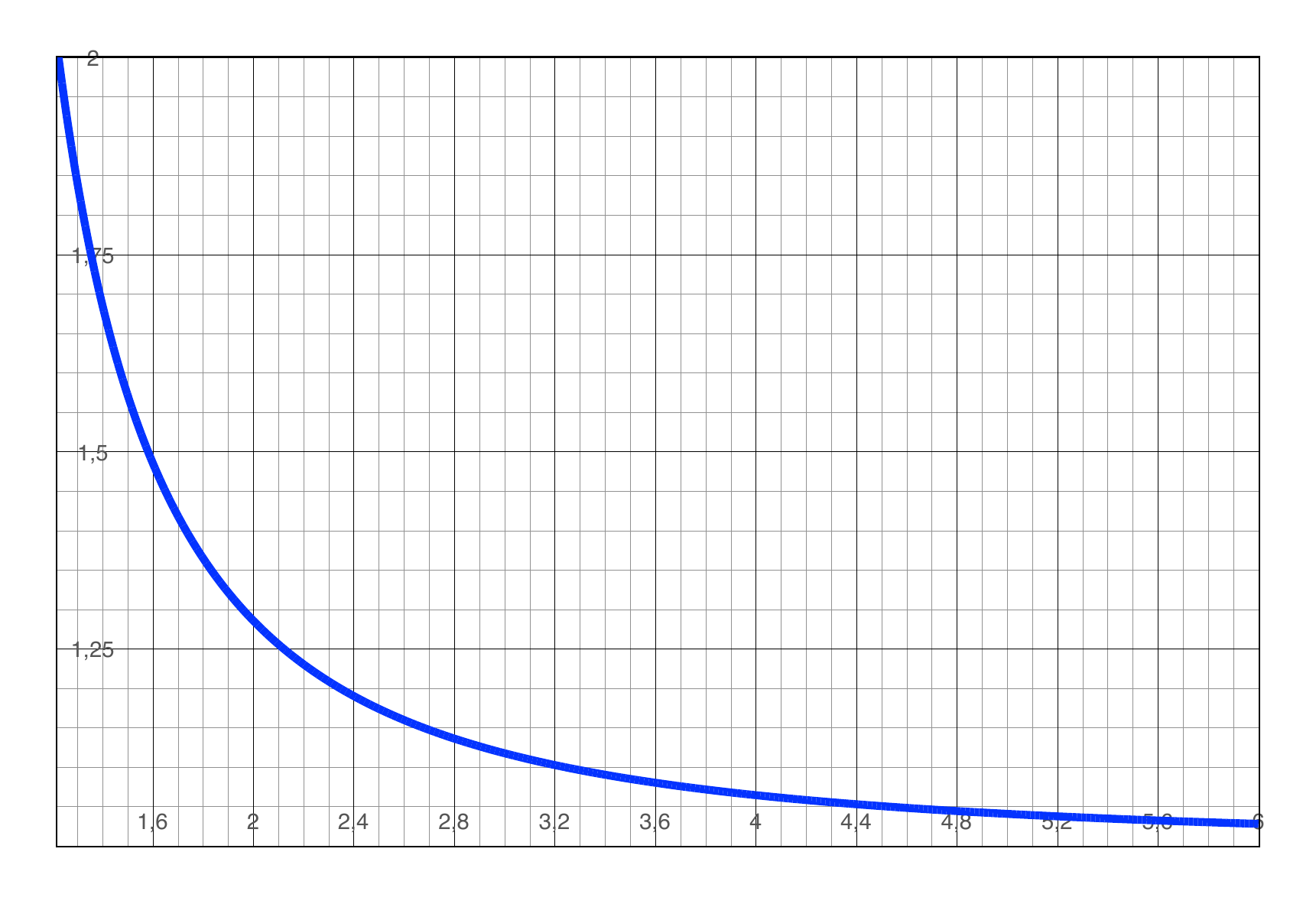}
\hspace{22pt}
\includegraphics[width=.46\textwidth]{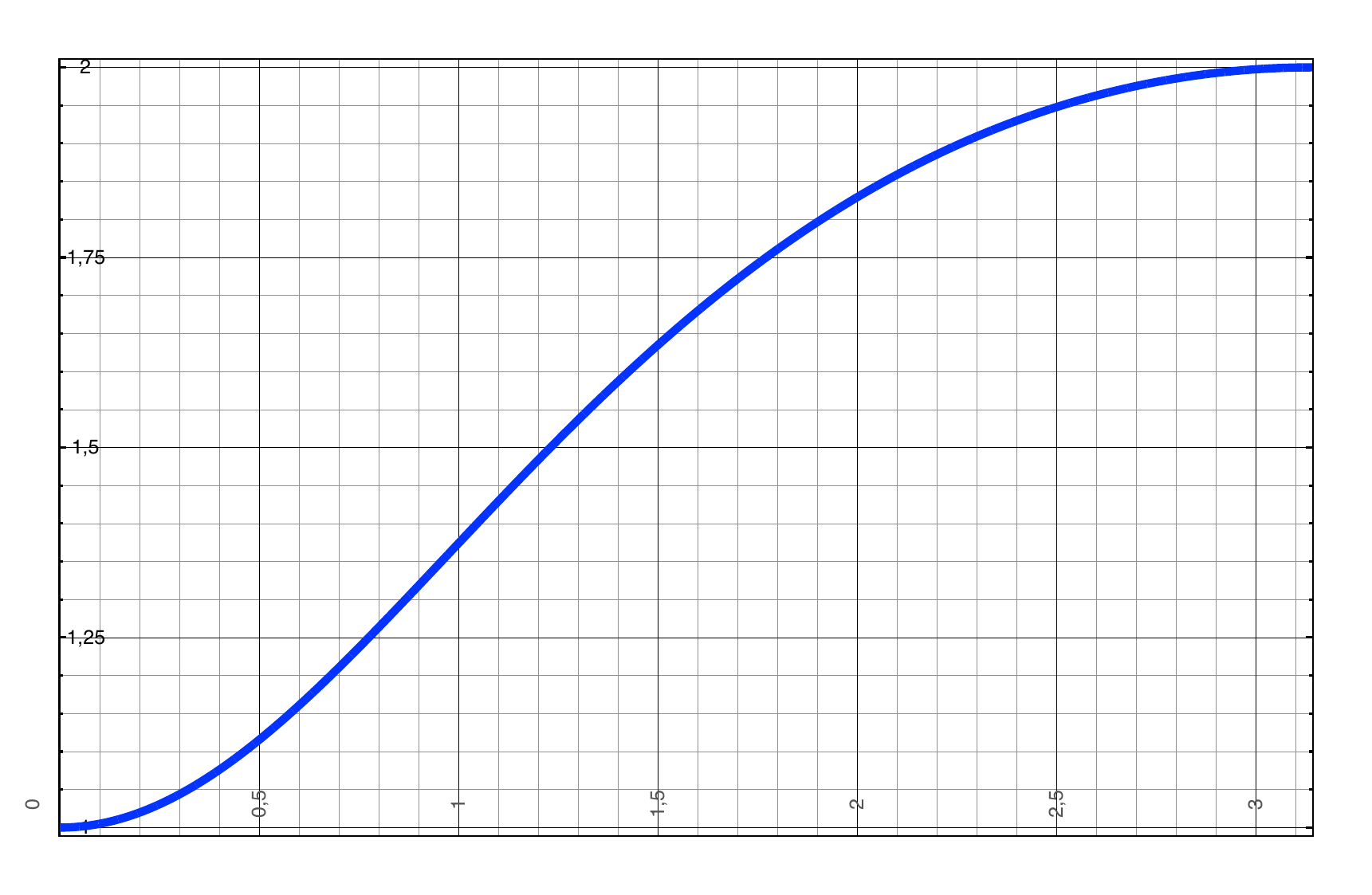}
\put(-275,27){\footnotesize $\gamma=\frac{2\gamma_*^2+(\gamma_{12}-1)}{2\gamma_*^2-(\gamma_{12}-1)}$}
\put(-3,8){\small $\varphi$}
\put(-167,98){\small $\gamma$}
\put(-110,27){\footnotesize $\gamma=\frac{\gamma_{12}(3-\cos\varphi)-\cos\varphi -1)}{\gamma_{12}(1+\cos\varphi)-3\,\cos\varphi +1}$}
\put(-195,8){\small $\gamma_*$}
\put(-353,98){\small $\gamma$}
\caption{\label{fig:GammaVersusTilt}
Graphs of $\gamma(\gamma_*)$
and $\gamma(\varphi)$ showing how $\gamma(s,s_1,s_2)$ decreases for 
increasing ``tilt'' of $s$ against 
$\Span\{s_1,s_2\}$.}
\end{figure}

Another way to parametrise 
$\gamma$ is to use
\eqref{eq:OrthogonalProjectionPlane-6}
where $\beta_1=\beta_2=:\beta_*$ and
$\varphi$ the angle between 
$\vec\beta_1$ and $\vec\beta_2$.  
We have 
\begin{equation}
\label{eq:TiltDependence-3}
\gamma_{12}=-s_1\cdot s_2
=\gamma^2_*(1-\beta_*^2\,\cos\varphi)
=\gamma_*^2(1-\cos\varphi)+\cos\varphi\,.
\end{equation}
using $\beta_*^2=(1-\gamma_*^{-2})$
in the last step. Solving this for
$\gamma^2_*$
\begin{equation}
\label{eq:TiltDependence-4}
\gamma^2_*=\frac{\gamma_{12}-\cos\varphi}
{1-\cos\varphi}
\end{equation}
we can use it to eliminate $\gamma^2_*$ 
in \eqref{eq:TiltDependence-2} in favour 
of $\varphi$:
\begin{equation}
\label{eq:TiltDependence-5}
\gamma=\frac{\gamma_{12}\,(3-\cos\varphi)-\cos\varphi-1}{\gamma_{12}\,(1+\cos\varphi)-3\cos\varphi+1}\,.
\end{equation}
Now $\gamma$ is a strictly monotonically
increasing function in $\varphi\in[0,\pi]$ 
connecting the minimum at 
$(\varphi=0,\gamma=1)$ with the maximum 
at $(\varphi=\pi,\gamma=\gamma_{12})$.
The graph is shown to the right in 
Fig.\,\ref{fig:GammaVersusTilt}.

\section{Comparison with 
Galilei-Newton spacetimes}
\label{sec:GalileiNewton}
It is interesting to compare the geometric 
developments above to the corresponding one 
in for Galilei-Newton spacetime. For that we 
start with a summary of the decomposition of a Lorentz transformation into a boost and a rotation,
without using the notion of a polar 
decomposition, which is really not 
essential here.

\subsection{
Boost-Rotation decomposition in SR rephrased
\label{sec:BoostRotationNoPolar}
}
Given $L\in\Lor$ we do the following: 
\begin{enumerate}
\item
Choose a state $s\in\som$ and let $s_1:=Ls$.
\item
Let $B:=B(s,s_1)$ as in \eqref{eq:RelBoost-3}, 
i.e. the unique Lorentz transformation that 
maps $s$ to $s_1$ and pointwise fixes the 
$\eta$-orthogonal complement to 
$\Span\{s,s_1\}$. The corresponding 
$\vec\beta(s,s_1)$ is then given by 
\eqref{eq:DefRelVel-2}.
\item
Define 
\begin{equation}
\label{eq:DefRotationLor-1}
R:=B^{-1}\circ L\,.
\end{equation}
Clearly, $R$ is again a Lorentz transformation 
that fixes s: $Rs=s$. Hence it is an 
element of the stabiliser subgroup 
\begin{equation}
\label{eq:DefRotationLor-2}
\mathrm{Stab}_s\bigl(\Lor\bigr)
:=\bigl\{L\in\Lor:Ls=s\bigr\}\,,
\end{equation} 
which consists of Lorentz transformations 
that map the orthogonal complement 
to $s$ (which is a 3-dimensional vector 
space with Euclidean inner product induced 
from $\eta$) isometrically into itself. 
Hence each $\mathrm{Stab}_s\bigl(\Lor\bigr)$ 
is isomorphic to $\group{SO}(3)$, but for 
$s\ne s'$ the corresponding stabiliser 
subgroups differ. In fact, they are conjugate 
subgroups in $\Lor$,
\begin{equation}
\label{eq:DefRotationLor-3}
\mathrm{Stab}_{s'}\bigl(\Lor\bigr)
=
B(s,s')\circ\mathrm{Stab}_s\bigl(\Lor\bigr)\circ
[B(s,s')]^{-1}\,,
\end{equation} 
where $B(s,s')$ is as in 
$\eqref{eq:RelBoost-3}$ for $s_1=s'$. 
\item
Rewrite \eqref{eq:DefRotationLor-1} as  
\begin{equation}
\label{eq:DefRotationLor-4}
L=B\circ R
\end{equation} 
with both factors, $B$ and $R$, depending on $s$. 
The $s$-dependent decomposition \eqref{eq:DefRotationLor-4} 
may now be read as a polar decomposition with 
respect to the $s$-dependent Euclidean metric 
$g=\eta+2\sigma\otimes\sigma$, where 
$\sigma:=\eta_\downarrow(s)$, but that 
interpretation may be regarded as redundant.
\end{enumerate}

In Section\,\ref{sec:BoostRotationDecGN} 
below will describe the analogous procedure 
in the Galilei-Newton setting, where we will 
also find a decomposition like \eqref{eq:DefRotationLor-4}. There, in 
contrast, the factor $B$ will \emph{not} 
depend on the choice of $s$ whereas the 
situation for the ``rotational'' factor $R$
is just as in the SR case. But now this 
factorisation has no obvious interpretation 
as polar decomposition.

\subsection{
Galilei-Newton spacetime
\label{sec:GalileiNewtonST}
}
Like Minkowski spacetime, Galilei-Newton 
spacetime is an affine space 
$(M,V,+)$.\footnote{For readers unfamiliar with 
this notation, we explain the meaning of an affine 
space characterised by the triple $(M,V,+)$ in 
Appendix\,\ref{sec:AffineStructures}. Suffice 
it to say here that $M$ is a set (of events) 
and $V$ a real vector space that acts simply 
transitively by an action called ``$+$'' on 
$M$; i.e. $M\times V\ni (m,v)\mapsto m+v\in M$.}
However, their geometric structures differ.  
In case of Minkowski spacetime we had 
$(\eta,\uparrow)=(\text{spacetime metric},
\text{time-orientation})$, which for 
Galilei-Newton spacetime is replaced by 
a  pair $(\tau, h)$. Here $\tau\in V^*$ 
is an oriented time-difference function 
and $h\in [\kernel(\tau)]^*\otimes [\kernel(\tau)]^*$ is a symmetric, positive 
definite inner product on the 
$\kernel(\tau)$. We will explain these 
structures in turn. 
\begin{itemize}
\item
$\tau$ is an oriented time metric, i.e. 
it allows to assign an oriented time 
difference to any ordered pair $(p,q)$ of 
spacetime points, given by $\tau(p-q)$.
It can assume positive as well 
as negative values -- hence ``oriented''. 
A vector $v\in V$ is called future-pointing 
if $\tau(v)>0$ and past-pointing if 
$\tau(v)<0$. Correspondingly, $p\in M$ is 
called to the future or past of $q\in M$ if 
$(p-q)$ is future-pointing or past-pointing, 
respectively. Two events $p,q$ are called 
simultaneous iff their time difference 
vanishes, $\tau(p-q)=0$. Simultaneity defines 
an absolute\footnote{
i.e. invariant under the automorphism group 
of spacetime; compare \citep{Giulini:2002a}.}
equivalence relation given by: $p\sim q\Leftrightarrow (p-q)\in\kernel(\tau)$. 
The equivalence class $[p]$ of a point $p$ is 
then simply given by the 3-dimensional 
affine hyperplane $[p]=p+\kernel(\tau)$. 
For the sake of notational ease we shall 
write
\begin{equation}
\label{eq:DefV0}
V_0:=\kernel(\tau)
=\bigl\{v\in V:\tau(v)=0\bigr\}\,.
\end{equation}
Its associated dual space will be called 
$V_0^*$. An overall orientation of $V$ will 
induce a  orientation of $V_0$ in view of 
$\tau$.\footnote{A basis $\{e_1,e_2,e_3\}$ of 
$V_0$ is positively oriented iff its 
completion $\{e_0,e_1,e_2,e_3\}$ to a 
basis of $V$ with $\tau(e_0)>0$ is positively
oriented in $V$.} 
 
\item
$h\in V_0^*\otimes V_0^*$ is a positive definite
symmetric bilinear form, i.e. a  Euclidean metric,
on $V_0$. It defines a proper distance
function on any equivalence class of mutually 
simultaneous events through 
$\Vert p-q\Vert:=\sqrt{h(p-q,p-q)}$.
No spatial distance is associated to 
non-simultaneous events. 
\end{itemize}

In analogy to Definition\,\ref{def:StateOfMotion-1}
we now define the ``unit timelike vectors'' 
$V_1$, the ``future-oriented unit timelike 
vectors'' $V_1^+$, and the ``states of motion''  
as follows:  
\begin{subequations}
\label{eq:DefSoM-GN1}
\begin{alignat}{2}
\label{eq:DefSoM-GN1-a}
& V_1&&\,:=\,\bigl\{v\in V:\vert\tau(v)\vert=1\bigr\}\,,\\
\label{eq:DefSoM-GN1-b}
& V^+_1&&\,:=\,\bigl\{v\in V:\tau(v)=1\bigr\}\,.
\end{alignat}
\end{subequations}

\begin{definition}
\label{def:StateOfMotionGal}
The set $\som$ of \textbf{states of motion} 
is identified with $V^+_1$:
\begin{equation}
\label{eq:DefSoM-GN2}
\som:=V_1^+\,.
\end{equation}
$\som$ is a 3-dimensional real affine space 
over the Euclidean vector space $(V_0,h)$.  
At the same time it can also be regarded 
as a 3-dimensional Riemannian manifold with 
a flat Riemannian metric $h$ and hence a 
notion of global parallelism. 
\end{definition}

Whereas there clearly is a natural 
embedding  
\begin{equation}
\label{eq:GalileiEmbedding}
i:V_0\hookrightarrow V\,,
\end{equation}
there is no \emph{naturally} given projection 
map $V\rightarrow V_0$. The selection of such 
a map is equivalent to picking an element 
$s\in \som$. The corresponding projections 
from $V$  onto $\Span\{s\}$ and 
onto $V_0$ are then, respectively,  given 
by\footnote{In contrast to  
section\,\ref{sec:KinematicalSetting}, 
where due to the existence of a 
non-degenerate bilinear form $\eta$ we 
could identify $\End(V)$ with $V\otimes V$
in order to simplify notation, we here use 
the natural identification 
$\End(V)=V\otimes V^*$.}
\begin{subequations}
\label{eq:DefProjEndoGal}
\begin{alignat}{2}
\label{eq:DefProjEndoGal-a}
&\proj_s^\Vert 
&&\,=\,s\otimes \tau\,,\\
\label{eq:DefProjEndoGal-b}
&\proj_s^\top 
&&\,=\,\id_V-s\otimes \tau\,.
\end{alignat}
\end{subequations}
These should be compared to equations 
\eqref{eq:DefProjEndo}. An important 
difference is that all $\proj_s^\top$ 
in \eqref{eq:DefProjEndoGal-b}
project onto the same vector space 
$V_0$, independent of the state $s$.
In contrast, $\proj_s^\perp$ in 
\eqref{eq:DefProjEndo-b} projects 
onto the $\eta$-orthogonal complement 
of $s$, which clearly does depend on $s$. 
This is the reason why we made the notational 
change from $\proj_s^\perp$ in \eqref{eq:DefProjEndo-b} 
to $\proj_s^\top$ in \eqref{eq:DefProjEndoGal-b}: 
$V_0$ here is ``transversal'' (denoted by $\top$) 
to $\Span\{s\}$ but not in any defined sense 
``orthogonal'' (as $\perp$ would suggest). 

In analogy to Definition\,\ref{def:RelVel-1},
the relative velocity between $s_1\in\som$ 
and $s\in\som$ is now defined as in 
\eqref{eq:DefRelVel-1}, except that the length 
of vectors in $\Span\{s\}$ is measured 
by $\tau$, which is always 1:

\begin{definition}
\label{def:RelVelGN-1}
Given two states of motion $s$ and $s_1$. 
The \textbf{relative velocity between $s_1$ 
and $s$, judged from $s$,} is defined by 
\begin{equation}
\label{eq:DefRelVelGal-1}
\vec v(s,s_1):=\frac{\proj_s^\top(s_1)}{\tau\bigl(\proj_s^\Vert(s_1)\bigr)}
=s_1-s \in V_0.
\end{equation}
This, in geometric terms, is just the 
ordinary definition of velocity in 
Newtonian physics. 
\end{definition}
The second reference to $s$ expressed in the 
phrase ``judged from $s$'' is now indeed 
redundant due to the fact that all such 
relative velocities are members of the same 
space $V_0$. Clearly, we could have considered
the right-hand side of \eqref{eq:DefRelVelGal-1}
as element of $T_s\som$, as in \eqref{eq:DefRelVel-1}. But due to the 
flatness of $(\som,h)$ there are now
natural isomorphisms between all tangent 
spaces $T_s\som$ given by (path-independent) parallel transport.

\begin{remark}
\label{rem:GalileiStructure}
The geometric formulation of a Galilei-Newton
structure goes back to \citet[\S\,18]{Weyl:RZM1}.
Later generalisations are due to \citet{Friedrichs:1928}, \citet{Dombrowski.Horneffer:1964}, 
and also \citet{Kuenzle:1972}. Sometimes instead of 
$\tau$ only $\tau\otimes\tau$ is prescribed, i.e.
a degenerate ``time metric'' of rank one,  
which is equivalent to $\tau$ without the time 
orientation. Also, often $h\in V_0^*\otimes V_0^*$ is 
replaced with a $\bar h\in V\otimes V$ whose 
kernel is just $\Span\{\tau\}$. In fact, 
most modern texts follow this formulation, 
presumably in order to not deal with tensor 
fields over proper subspaces. Algebraically
both formulations are equivalent. In fact, 
there is a bijection between the sets of 
non-degenerate bilinear forms on $V_0$ and 
simply degenerate bilinear forms on $V^*$ 
whose kernel is $\Span\{\tau\}$. 
The bijective correspondence is this: 
Consider the natural isomorphism 
$V_0^*\otimes V_0^*\cong \mathrm{Lin}(V_0,V_0^*)$ 
und accordingly regard $h$ als Element von of  
$\mathrm{Lin}(V_0,V_0^*)$. Since $h$ is non-degenerate there exists the inverse map 
$h^{-1}\in \mathrm{Lin}(V_0^*,V_0)$. Let further 
$i:V_0\rightarrow V$ be the natural embedding and 
$i^*: V^*\rightarrow V_0^*$ its dual. The latter 
is given by 
$\lambda\mapsto i^*(\lambda):=\lambda\circ i$. 
Hence the kernel of $i^*$ is just given by 
$\Span\{\tau\}$.  Now we define 
$\bar h:=i\circ h^{-1}\circ i^*\in\mathrm{Lin}(V^*,V)\cong V\otimes V$. Since $i$ is injective 
and $h^{-1}$ an isomorphism, the kernel of 
$\bar h$ equals the kernel of $i^*$, i.e. 
$\Span\{\tau\}$.
\end{remark}

\subsection{
Automorphims of Galilei-Newton spacetime
\label{sec:AutGN}}
Like $\ILor\cong V\rtimes\Lor$, the identity component of 
the inhomogeneous Lorentz group (also called the 
proper orthochronous Poincaré group), is the 
automorphism group of $(M,V,+,\eta,o_V,o_T)$, 
where $o_V$ is the overall orientation of $V$ 
and $o_T$ is the time orientation, we define 
the inhomogeneous Galilei group to be the 
automorphism group of the geometric 
Galilei-Newton structure. 
\begin{definition}
\label{def:GalileiGroup}
The \textbf{inhomogeneous Galilei group} is 
the automorphism group of $(M,V,+,\tau,h,o_V)$; 
we write 
$\IGal:=\mathrm{Aut}(M,V,+,\tau,h,o_V)$.
Here $o_V$ stands for an orientation of $V$. 
Unlike in the special-relativistic case we 
do not need to also specify a time 
orientation $o_T$, since that is already 
provided by $\tau$. This group is 
a subgroup of the affine group 
$\mathrm{Aut}(M,V,+)$. Again it is isomorphic 
to a semi-direct product $\IGal\cong V\rtimes\Gal$,
where $\Gal$ may be considered as subgroup 
of $\group{GL}(V)$:
\begin{equation}
\label{eq:DefGal}
\begin{split}
\Gal:=&\mathrm{Aut}(V,\tau,h,o_V)\\
:=&\bigl\{G\in\mathrm{GL}(V):\tau\circ G=\tau,\ 
h\circ G\vert_{V_0}\times G\vert_{V_0}=h,\
\det(G)>0\bigr\}\,.
\end{split}
\end{equation} 
\end{definition} 
Note that any $G\in\group{GL}(V)$ that preserves 
$\tau$ also preserves $V_0=\kernel(\tau)$ 
i.e. it defines an element $G\vert_{V_0}\in\group{GL}(V_0)$ 
by restricting $G$ to $V_0$. That element is then 
required to be in $\group{SO}(V_0,h)$ by the 
second and third condition in \eqref{eq:DefGal} 
(preserving overall- and  time-orientation implies
that it also preserves space orientation). It is 
easy to see that this restriction-map defines a 
surjective homomorphism of groups, i.e. a 
``group projection'': 
\begin{equation}
\label{eq:GalGroupHom}
\pi:\group{Gal}\rightarrow\group{SO}(V_0,h)\,,\quad
G\rightarrow \pi(G):=G\vert_{V_0}\,.
\end{equation}

\begin{definition}
\label{def:BoostGal}
Elements in $\kernel(\pi)$ are called 
\textbf{boosts}. 
\end{definition}
\begin{proposition}
\label{thm:BoostGroupGal}
The set of boosts is given by 
\begin{equation}
\label{eq:SetBootsGal}
\group{Gal}_B:=\kernel(\pi)
=\bigl\{\id_V+\vec v\otimes\tau:\vec v\in V_0\bigr\}\,.
\end{equation}
$\group{Gal}_B$ is an abelian normal subgroup of
$\group{Gal}$ which is isomorphic to $V_0$ (considered
as abelian group with group multiplication given 
by vector addition). 
\end{proposition}
\begin{proof}
Being the kernel of a group homomorphism  
$\group{Gal}_B$ is clearly a normal subgroup.
Let $\{e_0,e_1,e_2,e_3\}$ be an ``adapted'' 
basis of $V$, which means that $\tau(e_0)=1$
and $\tau(e_a)=0$ for $a=1,2,3$. Let further 
$\{\theta^0,\theta^1,\theta^2,\theta^3\}$ 
be the dual basis, i.e. 
$\theta^\alpha(e_\beta)=\delta^\alpha_\beta$; 
then $\theta^0=\tau$. Writing
\begin{equation}
\label{eq:BoostGal-1}
G=\tensor{G}{^\alpha_\beta}\,e_\alpha\otimes\theta^\beta\,,
\end{equation}
the condition $\tau\circ G=\tau$ is equivalent to 
$\tensor{G}{^0_\beta}\theta^\beta=\theta^0$, i.e.
$\tensor{G}{^0_0}=1$ and $\tensor{G}{^0_a}=0$ ($a=1,2,3$).
Moreover,
\begin{equation}
\label{eq:BoostGal-2}
\pi(G)=G\vert_{V_0}
=\tensor{G}{^a_b}\,e_a\otimes\theta^b
\quad (a,b=1,2,3)\,,
\end{equation}
so that $G\in\kernel(\pi)$ iff $\tensor{G}{^a_b}=\delta^a_b$. 
Taken together, the most general element in $\kernel(\pi)$ is 
of the form
\begin{equation}
\label{eq:BoostGal-3}
G=e_0\otimes\theta^0+\delta^a_be_a\otimes\theta^b+\tensor{G}{^a_0}e_a\otimes\theta^0=\id_V+\vec v\otimes\tau\,,
\end{equation}
where $\vec v:=\tensor{G}{^a_0}e_a\in V_0$. 
Such elements form an abelian subgroup isomorphic to 
$V_0$ inside $\group{Gal}$. The corresponding embedding (injective group homomorphism) 
is  
\begin{equation}
\label{eq:BoostEmbeddingGN}
B: V_0\rightarrow\group{Gal}\,,\quad
B(\vec v):=\id_V+\vec v\otimes\tau\,,
\end{equation}
which maps isomorphically onto its image 
$\group{Gal}_B\subset\group{Gal}$. The relations
$B(\vec 0)=\id_V$ and 
$B(\vec v_1)\circ B(\vec v_2)=B(\vec v_1+\vec v_2)$
are quite obviously satisfied.
\end{proof}
\begin{proposition}
\label{thm:RotGroupGal}
For any $s\in\som$ there is an embedding 
(injective group homomorphisms)
\begin{subequations}
\label{eq:GalRotEmbedding1}
\begin{equation}
\label{eq:GalRotEmbedding1-a}
\sigma_s: \group{SO}(V_0,h)\rightarrow\group{Gal}\,,\quad
D\mapsto\sigma_s(D):=s\otimes\tau+D\circ \proj_s^\top\,,
\end{equation} 
such that 
\begin{equation}
\label{eq:GalRotEmbedding1-b}
\pi\circ\sigma_s=\id_{\group{SO}(V_0,h)}\,.
\end{equation} 
\end{subequations}
\end{proposition}
\begin{proof}
We first check hat $\sigma_s$ is a group homomorphisms. 
From \eqref{eq:DefProjEndoGal-b} one immediately sees 
that $\sigma_s(\id_{V_0})=\id_V$, and from  
\eqref{eq:GalRotEmbedding1-a} that
\begin{equation}
\label{eq:GalRotEmbedding2}
\begin{split}
\sigma_s(D_1)\circ\sigma_s(D_2)
&=(s\otimes\tau+D_1\circ\proj_s^\top)\circ
(s\otimes\tau+D_2\circ\proj_s^\top)\\
&=s\otimes\tau+(D_1\circ D_2)\circ\proj_s^\top\\
&=\sigma_s(D_1\circ D_2)\,,
\end{split}
\end{equation} 
where we used $\tau\circ D_2\circ\proj_s^\top=0$, 
$\proj_s^\top(s)=0$, and 
$\proj_s^\top\circ D_2\circ\proj_s^\top=D_2\circ\proj_s^\top$.
Finally,
\begin{equation}
\label{eq:GalRotEmbedding3}
\pi\bigl(\sigma_s(D)\bigr)=
(s\otimes\tau+D\circ\proj_s^\top)\big\vert_{V_0}=D\,,
\end{equation}
since $\tau\vert_{V_0}=0$ and 
$\proj_s^\top\vert_{V_0}=\id_{V_0}$.
\end{proof}
\begin{corollary}
There is a splitting exact sequence 
\begin{equation}
\label{eq:ExactSequenceSplit}
\begin{tikzcd}
\{1\}\arrow[r] 
& V_0\arrow[r,tail, "B"] 
& \group{Gal}\arrow[r,bend left=0, two heads, "\pi"]
& \group{SO}(V_0,h)\arrow[l,bend left=20, tail, "\sigma_s", end anchor=south, start anchor=south]\arrow[r]
&\{1\}	\,.
\end{tikzcd}
\end{equation}
Here $\{1\}$
denotes the trivial group with only one 
element. ``Exact''  means that at each node 
of the sequence of maps the image of the 
arriving equals the kernel of the departing 
map. For \eqref{eq:ExactSequenceSplit} this is equivalent to the three statements that 1)~$B$ is an injective group 
homomorphism (also called an ``embedding'', symbolised by 
a tailed arrow), that 2)~$\pi$ is a surjective group 
homomorphism (symbolised by a double-headed arrow), 
and that 3)~$\image(B)=\group{Gal}_B=\kernel(\pi)$.
That the exact sequence is ``splitting'' means 
that there is an injective group homomorphism 
-- here denoted by $\sigma_s$ -- from  
$\group{SO}(V_0,h)$ to $\group{Gal}$,
such that $\pi\circ\sigma_s=\id_{\group{SO}(V_0,h)}$
All these properties taken together are 
equivalent to the statement that 
$\group{Gal}$ is a semi-direct product 
of $V_0$ and $\group{SO}(V_0,h)$, usually 
denoted by 
\begin{equation}
\label{eq:GalSemiDirProd}
\group{Gal}=V_0\rtimes \group{SO}(V_0,h)\,.
\end{equation}
Note that the abelian normal subgroup of 
boosts, $B(V_0)=\group{Gal}_B$, is defined independent of a choice of $s\in\som$. In 
contrast, any selection of a subgroup 
of orthogonal spatial transformations within 
$\group{Gal}$, i.e. the image of 
$\group{SO}(V_0,h)$ under $\sigma_s$ in 
$\group{Gal}$, does depend on a choice of 
$s\in\som$. 
\end{corollary}

The group of boosts, $\group{Gal}_B$, has a simple 
transitive action on $\som=V_1$; that is, for any 
ordered pair $(s_1,s_2)$ of states there is precisely one boost that maps $s_1$ to $s_2$: 
\begin{equation}
\label{eq:BoostGal-4}
B(\vec v)s_1=s_2 \quad\text{for}\quad 
\vec v=\vec v(s_1,s_2):=s_2-s_1. 
\end{equation}
Hence we call $s_2-s_1$ the link-velocity between 
$s_1$ and $s_2$. There is now no need to add the phrase 
``judged from $s_1$'' since there is exactly one 
boost linking $s_1$ with $s_2$, whereas in the SR case 
there were many, one for each additional reference
state $s$, the choice of which determined the set of 
boosts to choose the linking transformation from. 
Note that since the composition of pure boosts result in the simple addition of vectors in $V_0$, we have, e.g.,  
\begin{equation}
\label{eq:VelAddGalilei-1}
\vec v(s,s_3)=\vec v(s,s_1)+\vec v(s,s_2)
=(s_1-s)+(s_2-s)\,.
\end{equation}
This can be easily solved for $\vec v(s,s_2)$ 
or $\vec v(s,s_1)$, unlike the 
special-relativistic case, which could only be 
solved for $\vec\beta_2$ in an elementary fashion,
as shown in \eqref{eq:SR-VelAdd-Inv}. The 
corresponding solution in the Galilei-Newton
case is 
\begin{equation}
\label{eq:VelAddGalilei-2}
\vec v(s,s_2)=\vec v(s,s_3)-\vec v(s,s_1)
=(s_3-s)-(s_1-s)=s_3-s_1\,,
\end{equation}
where we used the laws of differences in affine 
space (compare the discussion below equations 
\eqref{eq:DefAffineSpace-3} of 
Appendix\,\ref{sec:AffineSpaces}). 
This shows explicitly how the dependence 
on $s$ drops out -- in contrast to the 
special-relativistic case!

\subsection{
Boost-Rotation decomposition for the 
Galilei group
\label{sec:BoostRotationDecGN}
}
The steps outlined in 
Section\,\ref{sec:BoostRotationNoPolar} for the 
SR case can now be translated almost literally 
to the Galilei-Newton case. Indeed, given $G\in\group{Gal}$, we do the following: 
\begin{enumerate}
\item
Choose a state $s\in\som$ and let $s_1:=Gs$.
\item
Let $B:=B(\vec v)$ as in \eqref{eq:BoostEmbeddingGN}
with $\vec v=s_1-s$, i.e.
\begin{equation}
\label{eq:BoostRotGN-1}
B:=B(s,s_1)=(s_1-s)\otimes\tau+\id_V\,.
\end{equation}
\item
Define 
\begin{equation}
\label{eq:BoostRotGN-2}
R:=B^{-1}\circ G\,.
\end{equation}
Note that 
$\bigl(B(\vec v)\bigr)^{-1}=B(-\vec v)$. 
Clearly, $R$ is again a Galilei transformation 
that fixes $s$: $Rs=s$. Hence 
\begin{equation}
\label{eq:BoostRotGN-3}
\mathrm{Stab}_s\bigl(\group{Gal}\bigr)
:=\bigl\{G\in\group{Gal}:Gs=s\bigr\}\,,
\end{equation} 
consist of Galilei transformations that map 
$(V_0,h)$ isometrically into itself preserving 
orientation. Hence  $\mathrm{Stab}_s\bigl(\group{Gal}\bigr)\cong\group{SO}(V_0,h)$, but for $s\ne s'$ 
the corresponding stabiliser subgroups are 
conjugate subgroups in $\group{Gal}$:
\begin{equation}
\label{eq:BoostRotGN-4}
\mathrm{Stab}_{s'}\bigl(\group{Gal}\bigr)
=
B(s,s')\circ\mathrm{Stab}_s\bigl(\group{Gal}\bigr)\circ[B(s,s')]^{-1}\,,
\end{equation} 
where $B(s,s')$ is as in 
\eqref{eq:BoostRotGN-1} for $s_1=s'$. 
\item
Rewrite \eqref{eq:BoostRotGN-2} as  
\begin{equation}
\label{eq:BoostRotGN-5}
G=B\circ R
\end{equation}
with $R$ -- but, unlike the SR case,  
not $B$ -- depending on $s$. 
\end{enumerate}

The decomposition \eqref{eq:BoostRotGN-5} is just 
the one that follows generally from the semi-direct 
product structure of $\Gal$, i.e. the splitting 
of \eqref{eq:ExactSequenceSplit}, as explained in 
Appendix\,\ref{sec:SemiDirect}. Again, as the 
splitting $\sigma_s$ in \eqref{eq:ExactSequenceSplit} 
depends on $s\in\som$, so does the decomposition
(as far is $R$ is concerned).   

Now, according to a general argument also recalled 
in Appendix\,\ref{sec:SemiDirect}, we consider for 
any $G\in\Gal$ the element 
$\sigma_s\bigl(\pi(G)\bigr)\in\Gal$. According to 
\eqref{eq:GalGroupHom} and \eqref{eq:GalRotEmbedding1-a}
we have 
\begin{equation}
\label{eq:BoostRotGN-6}
\begin{split}
\sigma_s\circ\pi(G)
&=\sigma_s(G\vert_{V_0})\\
&=s\otimes\tau+G\vert_{V_0}\circ\proj^\top_s\\
&=s\otimes\tau+G\circ\proj^\top_s\\
&=G-(Gs-s)\otimes\tau\,.
\end{split}
\end{equation}
Here we used \eqref{eq:GalGroupHom} in the first 
line, \eqref{eq:GalRotEmbedding1-a} in the second, 
the fact that we can simply drop the restriction to 
$V_0$ if we right-compose $G$ with the projection 
$P_s^\top$ onto $V_0$ in the third, and, finally,  
\eqref{eq:DefProjEndoGal-b} in the fourth. 
Hence, 
\begin{equation}
\label{eq:BoostRotGN-7}
\sigma_s\bigl(\pi(G)\bigr)\circ G^{-1}
=\id_V-(Gs-s)\otimes\tau
=B(-\vec v)\,,
\end{equation}
where for the fist equality we used 
$\tau\circ G^{-1}=\tau$ (obvious since $G^{-1}\in\Gal$) 
and \eqref{eq:BoostEmbeddingGN} for the second,
denoting $\vec v:=Gs-s$. Hence, since 
$\bigl(B(\vec v)\bigr)^{-1}=B(-\vec v)$, we get 
\eqref{eq:BoostRotGN-5} with 
\begin{equation}
\label{eq:BoostRotGN-8}
R:=\sigma_s\bigl(\pi(G)\bigr)
\quad\text{and}\quad
B=B(\vec v)=G\circ\Bigl(\sigma_s\bigl(\pi(G)\bigr)\Bigr)^{-1}\,.
\end{equation}
We leave it to the reader to find out whether there 
is a ($s$-dependent) euclidean metric $g_s$ on $V$ 
(built from $\tau$ and $h$ in an $s$-dependent 
fashion) with respect to which $G=B\circ R$ becomes 
a polar decomposition. 
\newpage 

\begin{appendices}

\section{Polar decomposition}
\label{sec:PolarDecomposition}
Let $V$ be an $n$-dimensional real vector 
space and $g:V\times V\rightarrow\reals$ 
a positive-definite symmetric bilinear 
form  (also known as Euclidean inner 
product). Let $\group{End}(V)$ denote 
the linear space of all endomorphisms 
$V\rightarrow V$ and $\group{GL}(V)$ 
the subset of all invertible elements 
in $\group{End}(V)$,
\begin{equation}
\label{eq:PD-1}
\group{GL}(V):=\{A\in\group{End}(V):\det(A)\ne 0\}\,.
\end{equation}
The set $\group{GL}(V)$ is a group under composition 
of maps and is called the \emph{General Linear Group}.

The Euclidean inner product $g$ defines a map  
\begin{equation}
\label{eq:PD-2}
\begin{split}
\dagger:\group{GL}(V)&\rightarrow\group{GL}(V)\\
A&\mapsto A^\dagger
\end{split}
\end{equation}
through 
\begin{equation}
\label{eq:PD-3}
g(Av,w)=g(v,A^\dagger W)\quad\forall\,v,w\in V\,.
\end{equation}
Symmetry of $g$ implies that $\dagger$ is an involution, 
i.e. $\dagger\circ\dagger=\id_{\group{GL}(V)}$, or 
$(A^\dagger)^\dagger=A$ for any $A\in\group{GL}(V)$. 
Also, $\dagger$ is a group anti-isomorphism of 
$\group{GL}(V)$, i.e. for any $A,B\in\group{GL}(V)$
we have\footnote{In this subsection we abbreviate the 
composition of maps $A,B$ in $\group{GL}(V)$ by 
juxtaposition, i.e. we write $AB$ instead of $A\circ B$. Accordingly, $B^2$ means $BB=B\circ B$, etc.} 
$(\id_V)^\dagger=\id_V$ and 
$(AB)^\dagger=B^\dagger A^\dagger$ (the ``anti-'' 
denoting the reversal of orders). 

The subset of fixed points in $\group{GL}(V)$ under 
$\dagger$ is 
\begin{equation}
\label{eq:PD-4}
\group{Sym}(V,g):=\{A\in\group{GL}(V):A=A^\dagger\}\,.
\end{equation}
It is called the set of \emph{symmetric elements}
in $\group{GL}(V)$. Note that 
$\group{Sym}(V,g)$ is not a subgroup, i.e. 
if $A=A^\dagger$ and $B=B^\dagger$ then 
$(AB)^\dagger=B^\dagger A^\dagger=BA$
which does not equal $AB$ unless $A$ and $B$ 
commute.  

Another subset is  
\begin{equation}
\label{eq:PD-5}
\group{Pos}(V,g):=
\bigl\{
A\in\group{GL}(V): g(v,Av)>0\,\,
\forall\,v\in V\backslash\{0\}
\bigr\}\,.
\end{equation} 
Note that $A\in\group{Pos}(V,g)$ implies $A^\dagger\in\group{Pos}(V,g)$, i.e. $\dagger$ maps $\group{Pos}(V,g)$ to itself. 

Of interest to us is the intersection
\begin{equation}
\label{eq:PD-6}
\group{PS}(V,g):=\group{Pos}(V,g)\cap\group{Sym}(V,g)
\end{equation}
of elements in $\group{GL}(V)$ which are at the 
same time \emph{positive} and \emph{symmetric}. 
A standard result is
\begin{lemma}
\label{thm:MatrixSquareRoot}
For any $A\in\group{PS}(V,g)$ there exists a 
unique $B\in\group{PS}(V,g)$ such that $A=B^2$,
called its ``square root''. $B$ is also 
denoted by $\sqrt{A}$ or $A^{1/2}$.
\end{lemma}
\begin{proof}
Since $A$ is symmetric there exists an orthonormal 
basis $\{e_1,\dots,e_n\}$ of $V$ that diagonalises 
$A$; i.e. $g(e_a,e_b)=\delta_{ab}$
and $A(e_a)=\lambda_ae_a$ (no summation over $a$). 
Since $A$ is positive all eigenvalues $\lambda_a$ 
are positive. We define $B$ through 
$B(e_a)=\sqrt{\lambda_a}\,e_a$ (no summation over $a$), 
which clearly satisfies $A=B^2$, showing existence. 
To prove uniqueness, note that the $B$ just defined 
commutes with $A$, i.e. $AB=BA$. Hence there is a 
polynomial function $p$ of degree at most $n-1$ 
($n=\dim(V)$) such that $p(A)=B$. Indeed, if 
$(\lambda_1,\cdots,\lambda_k)$ with $k\leq n$ is 
the maximal number of pairwise distinct eigenvalues 
of $A$, we can choose the interpolating Lagrange 
polynomial for the $k$ pairs $\bigl\{(\lambda_1,\sqrt{\lambda_1}),\dots,(\lambda_k,\sqrt{\lambda_k})\bigr\}$.\footnote{%
The general construction of such a polynomial 
is as follows: let 
$\{(x_1,y_1),\cdots,(x_k,y_k)\}\subset\reals^2$ be
any $k$ points with pairwise different $x$-values; 
i.e. $i\ne j\Rightarrow x_i\ne x_j$. For each 
$1\leq j\leq k$ we define 
$\ell_j(x):=\prod_{i=1,i\ne j}^k (x-x_i)/(x_j-x_i)$, 
which is a polynomial of degree $(k-1)$ satisfying 
$\ell_j(x_i)=\delta_{ij}$. Hence 
$\ell:=\sum_{j=1}^k y_j\ell_j$ is a polynomial of 
degree at most $(k-1)$ satisfying $\ell(x_i)=y_i$
for each $i\in\{1,\cdots,k\}$.
It is called the \emph{Lagrange interpolating 
polynomial} for the given set of $k$ points in 
$\reals^2$.} 
Now, if $C$ is any other element in $\group{PS}(V,g)$ 
satisfying $A=C^2$, we 
have $B=p(C^2)$, implying that $B$ and $C$ commute 
and hence that there exists an orthonormal basis 
diagonalising  both of them simultaneously. 
$A=B^2=C^2$ then implies that the squares of the 
eigenvalues and hence the eigenvalues themselves 
coincide (since they are positive). This shows 
$B=C$ and hence uniqueness. 
\end{proof}

Finally we mention the \emph{orthogonal group} of $\group{O}(V,g)$,
which is a subgroup of $\group{GL}(V)$ and defined by 
\begin{equation}
\label{eq:PD-7}
\begin{split}
\group{O}(V,g)\,:=\,
&\{A\in\group{GL}(V):g(Av,Aw)=g(v,w)\,\forall v,w\in V\}\\
=\,&\{A\in\group{GL}(V):A^{-1}=A^\dagger\}\,.
\end{split}
\end{equation}    

We can now state the main theorem underlying polar 
decomposition: 
\begin{theorem}[Existence and uniqueness of polar decomposition]
\label{thm:PolarDecomposition}
For any $A\in\group{GL}(V)$ and given Euclidean inner 
product $g$ there exists a unique $B\in\group{PS}(V,g)$ 
and a unique $R\in\group{O}(V,g)$, such that
\begin{equation}
\label{eq:ThmPolarDec}
A=BR\,.
\end{equation} 
\end{theorem} 
\begin{proof}
For $A\in\group{GL}(V)$ have 
$AA^\dagger\in\group{PS}(V,g)$. Set $B:=\sqrt{AA^\dagger}$ and $R:=B^{-1}A$. Then $R^\dagger R=
A^\dagger B^{-1}B^{-1}A=A^\dagger (AA^\dagger)^{-1}A=\id_V$. Hence $R^{-1}=R^\dagger$ and 
$R\in\group{O}(V,g)$, showing existence. To show 
uniqueness assume $(B_1,R_1)$ and $(B_2,R_2)$ both 
satisfy $A=B_1 R_1=B_2R_2$. 
Then $B_2^{-1}B_1=R_2R_1^{-1}\in\group{O}(V,g)$ (since 
$\group{O}(V,g)$ is a group). Hence, since $B^\dagger_i=B_i$, 
$(B_2^{-1}B_1)^\dagger=(B_2^{-1}B_1)^{-1}$ is equivalent to 
$B_1B_2^{-1}=B_1^{-1}B_2$ or to $B_1^2=B_2^2$. Hence $B_1$
and $B_2$ are both square roots of the same element in 
$\group{PS}(V,g)$.  Lemma\,\ref{thm:MatrixSquareRoot} now 
implies $B_1=B_2$ and hence also $R_1=R_2$.    
\end{proof}

\begin{remark}
On the right-hand side of \eqref{eq:ThmPolarDec} we 
have put $B$ to the left of $R$. We could have chosen 
the reversed order and proven a corresponding 
existence and uniqueness result. Then $A=BR=R'B'$, 
with uniquely determined $B,B'\in\group{PS}(V,g)$ 
and $R,R'\in\group{O}(V,g)$. But since 
$BR=R(R^{-1} BR)$ with 
$R^{-1} BR=R^\dagger BR\in\group{PS}(V,g)$, uniqueness 
shows $R'=R$ and $B'=R^{-1}BR$; that is, the orthogonal 
factor in the polar decomposition does indeed not 
depend on the convention concerning the order of the 
factors, whereas the positive symmetric part does depend 
on it and varies by conjugation with an orthogonal transformation. 
\end{remark}

\begin{remark}
The way in which we calculated $B$ and $R$ from $A$ 
shows that both are continuous functions of $A$ on the
domain $\group{GL}(V)$. The inverse map 
$(B,R)\rightarrow A=:BR$ is trivially also continuous. 
Hence we have a topological equivalence\footnote{In fact, 
this topological equivalence is a $C^\infty$ 
diffeomorphism in the natural differentiable structures 
that these manifolds carry.}
\begin{equation}
\label{eq:PolDecTop}
\group{GL}(V)\cong \group{PS}(V,g)\times\group{O}(V,g)\,.
\end{equation}  
As $\group{PS}(V,g)$ is contractible (being an open 
convex cone in a vector space) and $\group{O}(V,g)$ 
is compact, all global topological features of 
$\group{GL}(V)$ reside entirely in the latter.   
\end{remark}

\begin{remark}
If $G\subset\group{GL}(V)$ is a subgroup, polar 
decomposition of $A\in G$ will result in some 
$B$ and $R$ in $\group{GL}(V)$ which need not necessarily 
again be elements of the subset $G$. Whether or not 
that will be the case may depend on the chosen $g$. 
However, as we have seen in Section\,\ref{sec:PolDecNonNat} 
for the Lorentz group, polar decomposition with respect 
to the Euclidean metric $g=\eta+2\sigma\otimes\sigma$ 
(compare \eqref{eq:DefStateOfMotion-5}) will again 
result in factors lying within the Lorentz group. 
\end{remark}

\section{%
Parallel transport along geodesics on state space}
\label{sec:ParallelTransport}
We consider state space, i.e. the 3-dimensional 
Riemannian manifold $(\som,h)$ as defined in 
Section\,\ref{sec:PolDecNonNat}. We recall that 
its metric $h$ is just that induced from the 
flat Minkowski metric $\eta$ of the ambient $V$ 
into which $\sigma$ is embedded as a spacelike 
hypersurface, i.e. $h=\eta\vert_{T\som}$.  
In the paragraph below equation\,\eqref{eq:RelBoost-5}
we made a statement equivalent to the following:
\begin{proposition}
\label{thm:ParallelTransport}
Let $s_1$ and $s_2$ be two (non-coinciding) points in 
$\som$ and $B(s_1,s_2)$ the unique 
boost in the plane $\Span\{s_1,s_2\}$ mapping 
$s_1$ to $s_2$. Let further $\gamma:\reals\ni[\sigma_1,\sigma_2]\rightarrow\som$,
$\gamma(\sigma_i)=s_i$,  be the unique geodesic 
on $\som$ with respect to the Levi-Civita 
connection for $h$. Then parallel transport 
of any $Y_1\in T_{s_1}\som$ along $\gamma$
results in $Y_2=B(s_1,s_2)Y_1\in T_{s_2}\som$.
\end{proposition}
\begin{proof}
Rather than engaging in a direct calculation, we 
will here follow a more geometric reasoning which 
we partition into the following four steps.
\begin{enumerate}
\item
The timelike 2-plane $\Span\{s,s_1\}$ intersects 
$\som$ in a geodesic. This is because each 
point of this intersection is a fixed point of 
the isometry resulting from the reflection in $V$ 
at $\Span\{s,s_1\}$. But fixed-point sets of 
isometries are totally geodesic (meaning that 
any geodesic starting in and tangential to that
set remains within it). In particular, if the 
fixed-point set is a curve, it must itself be 
a geodesic.
\item
By its very definition, the Levi-Civita covariant 
derivative intrinsic to an isometrically embedded 
hypersurface in Euclidean or 
semi-Euclidean\footnote{In the semi-Euclidean 
case, the hypersurface is assumed to be nowhere 
lightlike.} space is given by the extrinsic 
covariant derivative (in flat embedding space) 
followed by orthogonal projection tangent 
to the hypersurface; see, e.g., the final (5th) 
edition of Weyl's classic text  
``Raum-Zeit-Materie'', which is again available  
\citep[\S\,12]{Weyl:RZM2025}.\footnote{This \S\,12 
has been added by Weyl in the transition from 
the 4th to the 5th (and last) edition. Hence it is 
not contained in the only existing english translation
\citep{Weyl:RZM4-Engl}, which is from the 4th edition.}
In our case this means that if $Y$ is tangent to 
$\som$, its covariant derivative along a curve 
$\gamma$ in $\som$ is given by $\nabla_{\dot\gamma}Y
=\proj_s^\perp(\partial_{\dot\gamma} Y)$,
where $\partial_{\dot\gamma}$ is the Levi-Civita 
covariant derivative in $(V,\eta)$ (i.e.
``ordinary'' flat derivative) and $\proj_s^\perp$
the projector in $V$ perpendicular to $s$ 
(compare \eqref{eq:DefProjEndo-b}). Note that 
since $\dot\gamma$ is tangent to $\som$, we do 
not need to know $Y$ outside $\som$ in order to 
compute $\partial_{\dot\gamma}Y$. It follows 
that $Y$ is parallely transported along $\gamma$ 
within $\som$ iff at each parameter value $\sigma$
the ambient (flat) covariant derivative  
$\partial_{\dot\gamma}Y$ in $V$ is proportional to 
the normal to $\som$ at that point, which is just 
$\gamma(\sigma)\in V$. This we write as
\begin{equation}
\label{eq:ParallelTransport}
\nabla_{\dot\gamma}Y=0\,\Leftrightarrow\,
\partial_{\dot\gamma}Y\propto\gamma\,.
\end{equation}
\item
This implies that if 
$[\sigma_1,\sigma_2]\ni\sigma\mapsto Y(\sigma)\in V$ 
is a one-parameter family of vectors in $V$ that 
obey the law \eqref{eq:ParallelTransport} of parallel 
transport, any two $Y(\sigma)$ and $Y(\sigma')$
for $\sigma,\sigma'\in[\sigma_1,\sigma_2]$ differ 
only by vectors in the plane $\Span\{\gamma(\sigma),\gamma(\sigma')\}=\Span\{s_1,s_2\}$. In particular, 
the component $Y_\perp(\sigma)\in V$ in the 
orthogonal complement of $\Span\{s_1,s_2\}$ is 
constant (independent of $\sigma$).  
\item
By definition, the boost $B(s,s_1)$ acts in the plane 
$\Span\{s,s_1\}$ and pointwise fixes its orthogonal 
complement in $V$. Hence, applied to any 
$Y(\sigma_1)\in T_{s_1}\som$, it keeps its component 
$Y_\perp$ orthogonal to $\Span\{s,s_1\}$ fixed and 
maps the component $Y_\Vert(s)$ within $\Span\{s,s_1\}$ 
in such a way so as to preserve its length and keeping 
it tangent to $\som$ and hence tangent to the 
intersection $\som\cap\Span\{s,s_1\}$, which is just 
the image of $\gamma$. But this is precisely what 
parallel propagation does, which completes the proof.
\end{enumerate}
\end{proof}

\section{Semi-direct products of groups}
\label{sec:SemiDirect}
Semi-direct products are special examples
of \emph{group extensions}. Let us therefore 
first explain this more general concept. 

Given two groups, called $N$ and $Q$, we 
will combine the two into a new group,
called  $G$, such that $G$ contains a 
unique normal subgroup $N'$ isomorphic 
to $N$ with quotient $G/N'$ isomorphic 
to $Q$. Formally this is expressed by 
arranging the triple $N$-$G$-$Q$ into a so-called 
\emph{short exact sequence}:\footnote{The 
``short'' refers to the fact that three is 
the smallest number of groups for which an
exact sequence makes a non-trivial statement. 
Two groups related by an exact sequence are 
merely isomorphic.}
\begin{equation}
\label{eq:GroupsShortExactSequence}
\begin{tikzcd}
\{e\}\arrow[r] 
& N\arrow[r,tail, "i"] 
& G\arrow[r,two heads, "\pi"]
& Q\arrow[r]
&\{e\}\,.	
\end{tikzcd}
\end{equation}
Here $\{e\}$ stands for the trivial group 
consisting only of the identity element $e$, 
and all arrows denote group homomorphisms.
That this sequence be \emph{exact} means that 
at each node $(N,G,Q)$ the image of the arriving 
map equals the kernel of the departing one. 
Clearly, the image of the first map 
from $\{e\}$ to $N$ must be the neutral element 
in $N$ so that exactness implies that the 
following map, $i$, is injective, i.e. an 
embedding of $N$ into $G$. This is indicated 
by the tailed arrow from $N$ to $G$.  Also, 
as the last map from $Q$ to $\{e\}$ has all 
of $Q$ in its kernel, exactness implies that 
the map $\pi$ is surjective. This 
we indicated by a double-headed 
arrow from $G$ to $Q$. 

In group-theorists' terminology, there are 
two ways to express the relation 
\eqref{eq:GroupsShortExactSequence} 
for $(N,G,Q)$: one either says that $G$ 
is an \emph{upward extension} of $Q$ by 
$N$ or, alternatively, that $G$ is a 
\emph{downward extension} of $N$ by 
$Q$; see, e.g., \citep[p.\,XX]{Atlas:FiniteGroups}.
Note that simply speaking of ``an extension''
of one group by another is ambiguous, as it 
does not tell which of the two is going to 
have a normal embedding in the group to be 
constructed. 

Now, a \emph{semi-direct product} is a special 
case of  \eqref{eq:GroupsShortExactSequence},
which is characterised by the existence of 
an injective homomorphism (an embedding)
\begin{equation}
\label{eq:GroupsSplittingHomomorphism}	
\sigma:Q\rightarrow G
\quad\text{such that}\quad
\pi\circ\sigma=\id_Q\,.
\end{equation}
In that case we write instead of \eqref{eq:GroupsShortExactSequence}
\begin{equation}
\label{eq:GroupsShortExactSequenceSplit}
\begin{tikzcd}
\{e\}\arrow[r] 
& N\arrow[r,tail, "i"] 
& G\arrow[r,bend left=20, two heads, "\pi"]
& Q\arrow[l,bend left=20, tail, "\sigma"]\arrow[r]
&\{e\}	
\end{tikzcd}
\end{equation}
and say that the sequence \emph{splits}. The map 
$\sigma$ \eqref{eq:GroupsSplittingHomomorphism}
is then called a \emph{splitting homomorphism}. 

\begin{definition}[first definition of semi-direct product]
\label{def:SDP-1}%
A group $G$ is called the \textbf{semi-direct product of 
groups $N$ and $Q$} if they can be arranged in a 
splitting short exact sequence \eqref{eq:GroupsShortExactSequenceSplit}. 
We write
\begin{equation}
\label{eq:Def_SDP-1}
G=N\rtimes_\sigma Q\,.
\end{equation}
\end{definition}

The images of $i$ and $\sigma$ define subgroups 
of $G$ which we denote by  
\begin{equation}
\label{eq:GroupsImages}
N':=i(N)\subset G\quad
\text{and}\quad
Q':=\sigma(Q)\subset G\,.
\end{equation}
Since $N'$ is the kernel of $\pi$ it is 
clearly normal in $G$. If $q\in Q$ is 
different from the neutral element then 
$q':=\sigma(q)\notin N'$ since $\pi(q')$
must be $q$ and hence $q'$ cannot be in 
the kernel of $\pi$. 
Hence 
\begin{equation}
\label{eq:GroupsIntersection}
N'\cap Q'=\{e\}\,,
\end{equation}
where $e\in G$ denotes the neutral element 
in $G$. Moreover, any $g\in G$ is the unique 
product of an element $n'\in N'$ and an element 
$q'\in Q'$. To see this, consider the 
homomorphism
\begin{equation}
\label{eq:SDP-Projection}
p:=\sigma\circ\pi: G\rightarrow Q'
\,,\quad g\mapsto q':=\sigma\bigl(\pi(g)\bigr). 
\end{equation}         
It satisfies $p\vert_{Q'}=\id_{Q'}$ and 
$p\circ p=p$ due to 
\eqref{eq:GroupsSplittingHomomorphism}; i.e. 
it is a projection homomorphism from $G$ onto 
the subgroup $Q'$. Now, for any $g\in G$, 
define $q'\in Q'$ as above and 
$n':=gq'^{-1}\in N'$ (which is indeed in 
$N'=\kernel(\pi)$ since $\pi(q')=\pi(g)$). 
This decomposition is unique, for if 
$(n_1',q_1')$ and $(n_2',q_2')$ both satisfy
$g=n'_1q'_1=n'_2q'_2$ it follows that 
${n'}_2^{-1}{n'}_1=q'_2{q'}_1^{-1}$. But the 
left-hand side is in $N'$ and the right-hand 
side in $Q'$, so that \eqref{eq:GroupsIntersection}
implies that both sides must equal $e$, hence  
$n'_1=n'_2$ and $q'_1=q'_2$.   
This discussion gives rise to two alternative 
definitions of semi-direct products: 

 \begin{definition}[second definition of semi-direct product]
\label{def:SDP-2}%
A group $G$ is called the \textbf{semi-direct product of 
its subgroups $N'$ and $Q'$} if the set 
$N'Q':=\{n'q':n'\in N'\,,q'\in Q'\}$ 
equals $G$, $N'\cap Q'=\{e\}$, and $N'$ is 
normal in $G$. 
\end{definition}

\begin{definition}[third definition of semi-direct product]
\label{def:SDP-3}%
A group $G$ is called the \textbf{semi-direct product of 
its subgroups $N'$ and $Q'$} if there is a  
projection homomorphism $p:G\rightarrow Q'$ with 
kernel $N'$ (projection meaning: 
$p\vert_{Q'}=\id_{Q'}$ and $p\circ p=p$). 
\end{definition}  

Composing $\sigma$ with the map 
$\mathrm{Ad}_{n'}:G\rightarrow G$, 
$g\mapsto \mathrm{Ad_{n'}(g)}:=n'gn'^{-1}$,
where $n'$ is some element from $N'$, 
clearly gives another splitting homomorphism 
$\sigma':=\mathrm{Ad}_{n'}\circ\sigma$ satisfying 
$\pi\circ\sigma'=\id_Q $ since $n'$ is in 
the kernel of $\pi$. Hence, in general, neither 
$\sigma$ nor $Q'$ are unique. For example, 
the group $\group{E}(3)$ of euclidean motions
is a semi-direct product of translations 
$N=\reals^3$ with rotations $Q=\group{SO}(3)$. 
Any splitting embedding 
$\sigma:\group{SO(3)}\rightarrow\group{E}(3)$ is 
characterised by the point $o\in\reals^3$
(the ``origin'') about which the elements 
of $\group{SO(3)}$ rotate, i.e. which is 
fixed under the action of all elements of 
$\group{SO(3)}$. And any two $\group{SO(3)}$ 
subgroups in $\group{E}(3)$ differ by a 
conjugation with the translation that maps the 
origin of the first into the origin of the 
second rotation group.  

Let us consider Definition\,\ref{def:SDP-2}. 
The multiplication of $g_1=n'_1q'_1$ with 
$g_2=n'_2q'_2$ is 
\begin{subequations}
\label{eq:SDP-Multiplication-1}
\begin{equation}
\label{eq:SDP-Multiplication-1a}
g_1g_2
=n'_1q'_1\,n'_2q'_2
=n'_1\,q'_1n'_2{q'}_1^{-1}\,q'_1q'_2
=:n'_3q'_3\,,
\end{equation}
where 
\begin{equation}
\label{eq:SDP-Multiplication-1b}
n'_3:= n'_1(q'_1n'_2{q'}_1^{-1})
=n'_1 \mathrm{Ad}_{q'_1}(n'_2)
\quad\text{and}\quad
q'_3=q'_1q'_2\,.
\end{equation} 
\end{subequations} 
Note that $n'_3\in N'$ since $N'$ is normal
and that, for any $q'\in Q'$, 
$\mathrm{Ad}_{q'}\vert_{N'}\in\group{Aut}(N')$.
In fact, it is obvious that the map 
$Q'\rightarrow\group{Aut}(N')$, 
$q'\mapsto\mathrm{Ad}_{q'}\vert_{N'}$ is a 
homomorphism. Identifying $N'$ with $N$ 
via the homomorphism $i:N\rightarrow G$
which is an isomorphism onto its image $N'$,
and likewise identifying $Q$ with $Q'$ 
via the homomorphism $\sigma: Q\rightarrow G$,
which is an isomorphism onto its image $Q'$,
we get yet another definition of a semi-direct
product that, like Definition\,\ref{def:SDP-1},
is in terms of $N$ and $Q$:
\begin{definition}[fourth definition of semi-direct product]
\label{def:SDP-4}%
Let $N$ and $Q$ be groups. Let further
\begin{equation}
\label{eq:DefSDP-4-1}
\alpha: Q\rightarrow\group{Aut}(N)\,,\quad
q\mapsto\alpha_q
\end{equation}
be a homomorphism so that $\alpha_q=\id_N$ if $q=e_Q$
(neutral element in $Q$) and 
$\alpha_{q_1}\circ\alpha_{q_2}=\alpha_{q_1q_2}$. 
Then the set $G:=N\times Q$ is made into a group
by defining the multiplication through
\begin{equation}
\label{eq:DefSDP-4-2}
(n_1,q_1)(n_2,q_2)
=\bigl(n_1\alpha_{q_1}(n_2)\,,\,q_1q_2\bigr)\,,
\end{equation} 
which is called the \textbf{semi-direct product of $N$ 
with $Q$ relative to $\alpha$} and denoted by 
\begin{equation}
\label{eq:DefSDP-4-3}
G=N\rtimes_\alpha Q\,.
\end{equation}
\end{definition}

We note the following more or less immediate 
consequences of this definition: 
\begin{enumerate}
\item 
If $e_N$ and $e_Q$ denote the neutral elements 
of $N$ and $Q$, respectively, the neutral element 
of $G$ is 
\begin{equation}
\label{eq:DefSDP-4-4}
e_G=(e_N\,,\,e_Q)
\end{equation}
and the inverse element of $(n,q)\in G$ 
is 
\begin{equation}
\label{eq:DefSDP-4-5}
\bigl(n,q\bigr)^{-1}
=\bigl(\alpha_{q^{-1}}(n)\,,\,q^{-1}\bigr)\,.
\end{equation} 
\item
It is immediate from \eqref{eq:DefSDP-4-2}
that the map 
\begin{equation}
\label{eq:DefSDP-4-6}
\pi: G\rightarrow Q\,,\quad
(n,q)\mapsto\pi(n,q):=q
\end{equation} 
is a surjective homomorphism whereas the projection 
onto the first factor $G\rightarrow N$, 
$(n,q)\mapsto n$ fails to be a homomorphism unless
$\alpha$ is trivial, that is, $\alpha_q=\id_N$ for all
$q\in Q$, in which case $G=N\times Q$ is a proper 
direct product of groups.  
\item
$N$ and $Q$ can be embedded into $G$ via the 
injective homomorphisms 
\begin{subequations}
\label{eq:DefSDP-4-7}
\begin{alignat}{3}
\label{eq:DefSDP-4-7a}
i&: N\rightarrow G\,,\quad 
&&n\mapsto i(n)
&&:=(n,e_Q)\,,\\
\label{eq:DefSDP-4-7b}
\sigma &: Q\rightarrow G\,, \quad
&&q\mapsto \sigma(n)
&&:=(e_N,q)\,,
\end{alignat} 
\end{subequations}
the images of which are  
\begin{subequations}
\label{eq:DefSDP-4-8}
\begin{alignat}{3}
\label{eq:DefSDP-4-8a}
i(N)&=:N'&&=\{(n,e_Q):n\in N\}&&\subset G\,,\\
\label{eq:DefSDP-4-8b}
\sigma(Q)&=:Q'&&=\{(e_N,q):q\in Q\}&&\subset G\,.
\end{alignat} 
\end{subequations}
Obviously $N'=\kernel(\pi)$ and $\pi\circ\sigma=\id_Q$.
Therefore, $G$ defined in \eqref{eq:DefSDP-4-3}, together 
with its subgroups defined in \eqref{eq:DefSDP-4-8} 
and maps defined in \eqref{eq:DefSDP-4-6} and \eqref{eq:DefSDP-4-7}, are related by a short exact 
sequence \eqref{eq:GroupsShortExactSequenceSplit},
leading us back to Definition\,\ref{def:SDP-1}.
 
\item
From \eqref{eq:DefSDP-4-2} have $(h,e_Q)(e_H,q)=(h,q)$
and hence
\begin{equation}
\label{eq:DefSDP-4-9}
H'\cap Q'=e_G
\quad\text{and}\quad
H'Q'=G\,,
\end{equation} 
leading us back to Definition\,\ref{def:SDP-2}.
Similarly for Definition\,\ref{def:SDP-3}, 
since the maps $\pi$ from \eqref{eq:DefSDP-4-6}
and $\sigma$ from \eqref{eq:DefSDP-4-7b} combine
to $p:=\sigma\circ\pi$, which is the required 
projection homomorphism $G\rightarrow Q'$.  
\item
Definition\,\ref{def:SDP-4} shows that any 
possible homomorphism \eqref{eq:DefSDP-4-1} 
from $Q$ into $\group{Aut}(N)$ can be realised 
in a semi-direct product, in which the 
automorphisms of $N'$ (which are usually outer) 
then appear as restrictions of inner automorphism 
of $G$ to its normal subgroup $N'$, as in 
\eqref{eq:SDP-Multiplication-1b}:
\begin{equation}
\label{eq:DefSDP-4-10}
\begin{split}
q'_1n'_2{q'}_1^{-1}
&=(e_N,q_1)(n_2,e_Q)(e_N,q_1)^{-1}\\
&=\bigl(\alpha_{q_1}(n_2),q_1\bigr)
\bigl(\alpha_{q_1^{-1}}(e_N),q_1^{-1}\bigr)\\
&=\bigl(\alpha_{q_1}(n_2),e_Q\bigr)\,,
\end{split}
\end{equation} 
using \eqref{eq:DefSDP-4-2}, \eqref{eq:DefSDP-4-5},
and that $\alpha_q(e_N)=e_N$ for all $q\in Q$.
\end{enumerate}
\begin{remark}
\label{rem:SemiDirectProd}
Note that Definitions\,\ref{def:SDP-1} 
and \ref{def:SDP-4} use groups $N$ and 
$Q$ to construct a new group $G$, whereas 
 Definitions\,\ref{def:SDP-2} 
and \ref{def:SDP-3} consider $G$ as given 
an characterise it in terms of subgroups
$N'\subset G$ and $Q'\subset G$. Therefore, 
Definitions\,\ref{def:SDP-1} 
and \ref{def:SDP-4} are often said to define 
an \emph{outer} and  Definitions\,\ref{def:SDP-2} 
and \ref{def:SDP-3} an \emph{inner} 
semi-direct product. 
\end{remark}

\section{Affine structures}
\label{sec:AffineStructures}
In this appendix we review the notion of 
affine spaces which underlies Minkowski 
spacetime in Special Relativity and also 
Galilei-Newton spacetime. One motivation 
to do so is to stress and make precise the 
somewhat subtle difference between affine 
spaces (which are homogeneous) and vector 
spaces (which are \emph{not} homogeneous).
Most likely, much of what is being said 
here will be known to the reader in one 
form or another, though perhaps it is useful 
to recall the essential structural properties 
in a way adapted to the language used in 
this paper.

In the main text we characterised Minkowski 
spacetime by the following 6-tuple 
$(M,V,+,\eta,o_V,o_T)$, where $(M,V,+)$ is a 
4-dimensional real affine space, $\eta\in V^*\otimes V^*$ 
is a symmetric, non-degenerate, bilinear 
form on $V$ of signature $(-,+,+,+)$, $o_V$ 
denotes an overall orientation of $V$, and,
finally, $o_T$ denotes a time-orientation of 
$(V,\eta)$. Similarly we characterised 
Galilei-Newton spacetime as a 6-tuple 
$(M,V,+,\tau,h,o_V)$, where $\tau\in V^*$ is 
an oriented time-distance function and 
$h\in[\kernel(\tau)]^*\otimes[\kernel(\tau)]^*$ 
is a Euclidean metric (symmetric positive definite
bilinear form) on $\kernel(\tau)$. Note that 
$\kernel(\tau)$ receives an orientation from 
$o_V$ and $\tau$ (the latter defines a 
time-orientation). 

Now, the purpose of this appendix is to recall the 
precise meaning of $(M,V,+)$ and also explain the 
notion of affine automorphisms and affine bases.  
There will be no need to restrict the dimension $n$ 
which we keep general. Also, we could have 
easily generalised the underlying number field 
to $\complex$, but for definiteness we stick with 
$\reals$. Hence we focus on the notion of a real 
affine space, the definition of which we now wish 
to carefully develop.  

\subsection{Affine spaces}
\label{sec:AffineSpaces}
Very roughly speaking, an affine space is like 
a vector space with slightly less structure. 
More precisely, it is the inhomogeneity caused 
by the distinction of a preferred vector, 
namely the null vector, that will be erased in 
the transition from a vector- to the associated 
affine space. The proper mathematical way to do 
this is to define an affine space by a set on 
which a vector space acts simply transitively. 
In order to appreciate the precise meaning of 
these words we will first introduce the notion 
of a ``group-action'' and mention some of the 
properties it may have. This will then allow 
us to give a lucid and concise definition of 
an affine space and appreciate its special 
features.

\begin{definition}[Groups actions and their properties]
\label{def:GroupAction}
Let $G$ be a group and $M$ a set. The set of 
bijections of $M$ onto itself will be  denoted 
by $\mathrm{Bij}(M)$ and is itself a group, 
with group multiplication being given by composition 
of maps and the neutral element being the 
identity map of $M$. An \textbf{action} of $G$ 
on $M$ is then simply a group homomorphism 
\begin{equation}
\label{eq:GroupAction}
\phi:G\rightarrow\mathrm{Bij}(M)\,,\quad
g\mapsto\phi_g\,.
\end{equation}
Recall that the homomorphism-property means 
that the map $\phi$ satisfies the two 
conditions ($e\in G$ being the neutral 
element)
\begin{equation}
\label{eq:DefAction}
\phi_e=\id_M\quad\text{and}\quad
\phi_g\circ\phi_{g'}=\phi_{gg'}\quad (\forall g,g'\in G)\,.
\end{equation}  
The action is called \textbf{effective} if $\phi$ 
is injective, i.e. each $g\ne e$ moves some $m$. 
Non effective actions are not really interesting, 
since a non effective action of $G$ may just 
be considered as an effective action of 
$G':=G/\kernel(\phi)$. An action is called 
\textbf{free} if each $g\ne e$ moves any $m$ 
($\phi_g$ has no fixed points). 
Clearly, being free implies being effective, 
but the converse is generally false, except
in very special cases, of which a relevant one 
will be given below. 
The set of points in $M$ reachable from 
a given point $m\in M$ by applying all $g\in G$ 
is called the \textbf{orbit of $G$ through $m$}:
\begin{equation}
\label{eq:DefOrbit}
\mathrm{Orb}_m(G):=\{\phi_g(m):g\in G\}\,.
\end{equation}  
The fact that the orbits result from a
group action implies that ``lying in the 
same orbit'' is an equivalence relation 
on $M$.  Hence $M$ is partitioned by (is 
the disjoint union of) orbits. The set of 
elements $g$ that fix a given $m\in M$ 
form a subgroup on $G$ called the 
\textbf{stabiliser subgroup of $G$ at $m$}:
\begin{equation}
\label{eq:DefStab}
\mathrm{Stab}_m(G):=\{g\in G:\phi_g(m)=m\}\,.
\end{equation}  
Property \eqref{eq:DefAction} implies that 
stabiliser subgroups of points in the same 
orbit are conjugate: 
\begin{equation}
\label{eq:StabConj}
\mathrm{Stab}_{\phi_g(m)}(G)=
g\,\mathrm{Stab}_m(G)\,g^{-1}\,.
\end{equation}  
The action is called \textbf{transitive} 
if $\mathrm{Orb}_m(G)=M$ for some (and hence 
all) $m$. Property \eqref{eq:StabConj} then implies 
that  all $\mathrm{Stab}_m(G)$ are conjugate. 
As for abelian groups conjugation is the 
identity, this says that for transitively acting 
abelian groups all stabiliser subgroups coincide. 
The action is called \textbf{simply transitive} if 
in addition to transitivity  
$\mathrm{Stab}_m(G)=\{e\}$. In this case 
any two points $m,m'\in M$ are connected by 
a unique $g\in G$. 
\end{definition}
\begin{proposition}
Simple transitivity is equivalent to 
transitivity and freeness. For abelian groups
this remains true if freeness is replaced 
by effectiveness. 
\end{proposition}
\begin{proof}
Simple transitivity trivially implies 
transitivity. Hence all stabiliser 
subgroups are conjugate. But in that case 
freeness is equivalent to each of them 
being equal to $\{e\}$, which is equivalent 
to simplicity. Now, if $G$ is abelian 
and acting transitively,  
effectiveness implies simplicity 
(the converse being trivial). Indeed, 
being effective means that for each  
$g\ne e$ there is some $m\in M$ such that 
$g$ moves $m$, i.e. $g\ne\mathrm{Stab}_m(G)$. 
Hence this is true for all $m$, since for 
transitive abelian groups all stabiliser 
subgroups coincide. Hence all 
$\mathrm{Stab}_m(G)$ are trivial.
\end{proof}  

Based on the foregoing discussion we now define
affine spaces as follows:   
\begin{definition}
\label{def:AffineSpace}
A \textbf{real affine space of dimension $n$} 
is a triple $(M,V,\phi)$, where $M$ is 
a set, $V$ is an $n$-dimensional real 
vector space, and $\phi$ is an effective and 
transitive (hence simply-transitive) 
action of $V$ on $M$. Here $V$ is 
considered as abelian group with group 
multiplication given by vector addition 
and the neutral element equal to $0$, so that 
\eqref{eq:GroupAction} and \eqref{eq:DefAction} 
now read
\begin{equation}
\label{eq:DefAffineSpace-1}	
\phi: V\rightarrow\mathrm{Bij}(M)\,,\quad
v\mapsto \phi_v
\end{equation}
and 
\begin{subequations}
\label{eq:DefAffineSpace-2}
\begin{alignat}{2}
\label{eq:DefAffineSpace-2a}	
&\phi_{v=0}&&\,=\,\mathrm{id}_V\,,\\
\label{eq:DefAffineSpace-2b}
&\phi_{v'+v}&&\,=\,\phi_{v'}\circ\phi_v\,.
\end{alignat}
\end{subequations}
\end{definition}

It is general practice, although this
may be confusing at first, to denote 
the action of $V$ on $M$ by the very 
same $(+)$-symbol as vector addition, i.e 
to write $m+v$ instead 
of $\phi_v(m)$, and hence to 
eliminate all explicit reference to 
$\phi$. For a given group action $\phi$,
this does not lead to ambiguities since 
whether a ``$+$'' means group action 
on $M$ or vector addition in $V$ is 
uniquely determined by whether the 
``$+$'' stands between an element of 
$M$ and an element of $V$ or between two 
elements of $V$, respectively. 
Equations \eqref{eq:DefAffineSpace-2} then 
assumes the simple form 
\begin{subequations}
\label{eq:DefAffineSpace-3}
\begin{alignat}{2}
\label{eq:DefAffineSpace-3a}	
&m+0=m\,,\\
\label{eq:DefAffineSpace-3b}
&(m+v)+v'=m+(v+v')\,.
\end{alignat}
\end{subequations}
Note that in \eqref{eq:DefAffineSpace-3b} both ``$+$'' on the left-hand side are 
group actions whereas the first ``$+$'' 
on the right-hand side is a group action 
and the second is vector addition. 

Since for any given two $m,m'\in M$ there 
exists a unique $v\in V$ so that 
$m'=m+v$, we may write $m'-m=v$.
Then, trivially, $m'=m+(m'-m)$. There are other ``obvious'' relations, like 
$(m'-p)+(p-m)=m'-m$ or $m+(m'-p)=m'+(m-p)$,
valid for all $m',m,p$ in $M$. Likewise, 
one has $(m'-m)=-(m-m')$, where here 
the two ``$-$'' on the right-hand side 
have two different meanings: In the 
bracket it denotes the difference 
operation in $M$, in front of the bracket 
scalar multiplication with 
$(-1)\in\reals$ in $V$. Other obvious 
notational simplifications apply, like 
$m-v:=m+(-v)$.  

As an alternative to Definition\,\ref{def:AffineSpace} 
above, affine spaces can be defined via the difference 
map just introduced:
\label{eq:DefAffineSpace-4}		
\begin{equation}
\Delta:M\times M\rightarrow V\,,\quad
(m',m)\mapsto\Delta(m',m)=m'-m\,.
\end{equation}
\begin{definition}[alternative to 
Definition\,\ref{def:AffineSpace}]
\label{def:AffineSpace-alt}
A \textbf{real affine space of dimension $n$} is a triple 
$(M,V,\Delta)$, where $M$ is a set, $V$ is an 
$n$-dimensional real vector space, and 
$\Delta:M\times M\rightarrow V$ is a map that 
satisfies the following two conditions
for any $o,m,m',m''$ in $V$:
\begin{subequations}
\label{eq:DefAffineSpace-5}
\begin{alignat}{1}
\label{eq:DefAffineSpace-5a}	
&\Delta_o:M\rightarrow V\,,\quad 
m\mapsto\Delta_o(m):=\Delta(m,o)\quad \text{is a bijection}\,,\\
\label{eq:DefAffineSpace-5b}
&\Delta(m'',m')+\Delta(m',m)
=\Delta(m'',m)\,.
\end{alignat}
\end{subequations}
\end{definition}
That the existence of $(M,V,\phi)$ implies 
$(M,V,\Delta)$ has been shown above. Conversely, 
given $(M,V,\Delta)$ with $\Delta$ satisfying the 
axioms above, we deduce a simply-transitive action 
$\phi$ of $V$ on $M$ by setting $\phi_v(m):=\Delta_m^{-1}(v)$. 
Indeed, from  \eqref{eq:DefAffineSpace-5b} we get 
for $m''=m'=m$ that $\Delta(m,m)=0$, 
which is equivalent to $\Delta^{-1}_m(0)=m$
for all $m$, which in turn implies  \eqref{eq:DefAffineSpace-2a}. 
The second condition  \eqref{eq:DefAffineSpace-2b} can be deduced as follows: \eqref{eq:DefAffineSpace-5b} 
is equivalent to 
$\Delta_{m'}(m'')+\Delta_m(m')=\Delta_m(m'')$, which in turn is equivalent to   
$\Delta_m^{-1}\bigl(\Delta_{m'}(m'')+\Delta_m(m')\bigr)=m''$ for all 
$m'',m',m$. Setting $\Delta_{m'}(m'')=:v'$
and $\Delta_m(m')=:v$ in order to replace 
in that equation $m''$ and $m'$ by 
$v'$ and $v$, this is equivalent to 
$\Delta_m^{-1}(v'+v)=\Delta^{-1}_{m'}(v')$. Finally, 
setting $m'=\Delta_m^{-1}(v)$ on the right-hand side, 
this is seen to be equivalent to 
$\phi_{v'+v}(m)=\phi_{v'}\bigl(\phi_v(m)\bigr)$ for 
all  $m,v,v'$, and hence to \eqref{eq:DefAffineSpace-2b}.

Summing up we can say that in an affine space we can add vectors 
to points and take differences of points according to the 
rules given above. However, points cannot be added.  To be 
sure, any point  $o\in M$ defines a bijection 
$\phi_o:M\rightarrow V$ via  $m\mapsto \phi_o(m):=(m-o)$. 
But the linear structure thereby pulled back to $M$, which 
is given by $m+m':=\phi^{-1}_o\bigl(\phi_o(m)+\phi_o(m')\bigr)=o+(m-o)+(m'-o)$, depends on the choice of $o$. In fact, 
through an appropriate choice of $o$ \emph{any} point 
$p$ of $M$ can be obtained as result of such an 
``addition'' of $m$ and $m'$: just choose 
$o=m+(m'-p)=m'+(m-p)$. 

\subsection{Affine maps and groups}
\label{sec:AffineMapsGroups}
\begin{definition}
\label{def:AffineMaps}
Let $A=(M,V,\phi)$ and 	$A'=(M',V',\phi')$ be two 
affine spaces. An \textbf{affine map} from $A$ to 
$A'$ consists of a pair $(F,f)$ of maps,
\begin{subequations}
\label{eq:AffineMapsDef}
\begin{equation}
\label{eq:AffineMapsDef-a}
F: M\rightarrow M'\,,\qquad
f: V\rightarrow V'\quad\text{linear}\,,	
\end{equation}
such that 
\begin{equation}
\label{eq:AffineMapsDef-b}
F\circ\phi=\phi'\circ (f\times F)\,.
\end{equation}
Here we explicitly displayed the action $\phi$ of 
$V$ on $M$ and likewise $\phi'$ of $V'$ on $M'$,
which we consider as maps 
$\phi: V\times M\rightarrow M$, $(v,m)\mapsto\phi_v(m)$,
and 
$\phi': V'\times M'\rightarrow M'$, 
$(v',m')\mapsto\phi'_{v'}(m')$,
respectively.
In our simplified notation, in which both actions
are written by a common $(+)$-sign, this reads
\begin{equation}
\label{eq:AffineMapsDef-c}
F(p+v)=F(p)+f(v)\,,
\end{equation}
\end{subequations}
for all $p\in M$ and all $v\in V$.
\end{definition}

Before we proceed we explicitly check that 
\eqref{eq:AffineMapsDef-c} makes sense, i.e. 
leads to the same result independent of how 
we represent a point $p\in M$ as ``sum'' of 
a point with a vector. So let 
$p=p_1+v_1=p_2+v_2$;
then $F(p_1+v_1)=F(p_1)+f(v_1)$ and $F(p_2+v_2)=F(p_2)+f(v_2)$. But $F(p_2)=F(p_1+(p_2-p_1))=F(p_1)+f(p_2-p_1)=F(p_1)+f(v_1-v_2)$
and linearity of $f$ shows that indeed both sides
are equal. 
Note that condition \eqref{eq:AffineMapsDef-c} 
says that once we know the map $f$ and the 
value $F(q)$ of the map $F$ for 
a single point $q$, we know the map $F$, i.e. 
$F(p)$ for any $p$, namely $F\bigl(p=q+(p-q)\bigr)=F(q)+f(p-q)$.
 
In view of the alternative definition of affine spaces 
in terms of the difference map $\Delta$, we could also 
have given a corresponding  alternative definition of 
an affine map:

\begin{definition}[alternative to 
Definition\,\ref{def:AffineMaps}]
\label{def:AffineMaps-alt}
Let $A=(M,V,\Delta)$ and 	$A'=(M',V',\Delta')$ 
be two affine spaces. An \textbf{affine map} from 
$A$ to $A'$ consists of a pair $(F,f)$ of maps,
\begin{subequations}
\label{eq:AffineMapsDef-Alt}
\begin{equation}
\label{eq:AffineMapsDef-Alt-a}
F: M\rightarrow M'\,,\qquad
f: V\rightarrow V'\quad\text{linear}\,,	
\end{equation}
such that 
\begin{equation}
\label{eq:AffineMapsDef_Alt-b}
\Delta'\circ (F\times F) = f\circ\Delta\,.
\end{equation}
In our simplified notation, in which the 
difference-map is written by a $(-)$-sign, 
this reads
\begin{equation}
\label{eq:AffineMapsDef-Alt-c}
\Delta'\bigl(F(p)-F(q)\bigr)=
f(p-q)\,,
\end{equation}
\end{subequations}
for all $q,p\in M$.
\end{definition}

\begin{definition}
\label{def:AffineAutomorphims-1}
An affine map between affine spaces is called 
an \textbf{affine isomorphism} iff the map $F:M\rightarrow M'$ is a bijection.
This is equivalent so the requirement for 
 $f:V\rightarrow V'$ to be a bijection (and hence
a vector-space isomorphism).
An affine isomorhism of an affine space to itself 
is called an \textbf{affine automorphism}. 
The set of affine automorphisms of an affine 
space $A$ forms a group under composition which 
is called the \textbf{general affine group}, 
denoted by $\group{GA}(A)$.   
\end{definition}

In order to understand the group-theoretic 
structure of $\group{GA}(A)$ we first need a 
label-set that faithfully labels each of 
its elements. This can be obtained in the 
following way: Choose a reference point 
$o\in M$ and use it to label $F$ by the pair 
$(v,f)$, where $v\in V$ is defined by $F(o)=o+v$, 
or $o:=F(o)-o$. The action of $F$ on a general 
point $p$ is then 
\begin{equation}
\label{eq:AffineAutomorphisms-2}	
F(p)=F\bigl(o+(p-o)\bigr)=o+v+f(p-o)\,.
\end{equation}
Note that the argument of $f$ is always 
the difference between the argument $p$ and 
the chosen base-point $o$, so that no $f$ 
moves $o$. In this way we identify $\group{GL}(V)$
with that subgroup of $\group{GA}(A)$ which 
stabilises $o$. We will see below how this 
identification behaves under changes of $o$. 

Suppose now that we have two affine automorphisms
$F_1$ and $F_2$, which we label by 
$(v_1,f_1)$ and $(v_2,f_2)$ as just 
explained, with reference to the same point 
$o$. The action of the composition 
$F_1\circ F_2$ on a general point can then 
be calculated:
\begin{equation}
\label{eq:AffineMapsComposition-1}	
F_1\circ F_2(p)
=
F_1\bigl(o+v_2+f_2(p-o)\bigr)
=o+v_1+f_1(v_2)+f_1\circ f_2(p-o)\,.
\end{equation}
Form that we infer the multiplication law 
for $\group{GA}(A)$ in the chosen 
parametrisation to be
\begin{equation}
\label{eq:AffineMapsComposition-2}	
(v_1,f_1)(v_2,f_2)
=\bigl(v_1+f_1(v_2)\,,\,f_1f_2\bigr)\,,
\end{equation}
where compositions of maps are now written by 
simple juxtapositions in order to stress that
it is group multiplication (in $\group{GA}(A)$
and $\group{GL}(V)$). We infer that 
$\group{GA}(A)$ is isomorphic to the semi-direct 
product of the abelian group $V$ with 
$\group{GL}(V)$:
\begin{equation}
\label{eq:AffineMapsSemiDirectProd-1}	
GA(A)\simeq V\rtimes\group{GL}(V)\,.
\end{equation}
The homomorphism 
$\alpha:\group{GL}(V)\rightarrow\group{Aut}(V)$
that we need according to Definition\,\ref{def:SDP-4} 
in order to define a semi-direct 
product is just the identity if we use the 
isomorphism $\group{Aut}(V)\simeq\group{GL}(V)$. 

However, it is important to keep in mind that 
this isomorphism depends on the chosen basepoint 
$o$. Had we chosen another one, say $o'=o+w$, 
then 
\begin{equation}
\label{eq:AffineMapsSemiDirectProd-2}
\begin{split}
F(p)
&=o+v+f(p-o)\\
&=o'+v+(o-o')+f\bigl(p-o'+(o'-o)\bigr)\\
&=o'+v+\Bigl[-w+f\bigl(p-o'+w\bigr)\Bigr]\\
&=o'+v+\bigl[T_{-w}\circ f\circ T_w\bigr](p-o')\,,
\end{split}
\end{equation}
where $T_w:V\rightarrow V$, $v\mapsto v+w$ 
denotes the translation-action of $V$ onto 
itself. This means that the same affine map 
$F$ that with respect to the basepoint $o$
is represented by the pair $(v,f)$ will be 
represented with respect to $o'=o+w$ by 
the pair $(v',f')$ where $v'=v$ and 
$f'=T_{-w}\circ f\circ T_w$. This is 
intuitively obvious, since the representation 
\eqref{eq:AffineMapsSemiDirectProd-1} of 
$\group{GA}(A)$ by the semi-direct product 
selects amongst all subgroups $\group{GL}(V)$
that one which fixes the selected base-point $o$. 
In changing $o$ to $o'$ we also change the subgroup 
in $\group{GA}(A)$ from the stabiliser subgroup 
of $o$ to that of $o'=o+w$. These two stabiliser 
subgroups are clearly related by conjugation with 
a translation.

An invariant characterisation of of the affine group 
$\group{GA}(A)$ in terms of $V$  and $\group{GL}(V)$ 
would be so say that $\group{GA}(A)$ is a 
``splitting downward extension of $\group{GL}(V)$ 
by $V$'' or, equivalently, a  ``splitting upward 
extension of $V$ by $\group{GL}(V)$)'' (compare
Appendix\,\ref{sec:SemiDirect}), which in any 
case means that we have a splitting short exact 
sequence with normal subgroup $V$ (translations) and 
quotient group $Q=\group{GL}(V)$:  

\begin{equation}
\label{eq:AffineAutomExactSequence-1}
\begin{tikzcd}
\{e\}\arrow[r] 
& V\arrow[r,tail, "i"] 
& \group{GA}(A)\arrow[r,bend left=15, two heads, "\pi"]
& \group{GL}(V)\arrow[l,bend left=15, tail, "\sigma_o"]\arrow[r]
&\{e\}\,.	
\end{tikzcd}
\end{equation}
The splitting homomorphism (embedding) 
$\sigma_o:\group{GL}(V)\rightarrow\group{GA}(A)$
depends on the chosen base-point $o\in M$ and is 
not natural (i.e. no one is distinguished). And two 
different ones, say $\sigma_{o}$ and $\sigma_{o'}$
with $o'-o=w$,  result in different image-subgroups
in $\group{GA}(A)$ which are conjugate by an 
element in the image of $i$ (the embedding of $V$
into $\group{GA}(A)$):
\begin{equation}
\label{eq:AffineAutomExactSequence-2}
\sigma_{o+w}\Bigl(\group{GL}(V)\Bigr)
=i_{w}\circ 
\sigma_{o}\Bigl(\group{GL}(V)\Bigr)
\circ i_{(-w)}\,.
\end{equation} 
This means that whereas it makes sense to 
speak of \emph{the} translations, since 
they form a normal subgroup, it does 
\emph{not} make sense to speak of \emph{the}
subgroup of homogeneous transformations:
$\group{GL}(V)$ is a \emph{quotient}- 
not a sub-group. It may be considered as
a subgroup, though not uniquely. There are 
as many \emph{different} subgroups isomorphic
to  $\group{GL}(V)$ in $\group{GA}(A)$
as there are points in $V$. None of them 
is intrinsically distinguished. 

\subsection{Affine bases and charts}
\label{sec:AffineBasesAndCharts}
\begin{definition}
\label{def:AffineBasis}
Let $A=(M,V,\phi)$ be n affine space.
An \textbf{affine basis} $B$ of $A$ is a tuple 
$B=(o,b)$, where $o\in M$ and $b:=\{e_1,\cdots,e_n\}$
is a basis for $V$. We note that the basis $b$ of $V$
uniquely determines a dual basis $b^*=\{\theta^1,\cdots,\theta^n\}$ of $V^*$, the dual vector space to $V$.
It satisfies $\theta^a(e_b)=\delta^a_b$. 
\end{definition}

\begin{remark}
\label{rem:AffineBasis}
Note that an affine basis determines $(n+1)$ points 
$\{o,o+e_1,\cdots,o+e_n\}$ which are not contained 
in any $m$-dimensional affine subspace with $m<n$. 
Hence we may equivalently characterise an affine 
basis of an $n$-dimensional affine space by $(n+1)$
points $\{p_0,p_1,\cdots,p_n\}\subset M$ which are
independent in the sense of not being contained 
in any lower-dimensional affine subspace.
It is then obviously true that the differences 
$e_a:=p_a-p_0$ form a basis for $V$, independent
of which point $p_0$ is selected from the set of 
$(n+1)$ points.
\end{remark}

\begin{definition}
\label{def:AffineChart}
An \textbf{affine chart} is a bijection 
$\phi_B:M\rightarrow\reals^n$ induced by 
an affine basis $B$ in the following way:
the value of the $a$-th component in 
$\reals^n$ of $\phi_B(p)$ is 
\begin{subequations}
\label{eq:DefAffineChart}	
\begin{equation}
\label{eq:DefAffineChart-a}	
\bigl[\phi_B(p)\bigr]^a:=\theta^a(p-o)\,.	
\end{equation}
Its inverse is (summation convention)
\begin{equation}
\label{eq:DefAffineChart-b}	
\phi^{-1}_B(x^1,\cdots, x^n)=o+x^ae_a\,.
\end{equation}
\end{subequations}
\end{definition} 

Suppose now that we have two affine bases, 
$B=(o,b)$ and $B'=(o',b')$, with 
$b:=\{e_1,\cdots e_n\}$ and $b':=\{e'_1,\cdots e'_n\}$,
and corresponding dual bases 
$\{\theta^1,\cdots \theta^n\}$ and
$\{\theta'^1,\cdots \theta'^n\}$, respectively. 
Then 
\begin{equation}
\label{eq:AffineBasesTransition}
o'=o+a\,,\quad
e'_b=L\indices{^a_b}e_a\,,\quad
\theta'^a=[L^{-1}]\indices{^a_b}\theta^b\,.	
\end{equation}
\begin{definition}
\label{def:TransitionFunction}	
The \textbf{transition function} between the 
affine charts $\phi_B$ and $\phi_{B'}$ is the 
bijection 
\begin{subequations}
\label{eq:DefTransitionFunction}		
\begin{equation}
\label{eq:DefTransitionFunction-a}	
\phi_{BB'}:=\phi_B\circ\phi^{-1}_{B'}:\reals^n\rightarrow\reals^n\,.
\end{equation}
Using \eqref{eq:DefAffineChart}
and \eqref{eq:AffineBasesTransition}, 
the $a$-th component in $\reals^n$ of the 
transition function is given by
\begin{equation}
\label{eq:DefTransitionFunction-b}	
\phi^a_{BB'}(x^1,\cdots,x^n)=
a^a+L\indices{^a_b}x^b\,,
\end{equation}
\end{subequations}
where $a^a:=\theta^a(a)$ is the $a$-th component 
of $a$ in the basis $\{e_1,\cdots,e_n\}$ of $V$.
\end{definition}

If we write the coordinates of the point $p$ 
in the chart $B'$ by $x'^a(p)$ and that in the 
chart $B$ by $x^a(p)$,
\eqref{eq:DefTransitionFunction-b} reads
\begin{equation}
\label{eq:DefTransitionFunction-2}	
x^a(p)=
a^a+L\indices{^a_b}x'^b(p)\,.
\end{equation}
To be distinguished from that relation between 
the coordinates of one and the same point $p$
in two different charts is the coordinate 
representation of an affine map in a single chart
$B$. Given an affine automorphism $(F,f)$ of 
$A=(M,V,\phi)$ and an affine basis 
$B=(o,\{e_1,\cdots,e_n\})$, such that 
$F(o)=o+a$ and $f(e_b)=L\indices{^a_b}e_a$, then 
the coordinate representation of the affine map is 
\begin{subequations}
\label{eq:AffineMapCoordRep}
\begin{equation}
\label{eq:AffineMapCoordRep-a}
F_B:=\phi_B\circ F\circ\phi_B^{-1}:
 \reals^n\rightarrow\reals^n\,.	
\end{equation}
Using \eqref{eq:DefAffineChart-b}, \eqref{eq:AffineMapsDef-c}, and \eqref{eq:DefAffineChart-a} this leads to 
\begin{equation}
\label{eq:AffineMapCoordRep-b}
F^a_B(x^1,\cdots,x^n)=a^a+L\indices{^a_b}x^b\,.
\end{equation}
\end{subequations}
This has the same analytic form as \eqref{eq:DefTransitionFunction-b}, but the 
meaning is clearly different. To state it once 
more: 
\begin{remark}
\label{rem:ActiveVersusPassive}	
Whereas \eqref{eq:AffineMapCoordRep-b}  
relates the different coordinates of two 
\emph{different points} in $M$ 
with respect to the \emph{same chart}, \eqref{eq:DefTransitionFunction-b} relates the 
different coordinates of the \emph{same point}
in two \emph{different charts}. The former is 
sometimes called a \emph{passive} and the latter 
an \emph{active} coordinate transformation.   
\end{remark}
  	
\begin{remark}
\label{rem:AfiineChartsDiffStructure}	
Finally we point out that the existence of preferred 
charts endows affine spaces with a differentiable 
and even analytic structure. A function 
$f:M\rightarrow\reals$ is called differentiable/analytic, 
if $f_B:=f\circ\phi_B^{-1}:\reals^n\rightarrow\reals$
is. This is independent of what chart $B$ we use, 
since obviously $f_{B'}=f_B\circ\phi_{BB'}$.
As $\phi_{BB'}$ is, according to \eqref{eq:DefTransitionFunction-b}, an 
affine-linear map, hence in particular analytic, 
we infer that $f_{B'}$ is smooth/analytic iff
$f_B$ is. 
\end{remark}
\end{appendices}

\newpage
\bibliographystyle{apsr}
\bibliography{VelAd}

@book{Atlas:FiniteGroups,
author    = {Conway, John H. and others},
title     = {{ATLAS} of Finite Groups},
publisher = {Oxford University Press},                 
address   = {Oxford},
year      = {1985}}

@article{Benz:2000,
author  = {Benz, Walter},
title   = {A Characterization of Relativistic Addition},
journal = {Abhandlungen aus dem Mathematischen 
            Seminar der Universit{\"a}t Hamburg},
volume  = {70},
year    = {2000},
pages   = {251-258},
doi     = {10.1007/BF02940916},
url     = {\url{http://doi.org/10.1007/BF02940916}}
}

@article{Benz:2002,
author  = {Benz, Walter},
title   = {A Common Characterization of Classsical and 
           Relativistic Addition},
journal = {Journal of Geometry},
volume  = {74},
year    = {2002},
pages   = {38-43},
doi     = {10.1007/PL00012536},
url     = {\url{http://doi.org/10.1007/PL00012536}}
}

@article{Celakoska:2008,
author  = {Celakoska, Emilija G.},
title   = {On isometry links between 
           4-vectors of velocity},
journal = {Novi Sad Journal of Mathematics},
year    = {2008}, 
volume  = {38},
number  = {3},
pages   = {165-172}}

@article{Celakoska.EtAl:2015,
author  = {Celakoska, Emilija and 
           Chakmakov, Dushan and 
           Petrushevski, Mirko},
title   = {On parameterization of 
           {Lorentz} boost links},
journal = {International Journal 
           of Contemporary
           Mathematical Sciences},
year    = {2015}, 
volume  = {10},
number  = {2},
pages   = {85-90},
doi     = {10.12988/ijcms.2015.513},
url     = {\url{http://doi.org/10.12988/ijcms.2015.513}}
}

@article{Dombrowski.Horneffer:1964,
author = {Dombrowski, Hainz Dieter 
          and Horneffer, Klaus},
title = {{Die Differentialgeometrie des Galileischen 
          Relativit{\"a}tsprinzips}},
journal = {Mathematische Zeitschrift},
year    = {1964},
volume = {86},
number = {4},
pages  = {291-311},
doi = {10.1007/BF01110404},
url     = {\url{http://doi.org/10.1007/BF01110404}}
}

@article{Einstein-SRT:1905,
author  = {Einstein, Albert},
title   = {{Zur Elektrodynamik bewegter K{\"o}rper}},
journal = {Annalen der Physik},
year    = {1905},
volume  = {323},
number  = {13},
pages   = {891-921},
doi     = {10.1002/andp.19053221004},
url     = {\url{http://doi.org/10.1002/andp.19053221004}}
}

@article{Friedrichs:1928,
author = {Friedrichs, Kurt},
title = {{Eine invariante Formulierung 
          des Newtonschen Gravitationsgesetzes 
          und des Grenz{\"u}berganges vom
          Einsteinschen zum Newtonschen Gesetz}},
    journal = {Mathematische Annalen},
    year = {1928},
volume = {98},
number = {1},
pages = {566-575},
doi = {10.1007/BF01451608},
url     = {\url{http://doi.org/10.1007/BF01451608}}
}

@article{Giulini:2002a,
author  = {Giulini, Domenico},
title   = {Uniqueness of simultaneity},
journal = {Britisch Journal for the Philosophy of Science},
year    = {2001}, 
volume  = {52},
number  = {4},
pages   = {651-670},
doi = {10.1093/bjps/52.4.651},
url     = {\url{http://doi.org/10.1093/bjps/52.4.651}}
}

@incollection{Giulini:2006b,
author  = {Giulini, Domenico},
title   = {Algebraic and geometric structures 
           in {Special Relativity}},
editor  = {L\"ammerzahl, Claus and Ehlers, J\"urgen},
booktitle = {Special Relativity: 
             Will it Survive the Next 101 Years?},
pages   = {45-111},
year    = {2006},
publisher = {Springer Verlag},
address = {Berlin},
series  = {Lecture Notes in Physics},
volume  = {702},
doi     = {10.1007/3-540-34523-X_4},
url     = {\url{http://doi.org/10.1007/3-540-34523-X_4}}
}

@book{Gourgoulhon:SR,
author    = {Gourgoulhon, {\'E}ric},
title     = {Special Relativity in General Frames},
series    = {Graduate Texts in Physics},
publisher = {Springer Verlag},
address   = {Berlin},
year      = {2013},
doi       = {10.1007/978-3-642-37276-6},
url     = {\url{http://doi.org/10.1007/978-3-642-37276-6}}
}

@book{Greub:LinearAlgebra,
author    = {Greub, Werner H.},
title     = {Linear Algebra},
series     = {Graduate Texts in Mathematics},
volume     = {23},
year       = {1975},
edition    = {fourth},
publisher = {Springer Verlag}, 
address   = {New York},
doi       = {10.1007/978-1-4684-9446-4},
url     = {\url{http://doi.org/10.1007/978-1-4684-9446-4}}
}

@book{Laue-SRT:1911,
author    = {Laue, Max von},
title     = {{Das Relativit{\"a}tsprinzip}},
edition    = {1st},
series     = {Die Wissenschaft},
volume     = {38},
publisher = {Verlag von Friedrich Vieweg \& Sohn}, 
address   = {Braunschweig},
year      = {1911}
}

@book{Jacobson:BasicAlgebraI,
author    = {Jacobson, Nathan},
title     = {Basic Algebra I},
edition   = {second},
publisher = {W.H. Freeman and Co.}, 
address   = {New York},
year      = 1985}

@incollection{Koczan:2023,
author = {Koczan, Grzegorz Marcin},
title  = {{Relative binary and ternary 
           4D velocities in the Special
           Relativity in terms of 
           manifestly covariant Lorentz
           transformation}},
booktitle = {Scientific Legacy of 
             Professor Zbigniew Oziewicz},
editor = {Colin Garci­a, Hilda Maria and
           Cruz Guzm{\'a}n, Jos{\'e} de Jes{\'u}s and
           Kauffmann, Luis H. and 
           Makaruk, Hanna},
series = {Series on Knots and Everything},
volume = {75},
pages = {169-206},
year = {2023},
publisher = {World Scientific},
address = {Singapore},
doi = {10.1142/9789811271151_0009},
url     = {\url{http://doi.org/10.1142/9789811271151_0009}}
}

@article{Kuenzle:1972,
author  = {K{\"u}nzle, Hans-Peter},
title   = {{Galilei and Lorentz structures on 
            space-time: comparison of the 
            corresponding geometry and physics}},
journal = {Annales de l'Institut Henri 
           Poincar{\'e} (Section A)},
year    = {1972}, 
volume  = {17},
number  = {4},
pages   = {337-362},
url     = {\url{http://www.numdam.org/item?id=AIHPA_1972__17_4_337_0}}}

@article{Matolcsi.Goher:2001,
author = {Matolcsi, Tam{\'a}s and 
          Goher, A.},
title = {{Spacetime without reference frames: 
          An application to the velocity 
          addition paradox}},
journal = {Studies in History and Philosophy 
           of Science Part B: Studies in
           History and Philosophy of 
           Modern Physics},
year   = {2001},
volume = {32},
number = {1},
pages  = {83-99},
doi    = {10.1016/S1355-2198(00)00037-X},
url     = {\url{http://doi.org/10.1016/S1355-2198(00)00037-X}}
}

@article{Macfarlane:1962,
author = {Macfarlane, Alan J.},
title = {{On the Restricted Lorentz Group and 
          Groups Homomorphically Related to It}},
journal = {Journal of Mathematical Physics},
year = {1962},
volume = {3},
number = {6},
pages = {1116-1129},
doi = {10.1063/1.1703854},
url     = {\url{http://doi.org/10.1063/1.1703854}}
}

@article{Matolcsi:2005,
author = {Matolcsi, Tam{\'a}s and 
          Matolcsi, M{\'a}t{\'e}},
title = {{Thomas rotation and Thomas 
         precession}},
journal = {International Journal of
           Theoretical Physics},
year   = {2005},
volume = {44},
number = {1},
pages  = {63-77},
doi    = {10.1007/s10773-005-1437-y},
url     = {\url{http://doi.org/10.1007/s10773-005-1437-y}}
}

@article{Mocanu:1986,
author= {Mocanu, Constantin I.},
title = {{Some difficulties within the framework of relativistic electrodynamics}},
journal = {Archiv f{\"u}r Elektrotechnik},
year = {1986},
volume = {69},
pages  = {97-110},
doi = {10.1007/BF01574845},
url    = {\url{http://doi.org/10.1007/BF01574845}}
}

@book{O-Neill:SRG,
author = {O'Neill, Barrett},
title = {Semi-Riemannian Geometry. 
         With Applications to Relativity},
series = {Pure and Applied Mathematics},
year = {1983},
publisher = {Academic Press},
address = {New York}
}

@misc{Oziewicz2006,
author = {Oziewicz, Zbigniew},
title  ={{The Lorentz boost-link is not 
         unique. Relative velocity as a  
         morphism in a connected groupoid
         category of null objects}}, 
year={2006},
eprint={math-ph/0608062},
archivePrefix={arXiv},
primaryClass={math-ph},
doi = {10.48550/arXiv.math-ph/0608062},
url    = {\url{http://doi.org/10.48550/arXiv.math-ph/0608062}}
}

@article{Oziewicz2007,
author = {Oziewicz, Zbigniew},
title = {Relativity groupoid instead of 
         relativity group},
journal = {International Journal of Geometric
           Methods in Modern Physics},
year = {2007},
volume = {04},
number = {05},
pages = {739-749},
doi = {10.1142/S0219887807002260},
url    = {\url{http://doi.org/10.1142/S0219887807002260}}
}

@misc{Oziewicz.Page2011,
author = {Oziewicz, Zbigniew and 
          Page, William S.},
title  = {Concepts of relative velocity},
year = {2011},
doi = {10.48550/arXiv.1104.0684},
url    = {\url{http://doi.org/10.48550/arXiv.1104.0684}}
}

@misc{Oziewicz2011,
author={Oziewicz, Zbigniew},
title = {{Ternary relative velocity; 
         astonishing conflict of the 
         Lorentz group with relativity}}, 
year={2011},
eprint={1104.0682},
archivePrefix={arXiv},
primaryClass={physics.gen-ph},
doi = {10.48550/arXiv.1104.0682},
url    = {\url{http://doi.org/10.48550/arXiv.1104.0682}}
}

@book{Robb:OptGeom,
author    = {Robb, Alfred A.},
title     = {Optical Geometry of Motion: A New View of the 
             Theory of Relativity},
publisher = {W. Heffer \& Sons Ltd.}, 
address   = {Cambridge},
year      = {1911}}

@book{Silberstein:1914,
author = {Silberstein, Ludwik},
title  = {The Theory of Relativity},
year   = {1914},
publisher = {MacMillan and Co. Limited},
address = {London},
note = {Based on lectures delivered at 
        University College London in     
        1912-13}}

@book{Streater.Wightman:AllThat,
author    = {Streater, Raymond F. and Wightman, Arthur Strong},
title     = {PCT, Spin and Statistics, and All That},
publisher = {W.A. Benjamin, Inc.}, 
address   = {New York},
year      = {1964}}

@article{Ungar:1988,
author  = {Ungar, Abraham Albert},
title   = {{Thomas} Rotation and the Parametrization of the {Lorentz} 
           Transformation Group},
journal = {Foundations of Physics Letters},
year    = 1988, 
volume  = 1,
number  = 1,
pages   = {57-89},
doi     = {10.1007/BF00661317},
url    = {\url{http://doi.org/10.1007/BF00661317}}
}

@article{Ungar:1989,
author = {Ungar, Abraham},
title = {{The relativistic velocity composition 
         paradox and the Thomas rotation}},
journal = {Foundations of Physics},
year    = {1989},
volume  = {19},
number  = {11},
pages   = {1385-1396},
doi     = {10.1007/BF00732759},
url    = {\url{http://doi.org/10.1007/BF00732759}}
}

@article{Ungar:1997,
author  = {Ungar, Abraham Albert},
title   = {{Thomas} precession: its underlying  
           gyrogroup axioms and their use in 
           hyperbolic geometry and relativistic 
           physics},
journal = {Foundations of Physics},
year    = {1997}, 
volume  = {27},
number  = {6},
pages   = {881-951},
doi     = {10.1007/BF02550347},
url    = {\url{http://doi.org/10.1007/BF02550347}}
}

@book{Ungar:2005,
author    = {Ungar, Abraham Albert},
title     = {Analytic Hyperbolic Geometry: Mathematical Foundations 
             and Applications},
publisher = {World Scientific}, 
address   = {Singapore},
year      = {2005},
doi       = {10.1142/5914},
url    = {\url{http://doi.org/10.1142/5914}}
}

@book{Ungar:BeyondEinstein,
author  = {Ungar, Abraham},
title   = {Beyond the Einstein Addition Law
           and its Gyroscopic Thomas Precession.
           The Theory of Gyrogroups and Gyrovector  
           Spaces},
year    = {2002},
publisher = {Kluwer Academic Publishers},
address ={New York},
isbn    = {0-792-36909-2}
}

@article{Urbantke:1990,
author  = {Urbantke, Helmuth},
title   = {Physical holonomy: {Thomas} precession, 
           and {Clifford} algebra},
journal = {American Journal of Physics},
year    = {1990}, 
volume  = {58},
number  = {8},
pages   = {747-750},
doi     = {10.1119/1.16401},
url    = {\url{http://doi.org/10.1119/1.16401}},
note    = {Erratum ibid. 59(12), 1991, 1150-1151}}

@article{Urbantke:2003,
author  = {Urbantke, Helmuth},
title   = {Lorentz Transformations from 
           Reflections: Some Applications},
journal = {Foundations of Physics Letters},
year    = {2003},
volume  = {16},
number  = {2},
pages   = {111-117},
doi     = {10.1023/A:1024162409610},
url    = {\url{http://doi.org/10.1023/A:1024162409610}}
}

@book{Weyl:RZM1,
author = {Weyl, Hermann},
title = {Raum--Zeit--Materie},
year  = {1918},
edition = {1st},
publisher = {Verlag von Julius Springer},
address = {Berlin}
}

@book{Weyl:RZM4-Engl,
author = {Weyl, Hermann},
title = {Space--Time--Matter},
year  = {1922},
publisher = {Methuen \& CO. LTD},
address = {London},
note    = {Translation from the 4th 
           German edition by Henry L. Brose}
}

@book{Weyl:RZM2025,
author    = {Giulini, Domenico and 
             Scholz, Erhard},
title     = {Hermann Weyl 
             Raum--Zeit--Materie},
series =    {Klassische Texte der
             Wissenschaft},
publisher = {Springer Verlag}, 
address   = {Berlin},
year      = {2025},
doi       = {10.1007/978-3-662-70400-4},
url    = {\url{http://doi.org/10.1007/978-3-662-70400-4}},
note      = {Republication of the 5th edition 
             with additional material from earlier 
             editions, as well as many annotations 
             by the editors}
}
\end{document}